\documentclass{article}
\usepackage{kr}

\usepackage{times}
\usepackage{soul}
\usepackage{url}
\usepackage[hidelinks]{hyperref}
\usepackage[utf8]{inputenc}
\usepackage[small]{caption}
\usepackage{graphicx}
\usepackage{amsmath}
\usepackage{amsthm}
\usepackage{booktabs}
\usepackage{algorithm}
\usepackage{algorithmic}
\usepackage{misc}
\usepackage{amssymb,amsmath}
\usepackage{stmaryrd}
\usepackage{xspace}
\usepackage{graphicx}
\usepackage{hyperref}
\usepackage{tikz}
\usetikzlibrary{trees,arrows,shapes,snakes,fit,shadows,decorations.text,decorations.pathmorphing,decorations.pathreplacing,calc,automata,positioning,patterns}
\usepackage{enumitem}

\tikzstyle{dot}=[circle,fill,inner sep=0pt,minimum size=3pt]

\newcommand\sig{\textup{sig}}
\newcommand\A{\mathcal{A}}

\renewcommand\epsilon{\varepsilon}

\newtheorem{example}{Example}
\newtheorem{theorem}{Theorem}
\newtheorem{lemma}{Lemma}
\newtheorem{definition}{Definition}
\newtheorem{remark}{Remark}

\hypersetup{draft}

\title{Interpolants and Explicit Definitions in Extensions of the 
Description Logic $\mathcal{EL}$}

\author{%
Marie Fortin\and
Boris Konev\and
Frank Wolter \\
\affiliations
University of Liverpool\\
\emails
\{mfortin, konev, wolter\}@liverpool.ac.uk
}

\begin{document}

\maketitle

\begin{abstract}
 We show that the vast majority of extensions of the description logic $\mathcal{EL}$
 do not enjoy the Craig interpolation nor the projective Beth definability 
 property. This is the case, for  example, for $\mathcal{EL}$ with nominals, 
 $\mathcal{EL}$ with the universal role, $\mathcal{EL}$ with a role inclusion of the form 
 $r\circ s\sqsubseteq s$, and for $\mathcal{ELI}$. It follows in particular that 
 the existence of an 
 explicit definition of a concept or individual name cannot be reduced to subsumption 
 checking via implicit definability.
 We show that nevertheless the existence of interpolants and explicit definitions
 can be decided in polynomial time for standard tractable extensions of $\mathcal{EL}$ 
 (such as $\mathcal{EL}^{++}$) and in \ExpTime for $\mathcal{ELI}$ and various extensions. 
 It follows that these existence problems are not harder than subsumption which is in sharp contrast 
 to the situation for expressive DLs.
 We also obtain tight bounds for the size of interpolants and explicit definitions and the complexity of computing them: single 
 exponential for tractable standard extensions of $\mathcal{EL}$ and double exponential
 for $\mathcal{ELI}$ and extensions. We close with a discussion of Horn-DLs such as Horn-$\mathcal{ALCI}$.
\end{abstract}

\section{Introduction}

The \emph{projective Beth definability property} (PBDP) of a description logic (DL) $\Lmc$ states that 
a concept or individual name is explicitly definable under an $\Lmc$-ontology $\Omc$ by an $\Lmc$-concept using 
symbols from a signature $\Sigma$ of concept, role, and individual names if, and only if, it is 
implicitly definable using $\Sigma$ under $\Omc$. The importance of the PBDP for DL research stems from the 
fact that it provides a polynomial time reduction of the problem to decide the existence of an explicit definition 
to the well understood problem of subsumption checking. The existence of explicit definitions is important for numerous 
knowledge engineering tasks and applications of description logic ontologies, for example,
the extraction of equivalent acyclic TBoxes from ontologies~\cite{DBLP:conf/kr/CateCMV06,TenEtAl13},
the computation of referring expressions (or definite descriptions) for individuals~\cite{KR21defdes},
the equivalent rewriting of ontology-mediated queries into concepts~\cite{DBLP:conf/ijcai/SeylanFB09,DBLP:journals/lmcs/LutzSW19,DBLP:conf/dlog/TomanW21}, 
the construction of alignments between ontologies~\cite{DBLP:conf/ekaw/GeletaPT16}, and the decomposition of 
ontologies~\cite{DBLP:conf/kr/KonevLPW10}. 
 
The PBDP is often investigated in tandem with the \emph{Craig interpolation property} (CIP) which states that if an $\Lmc$-concept is subsumed by another $\Lmc$-concept under some $\Lmc$-ontology then one finds an interpolating $\Lmc$-concept using the shared symbols of the two input concepts only. In fact, the CIP implies the PBDP and the interpolants obtained using the CIP can serve as explicit definitions.

Many standard Boolean DLs such as $\mathcal{ALC}$, $\mathcal{ALCI}$, and $\mathcal{ALCQI}$ enjoy the CIP and PBDP and sophisticated 
algorithms for computing interpolants and explicit definitions have been developed~\cite{TenEtAl13}. Important exceptions are the extensions of any of the above DLs with nominals and/or role hierarchies. In fact, it has recently been shown that the problem of deciding the existence of an interpolant/explicit definition becomes 2\ExpTime-complete for $\mathcal{ALCO}$ ($\mathcal{ALC}$ with nominals) and for $\mathcal{ALCH}$ ($\mathcal{ALC}$ with role hierarchies). This
result is in sharp contrast to the \ExpTime-completeness of the
same problem for $\mathcal{ALC}$ itself inherited from the \ExpTime-completeness of subsumption under $\mathcal{ALC}$-ontologies~\cite{AJMOW-AAAI21}.

Our aim in this article is threefold: (1) determine which members of the $\mathcal{EL}$-family of DLs enjoy the CIP/PBDP;
(2) investigate the complexity of deciding the existence of interpolants/explicit definitions for those that do not enjoy it;
and (3) establish tight bounds on the size of interpolants/explicit definitions and the complexity of computing them. 

In what follows we discuss our main results. It has been shown in~\cite{DBLP:conf/kr/KonevLPW10,DBLP:journals/lmcs/LutzSW19}
already that $\mathcal{EL}$ and $\mathcal{EL}$ with role hierarchies enjoy the CIP and PBDP. Rather surprisingly, 
it turns out that none of the remaining standard DLs in the $\mathcal{EL}$-family enjoy the CIP nor the PBDP. 
\begin{theorem}\label{posneg}
The following DLs do not enjoy the CIP nor PBDP:
\begin{enumerate}
\item $\mathcal{EL}$ with the universal role, 
\item $\mathcal{EL}$ with nominals, 
\item $\mathcal{EL}$ with a single role inclusion $r\circ s\sqsubseteq s$, 
\item $\mathcal{EL}$ with role hierarchies and a transitive role, 
\item the extension $\mathcal{ELI}$ of $\mathcal{EL}$ with inverse roles.
\end{enumerate}	
In Points~2 to 5, the CIP/PBDP also fails if the universal role
can occur in interpolants/explicit definitions.
\end{theorem}
Theorem~\ref{posneg} also has interesting consequences that are not explicitly stated. For instance, it follows that neither the DL 
$\mathcal{EL}^{++}$ introduced in~\cite{IJCAI05-short} nor 
the extension of $\mathcal{ELI}$ with any combination of nominals,
role hierarchies, or transitive roles enjoy the CIP/PBDP. 
With the exception of the failure of the CIP/PBDP for $\mathcal{EL}$ with nominals (without the universal role in interpolants/explicit definitions)~\cite{KR21defdes}, our results are new. 

It follows from Theorem~\ref{posneg} that the behaviour of extensions of $\mathcal{EL}$ is fundamentally different from extensions of $\mathcal{ALC}$:
adding role hierarchies to $\mathcal{ALC}$ does not preserve the CIP/PBDP~\cite{KonevLWW09} but it does for $\mathcal{EL}$;
on the other hand, adding the universal role or inverse roles to $\mathcal{ALC}$ preserves the CIP/PBDP~\cite{TenEtAl13} but it does not for $\mathcal{EL}$.

Theorem~\ref{posneg} leaves open the behaviour of a few natural DLs between $\mathcal{EL}$ and its extension with
arbitrary role inclusions. For instance, what happens if one only adds transitive roles or, more generally,
role inclusions using a single role name only? To cover these cases we show a general result that implies that these
DLs enjoy the CIP and PBDP. In particular, it follows that in Point~4 of Theorem~\ref{posneg} the combination of
role hierarchies with a transitive role is necessary for failure of the CIP/PBDP.

We next discuss our main result about tractable extensions of $\mathcal{EL}$.
\begin{theorem}\label{inttractable}
For $\mathcal{EL}$ and any extension with any combination of nominals, 
role inclusions, the universal role, or $\bot$, the existence of interpolants and explicit definitions is in \PTime. If an interpolant/explicit definition exists, 
then there exists one of at most exponential size that can be computed in exponential time. 
This bound is optimal.  
\end{theorem}
It follows that for tractable extensions of $\mathcal{EL}$ the complexity of deciding the existence
of interpolants and explicit definitions does not depend on the CIP/PBDP, in sharp contrast to the behaviour
of $\mathcal{ALCO}$ and $\mathcal{ALCH}$. Moreover, the proof shows how interpolants and explicit definitions
can be computed from the canonical models introduced in~\cite{IJCAI05-short}, if they exist. It applies
derivation trees (first introduced in \cite{DBLP:conf/ijcai/BienvenuLW13} for DLs without nominals and role hierarchies)
to estimate the size of interpolants and provide an exponential time algorithm for computing them. 
\begin{theorem}~\label{thm:elli}
For $\mathcal{ELI}$ and any extension with any combination of nominals,
the universal role, or $\bot$, the existence of interpolants and 
explicit definitions is \ExpTime-complete. If an interpolant/explicit definition exists, then there exists one of at most 
double exponential size that can be computed in double exponential time. This bound is optimal.
\end{theorem} 
The proof of Theorem~\ref{thm:elli} shows how an interpolant or explicit definition can be
extracted from a (potentially infinite) tree-shaped canonical model. The \ExpTime complexity bound is proved using
an encoding as an emptiness problem for tree automata that also uses derivation trees. It does not seem possible 
to obtain tight bounds on the size of interpolants using derivation trees; instead we generalize transfer sequences 
for this purpose (also first introduced in~\cite{DBLP:conf/ijcai/BienvenuLW13}).

In the final section, we consider expressive Horn-DLs such as Horn-$\mathcal{ALCI}$.
We first observe that Theorem~\ref{thm:elli} also holds for Horn-$\mathcal{ALCI}$ 
and extensions with nominals and the universal role, provided one asks for interpolants and explicit definitions in 
$\mathcal{ELI}$ (and extensions with nominals and the universal role, respectively). If one admits expressive Horn-concepts as interpolants or 
explicit definitions, then sometimes interpolants and explicit definitions exist that previously did not exist. We show that nevertheless the CIP/PBDP also fail in this case for DLs including Horn-$\mathcal{ALC}$, $\mathcal{ELI}$, and Horn-$\mathcal{ALCI}$. 

Detailed proofs are given in the arxiv version of this article.
\section{Related Work}
\label{sec:rw}
The CIP and PBDP have been investigated extensively in databases, with applications to query rewriting under views and query compilation~\cite{DBLP:series/synthesis/2011Toman,DBLP:series/synthesis/2016Benedikt}. The computation of explicit definitions under Horn ontologies can be seen as an instance of query reformulation under constraints~\cite{DBLP:journals/sigmod/DeutschPT06} which has been a major research topic for many years. The Chase and Backchase approach that is central to this research closely resembles our use of canonical models. We do not assume, however, that the chase terminates. In~\cite{DBLP:series/synthesis/2016Benedikt,DBLP:conf/ijcai/BenediktKMT17}, it is shown that the reformulation of CQs into CQs under tgds can be reduced to entailment using Lyndon interpolation of first-order logic. By linking reformulation into CQs and definability using concepts, this approach can potentially be used to obtain alternative proofs of complexity upper bounds for the existence of interpolants and explicit definitions in our languages. Also relevant is the investigation of interpolation in basic modal logic~\cite{GabMaks} and
hybrid modal logic~\cite{DBLP:journals/jsyml/ArecesBM01,DBLP:journals/jsyml/Cate05}.

The main aim of this article is to investigate explicit definability of concept and individual names under ontologies. We have therefore chosen a definition of the CIP and interpolants that generalizes the projective Beth definability property and
explicit definability in a natural and useful way, following~\cite{TenEtAl13}. There are,
however, other notions of Craig interpolation that are of interest. Of particular importance for modularity and various other purposes is the following version: if $\Omc$ is an ontology and $C \sqsubseteq D$ an inclusion such that $\Omc \models C \sqsubseteq D$, then there exists an ontology $\Omc'$ in the shared signature of $\Omc$ and 
$C\sqsubseteq D$ such that $\Omc\models \Omc'\models C \sqsubseteq D$. 
This property has been considered for $\mathcal{EL}$ and various extensions  
in~\cite{DBLP:journals/lmcs/Sofronie-Stokkermans08,DBLP:conf/kr/KonevLPW10}. Currently, it is unknown whether there exists any interesting relationship between this version of the CIP and the version we investigate in this article.

Craig interpolants should not be confused with uniform interpolants (or forgetting)~\cite{DBLP:conf/kr/LutzSW12,DBLP:conf/ijcai/LutzW11,DBLP:journals/ai/NikitinaR14,DBLP:conf/aaai/KoopmannS15}. Uniform interpolants generalize Craig interpolants in the sense that a uniform interpolant is an interpolant for a fixed antecedent and any formula implied by the antecedent and sharing with it a fixed set of symbols.

Interpolant and explicit definition existence have only recently
been investigated for logics that do not enjoy the CIP or PBDP. 
Extending work on Boolean DLs we discussed already,
it is shown that they become harder than validity also in the guarded and two-variable fragment \cite{JW-LICS2021}. 
The interpolant existence problem for linear temporal logic LTL is considered in~\cite{DBLP:journals/corr/PlaceZ14}. 
In the context of referring expressions, explicit definition existence is investigated in~\cite{KR21defdes}, see also~\cite{DBLP:conf/kr/BorgidaTW16}. 
\section{Preliminaries}
Let $\NC$, $\NR$, and $\NI$ be disjoint and countably infinite sets of
\emph{concept}, \emph{role}, and \emph{individual names}. A \emph{role} is a role name $r$
or an \emph{inverse role} $r^{-}$, with $r$ a role name. \emph{Nominals}
take the form $\{a\}$, where $a$ is an individual name. The \emph{universal role} is denoted by $u$. 
\emph{$\mathcal{ELIO}_{u}$-concepts} $C$ are defined by the following syntax rule:
$$
C,C' \quad ::= \quad \top \mid A \mid \{a\} \mid C \sqcap C' \mid  
\exists r . C 
$$
where $A$ ranges over concept names, $a$ over individual names, and $r$ over roles (including the universal role). 
Fragments of $\mathcal{ELIO}_{u}$ are defined as usual. For example, \emph{$\mathcal{ELI}$-concepts} are 
$\mathcal{ELIO}_{u}$-concepts without nominals and the universal role, and \emph{$\mathcal{EL}$-concepts} are 
$\mathcal{ELI}$-concepts without inverse roles. 
Given any of the DLs $\Lmc$ introduced above, an \emph{$\Lmc$-concept inclusion ($\Lmc$-CI)} takes the form 
$C \sqsubseteq D$ with $C,D$ $\Lmc$-concepts. An \emph{$\Lmc$-ontology}
$\Omc$ is a finite set of $\Lmc$-CIs.

We also consider ontologies with \emph{role inclusions (RIs)},
expressions of the form $r_{1} \circ \cdots \circ r_{n} \sqsubseteq r$ with $r_{1},\ldots,r_{n},r$ role names. An $\mathcal{ELO}_{u}$-ontology with RIs is called
an \emph{$\mathcal{ELRO}_{u}$-ontology}. A set of RIs is a \emph{role hierarchy} if all its RIs are of the form $r\sqsubseteq s$ with $r,s$ role names.

A \emph{signature} $\Sigma$ is a set of concept, role, and individual names,
uniformly referred to as \emph{(non-logical) symbols}. We follow common practice and
do not regard the universal role $u$ as a 
non-logical symbol as its interpretation is fixed. We use $\text{sig}(X)$ to
denote the set of symbols used in any syntactic object $X$ such as a
concept or an ontology. If $\Lmc$ is a DL and $\Sigma$ a signature, then 
an \emph{$\Lmc(\Sigma)$-concept} $C$ is an $\Lmc$-concept with $\text{sig}(C) \subseteq \Sigma$. 
The \emph{size} $||X||$ of a syntactic object $X$ is the number of symbols needed to write it down.

The semantics of DLs is given in terms of \emph{interpretations}
$\Imc=(\Delta^\Imc,\cdot^\Imc)$, where $\Delta^\Imc$ is a non-empty
set (the \emph{domain}) and $\cdot^\Imc$ is the \emph{interpretation
	function}, assigning to each $A\in \NC$ a set $A^\Imc \subseteq
\Delta^\Imc$, to each $r\in \NR$ a relation $r^\Imc \subseteq
\Delta^\Imc \times \Delta^\Imc$, and to each $a \in \NI$ an element
$a^\Imc \in \Delta^{\Imc}$.
The interpretation $C^{\Imc}\subseteq \Delta^{\Imc}$ of a concept $C$
in $\Imc$ is defined as usual, see \cite{DBLP:books/daglib/0041477}. An
interpretation $\Imc$ \emph{satisfies} a CI $C \sqsubseteq D$ if
$C^{\Imc} \subseteq D^{\Imc}$ and an RI $r_{1} \circ \cdots \circ r_{n} \sqsubseteq r$ if $r_{1}^{\Imc} \circ \cdots \circ r_{n}^{\Imc} \subseteq r^{\Imc}$. 
We say that $\Imc$ is
a \emph{model of} an ontology $\Omc$ if it satisfies all inclusions in it. If $\alpha$ is a CI or RI, we
write $\Omc \models \alpha$ if all models of \Omc satisfy $\alpha$.
We write $\Omc\models C \equiv D$ if $\Omc\models C \sqsubseteq D$ and $\Omc \models D \sqsubseteq C$.

An ontology is in \emph{normal form} if its CIs
are of the form 
$$
\top \sqsubseteq A, \quad A_{1}\sqcap A_{2} \sqsubseteq B, \quad A \sqsubseteq \{a\}, \quad \{a\} \sqsubseteq A, 
$$
and
$$
A \sqsubseteq \exists r.B, \quad \exists r.B \sqsubseteq A
$$
where $A,A_{1},A_{2},B$ are concept names, $r$ is a role or the universal role, and $a$ is an individual name. It is well known that for any $\mathcal{ELIO}_{u}$-ontology $\Omc$ with or without RIs 
one can construct in polynomial time a conservative extension
$\mathcal{O}'$ using the same constructors as $\Omc$ that is in normal form. 

$\Lmc(\Sigma)$-concepts can be characterized using \emph{$\Lmc(\Sigma)$-simulations} which we define next.
Let $\Imc$ and $\Jmc$ be interpretations. A relation $S\subseteq \Delta^{\Imc}\times \Delta^{\Jmc}$ is called an \emph{$\mathcal{ELO}(\Sigma)$-simulation} between $\Imc$ and $\Jmc$ if the following conditions hold:
\begin{enumerate}
	\item if $d\in A^{\Imc}$ and $(d,e)\in S$, then $e\in A^{\Jmc}$, for all $A\in \NC\cap \Sigma$;
	\item if $d=a^{\Imc}$ and $(d,e)\in S$, then $e=a^{\Jmc}$, for all $a\in \NI\cap \Sigma$;
	\item if $(d,d')\in r^{\Imc}$ and $(d,e)\in S$, then there exists $e'$ with $(e,e')\in r^{\Jmc}$ and $(d',e')\in S$, for all $r\in \NR \cap \Sigma$.
\end{enumerate}
$S$ is called an \emph{$\mathcal{ELO}_{u}(\Sigma)$-simulation} if $\Delta^{\Imc}$ is the domain of $S$ and an \emph{$\mathcal{ELIO}(\Sigma)$-simulation} if Condition~3 also holds for inverse roles from $\Sigma$. Condition~2 is dropped if $\Lmc$ does not use nominals. We write $(\Imc,d)\preceq_{\Lmc,\Sigma}(\Jmc,e)$ if there exists an $\Lmc(\Sigma)$-simulation $S$ between $\Imc$ and $\Jmc$ with $(d,e)\in S$. We write $(\Imc,d)\leq_{\Lmc,\Sigma}(\Jmc,e)$
if $d\in C^{\Imc}$ implies $e\in C^{\Jmc}$ for all $\Lmc(\Sigma)$-concepts $C$. The following characterization is well known~\cite{DBLP:journals/jsc/LutzW10,DBLP:conf/ijcai/LutzPW11}.
\begin{lemma}\label{lem:simulationel}
	Let $\Lmc\in \{\mathcal{EL},\mathcal{EL}_{u},\ELO,\ELOu,\mathcal{ELI},\mathcal{ELI}_{u}\}$. Then 
	$
	(\Imc,d)\preceq_{\Lmc,\Sigma}(\Jmc,e)$ implies $(\Imc,d)\leq_{\Lmc,\Sigma}(\Jmc,e)$.
	The converse direction holds if $\Jmc$ is finite.  
\end{lemma}

\section{Craig Interpolation Property and Projective Beth Definability Property}
\label{CIPandPBDP}
We introduce the Craig interpolation property (CIP) as defined in~\cite{TenEtAl13}
and the projective Beth definability property (PBDP) and prove Theorem~\ref{posneg} from the introduction to this article.
We observe that the CIP implies the PBDP, but lack a proof of the converse direction. Nevertheless, all 
DLs considered in this paper enjoying the PBDP also enjoy the CIP. 

Set $\text{sig}(\Omc,C)= \text{sig}(\Omc)\cup \text{sig}(C)$, for any ontology $\Omc$ and concept $C$. 
Let $\Omc_{1},\Omc_{2}$ be $\Lmc$-ontologies and let $C_{1},C_{2}$ be $\Lmc$-concepts. 
Then an $\Lmc$-concept $D$ is called an \emph{$\Lmc$-interpolant}\footnote{Important variations of this definition are to drop $\Omc_{2}$ in Point~2 and $\Omc_{1}$ in Point~3, respectively, or to consider only one ontology $\Omc=\Omc_{1}=\Omc_{2}$ and regard the signature $\Sigma$ of the interpolant as an input given independently from $\Omc,C_{1},C_{2}$. This has an effect on the CIP, but our results on interpolant computation and existence are not affected.} for $C_{1}\sqsubseteq C_{2}$ under
$\Omc_{1},\Omc_{2}$ if
\begin{itemize}
	\item $\text{sig}(D) \subseteq \text{sig}(\Omc_{1},C_{1})\cap \text{sig}(\Omc_{2},C_{2})$;
	\item $\Omc_{1}\cup \Omc_{2} \models C_{1} \sqsubseteq D$;
	\item $\Omc_{1}\cup \Omc_{2} \models D \sqsubseteq C_{2}$.
\end{itemize}	
\begin{definition}
	A DL $\Lmc$ has the Craig interpolation property (CIP) if for any
	$\Lmc$-ontologies $\Omc_{1},\Omc_{2}$ and $\Lmc$-concepts $C_{1},C_{2}$ such that $\Omc_{1}\cup \Omc_{2}\models C_{1}\sqsubseteq C_{2}$ there exists an $\Lmc$-interpolant for 
	$C_{1}\sqsubseteq C_{2}$ under $\Omc_{1},\Omc_{2}$.
\end{definition}
We next define the relevant definability notions.  Let $\Omc$ be an
ontology and $A$ a concept name.  Let $\Sigma\subseteq \text{sig}(\Omc)$
be a signature. An $\Lmc(\Sigma)$-concept $C$ is an \emph{explicit
	$\Lmc(\Sigma)$-definition of $A$ under $\Omc$} if $\Omc \models A
\equiv C$. We call $A$ \emph{explicitly definable in $\Lmc(\Sigma)$ under
	$\Omc$} if there is an explicit
$\Lmc(\Sigma)$-definition of $A$ under $\Omc$.
The $\Sigma$-reduct $\Imc_{|\Sigma}$ of an interpretation $\Imc$ coincides with $\Imc$ except that no symbol that is not in $\Sigma$ is interpreted in $\Imc_{|\Sigma}$. A concept $A$ is called 
\emph{implicitly definable using $\Sigma$ under
	$\Omc$} if the $\Sigma$-reduct of any model $\Imc$ of $\Omc$
determines the set $A^{\Imc}$; in other words, if $\Imc$ and $\Jmc$ are both models of $\Omc$ such that $\Imc_{|\Sigma} = \Jmc_{|\Sigma}$, then $A^{\Imc}=A^{\Jmc}$. It is easy to see that implicit definability can be reformulated as a standard reasoning problem as follows: a concept name $A\not\in\Sigma$ is implicitly definable using $\Sigma$ under
$\Omc$ iff $\Omc \cup \Omc_{\Sigma}
\models A \equiv A'$, where $\Omc_{\Sigma}$ is obtained from $\Omc$
by replacing every symbol $X$ not in $\Sigma$ (including $A$)
uniformly by a fresh symbol $X'$.
\begin{definition} A DL $\Lmc$ has the projective Beth definable property (PBDP) if for
	any $\Lmc$-ontology $\Omc$, concept name $A$, and signature
	$\Sigma\subseteq \text{sig}(\Omc)$ the following holds: if $A$ is
	implicitly definable using $\Sigma$ under $\Omc$, then $A$ is
	explicitly $\Lmc(\Sigma)$-definable under $\Omc$. 
\end{definition}
\begin{remark}\label{rem:CIPBDP}{\em
	The CIP implies the PBDP. To see this, assume that an $\Lmc$-ontology $\Omc$, concept name $A$ and a signature $\Sigma$ are given, and that $A$ is implicitly definable from $\Sigma$ under $\Omc$. Then $\Omc\cup \Omc_{\Sigma}\models A \equiv A'$, with $\Omc_{\Sigma}$ defined above. Take
	an $\Lmc$-interpolant $C$ for $A \sqsubseteq A'$ under $\Omc,\Omc_{\Sigma}$. 
	Then $C$ is an explicit $\Lmc(\Sigma)$-definition of $A$
	under $\Omc$. }
\end{remark}
\begin{remark}\label{rem:nom}{\em The PBDP implies that
implicitly definable nominals are explicitly definable and that, more generally, every implicitly definable concept $C$ is explicitly definable. This can be shown by adding $A\equiv C$ to the ontology for a fresh concept name $A$ and
asking for an explicit definition of $A$ in the extended ontology.
}
\end{remark}
\begin{remark}\label{rem:bot1}
{\em
	The CIP and PBDP are invariant under adding $\bot$ (interpreted as the empty set) to the languages introduced above. The straightforward proof is given in the appendix of the full version.
}
\end{remark}  
We next prove that the majority of tractable extensions of $\mathcal{EL}$ does not enjoy the CIP nor PBDP. 

\medskip

\noindent
{\bf Theorem~\ref{posneg}.}
{\em The following DLs do not enjoy the CIP nor PBDP:
\begin{enumerate}
	\item $\mathcal{EL}$ with the universal role, 
	\item $\mathcal{EL}$ with nominals, 
	\item $\mathcal{EL}$ with a single role inclusion $r\circ s\sqsubseteq s$, 
	\item $\mathcal{EL}$ with role hierarchies and a transitive role, 
	\item $\mathcal{EL}$ with inverse roles.
\end{enumerate}	
In Points~2 to 5, the CIP/PBDP also fails if the universal role
can occur in interpolants/explicit definitions.}
\begin{proof}
We first show that $\mathcal{EL}_u$ does not enjoy the PBDP. Point~1 then follows using Remark~\ref{rem:CIPBDP}.
We define an $\mathcal{EL}_u$-ontology $\Omc_{u}$, signature $\Sigma$, and concept
name $A$ such that $A$ is implicitly definable using $\Sigma$ under $\Omc_{u}$ but not
$\mathcal{EL}_{u}(\Sigma)$-explicitly definable under $\Omc_{u}$. 
	Define $\Omc_{u}$ as the following set of CIs:
$$
A \sqsubseteq B, \quad	D \sqcap \exists u.A \sqsubseteq E, \quad B \sqsubseteq \exists r. C
$$
$$
C \sqsubseteq D, \quad B \sqcap \exists r. (C \sqcap E) \sqsubseteq A,
$$
and let $\Sigma = \{B,D,E,r\}$. We have $\Omc_{u} \models A \equiv B \sqcap \forall r.
	(D \rightarrow E)$,\footnote{Here and in what follows we use standard $\mathcal{ALC}$ syntax and semantics and set $C\rightarrow D:=\neg C \sqcup D$~\cite{DBLP:books/daglib/0041477}.} so $A$ is implicitly definable using $\Sigma$ under $\Omc_{u}$.
	The interpretations $\I$ and $\I'$ given in Figure~\ref{figure:fig1} show that
	$A$ is not explicitly $\mathcal{EL}_{u}(\Sigma)$-definable under $\Omc_{u}$.
	\begin{figure}[th]
		\begin{center}
			\begin{tikzpicture}[every label/.style={font=\small,inner sep=0pt},node distance = 0.5cm and 0.5cm]
				\node[label=above:{${\color{blue}A},B$}] (a) {$a$};
				\node[below left=of a,label=below:{${\color{blue}C},D,E$}] (b) {$b$};
				
				\node[label=above:{$B$},right=3cm of a] (a') {$a'$};
				\node[below left=of a',label=below:{$D,E$}] (b') {$b'$};
				\node[below right=of a',label=below:{${\color{blue}C},D$}] (d') {$b''$};
				
				\foreach \i/\j in {a/b,a'/b',a'/d'} {
					\draw[->] (\i) -- (\j) ;
				}
				
			\end{tikzpicture}
		\end{center}
		\caption{Interpretations $\Imc$ (left) and $\Imc'$ (right) used for $\Omc_{u}$.}
		\label{figure:fig1}
	\end{figure}
	Indeed, $\I$ and $\I'$ are both models of $\Omc_{u}$, $a\in A^{\Imc}$, $a'\not\in
	A^{\Imc'}$, and the relation
	$
	\{(a,a'),(b,b')\} 
	$
	is a $\mathcal{EL}_u(\Sigma)$-simulation between $\Imc$ and $\Imc'$.
	As $\mathcal{EL}_u(\Sigma)$-concepts are preserved under
	$\mathcal{EL}_u(\Sigma)$-simulations (Lemma~\ref{lem:simulationel}), if $\Omc_{u}\models A\equiv F$ for some $\mathcal{EL}_{u}(\Sigma)$-concept
	$F$, then from $a\in A^{\Imc}$ we obtain $a\in F^{\Imc}$. This implies $a'\in F^{\Imc'}$, and so $a'\in A^{\Imc'}$.
        As $a'\not\in A^{\Imc'}$, we obtain a contradiction.

\medskip

We next prove Point~2. An example from~\cite{KR21defdes}
shows that $\mathcal{ELO}$ does not enjoy the CIP/PBDP. 
Here we show that $\mathcal{ELO}$ does not enjoy the CIP/PBDP, even
if interpolants/explicit defintions are from $\mathcal{ELO}_{u}$. Let $\Omc_{n}$ 
contain the following CIs:
	$$
	A \sqsubseteq \exists r.(E \sqcap \{c\}), \quad \top \sqsubseteq \exists s.(Q_{2} \sqcap \exists s.\{c\}) 
	$$
	$$
	\exists s.(Q_{1} \sqcap Q_{2} \sqcap \exists s.\{c\}) \sqsubseteq A,
	\quad \exists s.E \sqsubseteq Q_{1}
	$$
	and let $\Sigma=\{c,s,Q_{1}\}$. Observe that $A$ is implicitly definable  using $\Sigma$ under $\Omc_{n}$ as
	$
	\Omc_{n} \models A \equiv 
	\forall s.(\exists s.\{c\} \rightarrow Q_{1})$.
	The relation $\{(a,a'),(b,b'),(c,c')\}$ is an $\mathcal{ELO}_{u}(\Sigma)$-simulation between the interpretations $\I$ and $\I'$ defined in Figure~\ref{figure:fig4}. Now we can apply the same argument as in Point~1 to show that $A$ is not explicitly $\mathcal{ELO}_{u}(\Sigma)$-definable under $\Omc_{n}$.
 
	\begin{figure}[th]
		\begin{center}
			\begin{tikzpicture}[every label/.style={font=\small,inner sep=0pt},node distance = 1cm and 1.5cm]
				\node[label=above:{${\color{blue}A}$}] (a) {$a$};
				\node[below right=of a,label=right:{${\color{blue}E, A},\{c\}$}] (c) {$c$};
				\node[below=of a, left=of c, label=below:{${\color{blue}A,Q_2}, Q_1$}] (b) {$b$};
				
				\node[right=4.2cm of a] (a') {$a'$};
				\node[below right=of a',label=right:{$\{c\}$}] (c') {$c'$};
				\node[below=of a', left=of c', label=below:{$Q_1$}] (b') {$b'$};
				\node[above=of c',label=above:{${\color{blue}Q_2}$}] (f) {$b''$};
				
				\path[->] 
				(a) edge node[left]  {$s$} (b)
				(a) edge[blue] node[above right] {$r$} (c)
				(b) edge[<->] node[below] {$s$} (c)
				(b) edge[in=-220,out=-180,loop] node[left,above] {$s$} (b)
				(a') edge node[left] {$s$} (b')
				(c') edge[<->] node[below] {$s$} (b')
				(b') edge[in=-220,out=-180,loop] node[left,above] {$s$} (b')
				(f) edge[in=0,out=-40,loop] node[right,below] {$s$} (f)
				(b') edge node[above left] {$s$} (f)
				(a') edge node[above right] {$s$} (f)
				(f) edge[<->] node[right] {$s$} (c')
				;
				
			\end{tikzpicture}
		\end{center}
		\caption{Interpretations $\Imc$ (left) and $\Imc'$ (right) used for $\Omc_{n}$.}
		\label{figure:fig4}
	\end{figure}

For Point~3, let $\Omc_{r}$ contain 
	$$
	A \sqsubseteq \exists r.E, \quad E \sqsubseteq \exists s.B, \quad \exists s.B \sqsubseteq A, \quad
r \circ s \sqsubseteq s,
	$$
	and let $\Sigma= \{s,E\}$. Then $A$ is implicitly definable using $\Sigma$ under $\Omc_{r}$ since
	$$
	\Omc_{r} \models \forall x (A(x) \leftrightarrow \exists y (E(y) \wedge \forall z (s(y,z) \rightarrow s(x,z))).
	$$
	We show that there does not exist any $\mathcal{EL}_{u}(\Sigma)$-explicit
	definition of $A$ under $\Omc_{r}$. 
	\begin{figure}[th]
		\begin{center}
			\begin{tikzpicture}[every label/.style={font=\small,inner sep=0pt},node distance = 1cm and 1cm]
				\node[label=above:{${\color{blue}A}$}] (a) {$a$};
				\node[below right=of a,label=above right:{$E,{\color{blue}A}$}] (c) {$c$};
				\node[below=of a, left=of c, label=below:{${\color{blue}B}$}] (b) {$b$};
				
				\node[right=3.5cm of a] (a') {$a'$};
				\node[below right=of a',label=above right:{$E,{\color{blue}A}$}] (c') {$c'$};
				\node[below=of a', left=of c'] (b') {$b'$};
				\node[below right=of c',label=above right:{$E,{\color{blue}A}$}] (f) {$c''$};
				\node[below=of c', left=of f, label=below:{${\color{blue}B}$}] (g) {$b''$};
				
				\path[->] 
				(a) edge node[left] {$s$} (b)
				(a) edge[blue] node[above right] {$r$} (c)
				(c) edge node[below] {$s$} (b)
				(c) edge[blue,in=-30,out=-60, loop] node[below right] {$r$} (c)
				(a') edge node[left] {$s$} (b')
				(c') edge node[below] {$s$} (b')
				(c') edge node[left] {$s$} (g)
				(c') edge[blue] node[above right] {$r$} (f)
				(f) edge node[below] {$s$} (g)
				(f) edge[blue,in=-30,out=-60, loop] node[below right] {$r$} (f)
				;
				
			\end{tikzpicture}
		\end{center}
		\caption{Interpretations $\Imc$ (left) and $\Imc'$ (right) used for $\Omc_{r}$.}
		\label{figure:fig2}
	\end{figure}
	The interpretations $\I$ and $\I'$ given in Figure~\ref{figure:fig2}
	are both models of $\Omc_{r}$, $a\in A^{\Imc}$, $a'\not\in
	A^{\Imc'}$, and the relation
	$\{(a,a'), (b,b'), (c,c')\}
	$ is an $\mathcal{EL}_u(\Sigma)$-simulation between $\Imc$ and $\Imc'$. One can now show in the same way as in Point~1
        that no $\mathcal{EL}_u(\Sigma)$-definition of $A$ under $\Omc_{r}$ 
        exists.

Point~4 is shown in the appendix of the full version using a modification of the ontology used for Point~3.

To prove Point~5, obtain an $\mathcal{ELI}$-ontology $\Omc_{i}$ from $\Omc_{u}$ defined above by replacing the second CI of $\Omc_{u}$ by $D \sqcap \exists r^-.A \sqsubseteq E$. 
	Let, as before, $\Sigma= \{B,D,E,r\}$. Then $A$ is implicitly definable from $\Sigma$ under $\Omc_{i}$ 
        (the same explicit definition works), but $A$ is not explicitly 
        $\mathcal{ELI}_{u}(\Sigma)$-definable under $\Omc_{i}$ (the same interpretations $\I$ and $\I'$ work).
\end{proof}
We next discuss a general positive result on interpolation and explicit definition existence that shows that Theorem~\ref{posneg} is essentially optimal. A set $\mathcal{R}$ of RIs is \emph{safe} for a signature $\Sigma$ if
for each RI $r_1\circ\dots\circ r_n\sqsubseteq r\in\mathcal{R}$,
$n\geq 1$, if $\{r_1,\dots,r_n,r\}\cap \Sigma\neq \emptyset$
then $\{r_1,\dots,r_n,r\}\subseteq \Sigma$.
\begin{theorem}\label{th:safety}
	Let $\Omc_{1}, \Omc_{2}$ be $\mathcal{EL}$-ontologies with RIs,
	$C_{1},C_{2}$ be $\mathcal{EL}$-concepts, and set $\Sigma=\text{sig}(\Omc_{1},C_1) \cap \text{sig}(\Omc_{2},C_2)$.
	Assume that the set of RIs in $\Omc_1\cup\Omc_2$ is safe for $\Sigma$ and 
	$\Omc_1\cup\Omc_2\models C_1\sqsubseteq C_2$.
	Then an $\mathcal{EL}$-interpolant for $C_1\sqsubseteq C_2$ under
	$\Omc_1$, $\Omc_2$ exists.
\end{theorem}
The proof technique is based on simulations and similar to~\cite{DBLP:conf/kr/KonevLPW10,DBLP:journals/lmcs/LutzSW19}. Theorem~\ref{th:safety} has a few interesting consequences. For instance, $\mathcal{EL}$ with transitive roles enjoys both the CIP and PBDP since transitivity is expressed by the role inclusion $r\circ r\sqsubseteq r$ which is safe for any signature (as it only uses a single role name).
\section{Interpolant and Explicit Definition Existence}
\label{sec:exist}
We introduce interpolant and explicit definition existence as decision problems and 
establish a polynomial time reduction of the latter to the former. We then show that it suffices to consider 
ontologies in normal form and that the addition of $\bot$ does not affect the complexity of the decision problems.
\begin{definition}
	Let $\Lmc$ be a DL. Then \emph{$\Lmc$-interpolant existence} is the problem to decide for any
	$\Lmc$-ontologies $\Omc_{1},\Omc_{2}$ and $\Lmc$-concepts $C_{1},C_{2}$ whether there exists an $\Lmc$-interpolant for 
	$C_{1}\sqsubseteq C_{2}$ under $\Omc_{1},\Omc_{2}$. 
\end{definition}
Observe that interpolant existence reduces to checking $\Omc_{1}\cup \Omc_{2}\models C_{1}\sqsubseteq C_{2}$ for logics with the CIP but that this is not the case for logics without the CIP. 

\begin{definition}
	Let $\Lmc$ be a DL. Then \emph{$\Lmc$-explicit definition existence} is the problem to decide for any
	$\Lmc$-ontology $\Omc$, signature $\Sigma$, and concept name $A$ whether $A$ is explicitly definable in 
	$\Lmc(\Sigma)$ under $\Omc$.
\end{definition}
\begin{remark}\label{rem:red2}
	{\em There is a polynomial time reduction of $\Lmc$-explicit definition existence to $\Lmc$-interpolant existence.
	Moreover, any algorithm computing $\Lmc$-interpolants also computes $\Lmc$-explicit definitions and any bound
	on the size of $\Lmc$-interpolants provides a bound on the size of $\Lmc$-explicit definitions.
	The proof is similar to the proof of Remark~\ref{rem:CIPBDP}.
       } 
\end{remark}
We next observe that replacing the original ontologies by a 
conservative extension preserves interpolants and explicit definitions. 
Thus, it suffices to consider ontologies in normal form and 
interpolants for inclusions between concept names. 
\begin{lemma}\label{lem:normal} Let $\Omc_{1},\Omc_{2}$ be ontologies and $C_{1},C_{2}$ concepts in any DL $\Lmc$ considered in this paper. Then one can compute in polynomial time $\Lmc$-ontologies $\Omc_{1}',\Omc_{2}'$ in normal form and with fresh concept names $A,B$ such that an $\Lmc$-concept $C$ is an interpolant for
$C_{1}\sqsubseteq C_{2}$ under $\Omc_{1},\Omc_{2}$ iff it is an interpolant for $A \sqsubseteq B$ under $\Omc_{1}',\Omc_{2}'$.
\end{lemma}
\begin{proof}
	Let $\Omc_{1}'$ and $\Omc_{2}'$ be normal form conservative extensions of $\Omc_{1}\cup \{A \equiv C\}$ and, respectively, $\Omc_{2}\cup \{ B \equiv D\}$, computed in polynomial time. One can show that $\Omc_{1}'$ and $\Omc_{2}'$ are as required.
\end{proof}
\begin{remark}\label{rem:bot2}
{\em Assume that $\Lmc$ is any of the DLs introduced above and let $\Lmc_{\bot}$ denote
	its extension with $\bot$. Then $\Lmc$-interpolant existence and $\Lmc$-explicit definition existence can be reduced in polynomial time to  $\Lmc_{\bot}$-interpolant existence and $\Lmc_{\bot}$-explicit definition existence, respectively. The converse direction also holds modulo an oracle deciding whether $\Omc\models C \sqsubseteq \bot$.}
\end{remark}

\section{Interpolant and Explicit Definition Existence in Tractable $\EL$ Extensions}\label{sec:existence1}

The aim of this section is to analyse interpolants and explicit definitions for extensions of $\mathcal{EL}$ with any combination of nominals, role inclusions, or the universal role. We show the following result from the introduction.

\medskip
\noindent
{\bf Theorem~\ref{inttractable}.}
{\em
For $\mathcal{EL}$ and any extension with any combination of nominals, 
role inclusions, the universal role, or $\bot$, the existence of interpolants and explicit definitions is in \PTime. If an interpolant/explicit definition exists, 
then there exists one of at most exponential size that can be computed in exponential time. 
This bound is optimal.  
}

\medskip

Before we start with a sketch of the proof we give instructive examples
showing that the exponential bound on the size of explicit
definitions is optimal.
\begin{example}\label{exm:succ}
	{\em Variants of the following example have already been used for various succinctness arguments in DL.
	Let 
	\begin{eqnarray*}
		\Omc_{b}  & = & \{A \sqsubseteq M\sqcap \exists r_{1}.B_{1} \sqcap \exists r_{2}.B_{1}\}\cup\\
		&   & \{B_{i} \sqsubseteq \exists r_{1}.B_{i+1} \sqcap \exists r_{2}.B_{i+1}\mid 1\leq i<n\}\cup\\
		&   & \{B_{n} \sqsubseteq B,\exists r_{1}.B \sqcap \exists r_{2}.B\sqsubseteq B, B \sqcap M\sqsubseteq A\}
	\end{eqnarray*}
	and $\Sigma_{0}=\{r_{1},r_{2},B_{n},M\}$. $A$ triggers a marker $M$ and 
	a binary tree of depth $n$ whose leafs are decorated with $B_{n}$. Conversely, if $B_{n}$ is true at all 
	leafs of a binary tree of depth $n$, then $B$ is true at all nodes of the tree and $B$ together with $M$ entail $A$ at its root.
	Let, inductively, $C_{0}:=B_{n}$ and $C_{i+1}=\exists r_{1}.C_{i} \sqcap \exists r_{2}.C_{i}$, for $0<i<n$,
	and $C= M \sqcap C_{n}$. Then $C$ is the smallest explicit $\mathcal{EL}(\Sigma_{0})$-definition of 
	$A$ under $\Omc_{b}$. Next let 
	\begin{eqnarray*}
		\Omc_{p}	& = & \{r_{i}\circ r_{i} \sqsubseteq r_{i+1} \mid 0\leq i <n\}\cup\\
		&   & \{A \sqsubseteq \exists r_{0}.B, B \sqsubseteq \exists r_{0}.B, \exists r_{n}.B \sqsubseteq A\}
	\end{eqnarray*}
	and $\Sigma_{1}= \{r_{0},B\}$. Then $\exists r_{0}^{2^{n}}.B$ is the smallest explicit $\mathcal{EL}(\Sigma_{1})$-definition
	of $A$ under $\Omc_{p}$.}
\end{example}
Observe that using $\Omc_{b}$ one enforces explicit definitions of exponential size by generating a binary tree of linear depth whereas using $\Omc_{p}$ this is achieved by generating a path of exponential length. The latter can only happen if role inclusions 
are used in the ontology. One insight provided by the exponential upper bound on the size of explicit definitions in
Theorem~\ref{inttractable} is that the two examples cannot be combined to enforce a binary tree of exponential depth.

To continue with the proof we introduce ABoxes as a technical tool that allows us to move from interpretations to (potentially incomplete) sets of facts and concepts.
An \emph{ABox $\Amc$} is a (possibly infinite) set of assertions of the form
$A(x)$, $r(x,y)$, $\{a\}(x)$, and $\top(x)$ with $A\in \NC$, $r\in \NR$, $a\in \NI$, and $x,y$ individual variables (we call individuals used in ABoxes variables to distinguish them from individual names used in nominals).
We denote by $\text{ind}(\Amc)$ the set of individual variables in $\Amc$. 
A \emph{$\Sigma$-ABox} is an ABox using symbols from $\Sigma$ only. Models of ABoxes are defined as usual. We do not make the unique name assumption.
  
Every interpretation $\Imc$ defines an ABox $\Amc_{\Imc}$ by identifying every $d\in \Delta^{\Imc}$ 
with a variable $x_{d}$ and taking $A(x_{d})$ if $d\in A^{\Imc}$, $r(x_{c},x_{d})$ if $(c,d)\in r^{\Imc}$,
$\{a\}(x_{d})$ if $a^{\Imc}=d$. Conversely, ABoxes $\Amc$ define interpretations in the obvious way (by identifying variables $x,y$ if $\{a\}(x),\{a\}(y)\in \Amc$). We associate with every ABox $\Amc$ a directed graph $G_{\Amc}=(\text{ind}(\Amc),\bigcup_{r\in \NR} \{(x,y) \mid r(x,y)\in \Amc\})$. Let $\Gamma$ be a set of individual names.
Then $\Amc$ is \emph{ditree-shaped modulo $\Gamma$} if after dropping some facts of the form $r(x,y)$ with $\{a\}(y)\in \Amc$ for some $a\in \Gamma$, it is \emph{ditree-shaped} in the sense that $G_{\Amc}$ is acyclic and $r(x,y)\in \Amc$ and $s(x,y)\in \Amc$ imply $r=s$. A \emph{pointed ABox} is a pair $\Amc,x$ with $x\in \text{ind}(\Amc)$. Then $\ELO_{u}(\Sigma)$-concepts correspond to pointed $\Sigma$-ABoxes $\Amc,x$ such that $\Amc$ is ditree-shaped modulo $\NI \cap \Sigma$ and $\mathcal{ELO}(\Sigma)$-concepts correspond to \emph{rooted} pointed $\Sigma$-ABoxes $\Amc,x$ such that $\Amc$ is ditree-shaped modulo $\NI \cap \Sigma$, where $\Amc,x$ is called rooted if for every $y\in \text{ind}(\Amc)$ there is a path from $x$ to $y$ in $G_{\Amc}$.
We write $\Omc,\Amc\models C(x)$ if $x^{\Imc}\in C^{\Imc}$ for every model
$\Imc$ of $\Omc$ and $\Amc$.

Given an $\mathcal{ELRO}_{u}$-ontology $\Omc$ in normal form and a concept name $A$,
one can construct in polynomial time the \emph{canonical model} 
\emph{$\Imc_{\Omc,A}$ of $\Omc$ and $A$} using the approach introduced 
in~\cite{IJCAI05-short}.
More generally, the canonical model $\Imc_{\Omc,\Amc}$ for an ABox $\Amc$ 
and ontology $\Omc$ can be constructed in polynomial time and 
is a model of both $\Omc$ and $\Amc$ such that for any 
$\mathcal{ELO}_{u}$-concept $C$ using symbols from $\Omc$ only and any 
$x\in \text{ind}(\Amc)$,
\begin{itemize}
\item[($\dagger$)] $\Omc,\Amc\models C(x)$ iff $x\in C^{\Imc_{\Omc,\Amc}}$, 
\end{itemize}
details are given in the appendix of the full version. We let $\Imc_{\Omc,A}=\Imc_{\Omc,\Amc}$ with  $\Amc=\{A(\rho_{A})\}$. Note that in~\cite{IJCAI05-short}
the condition ($\dagger$) is only stated for subconcepts $C$ of the ontology 
$\Omc$, thus ($\dagger$) requires a proof.
\begin{example}\label{exm:can}
  {\em The interpretations $\Imc$ defined in the proof of Theorem~\ref{posneg} define
	canonical models $\Imc_{\Omc,A}$ with $\rho_{A}=a$ for the ontologies $\Omc\in \{\Omc_{u},\Omc_{n},\Omc_{r},\Omc_{i}\}$. 
   The interpretations $\Imc'$ define canonical models $\Imc_{\Omc,\Amc_{\Omc}^{\Sigma}}$ with $\Amc_{\Omc}^{\Sigma}$ the $\Sigma$-reduct of $\Imc_{\Omc,A}$ regarded as an ABox and $\rho_{A}=a'$.
}   	
\end{example}
The \emph{directed unfolding} of a pointed $\Sigma$-ABox $\Amc,x$ into a pointed $\Sigma$-ABox $\Amc^{u},x$ that is ditree-shaped modulo $\Sigma\cap \NI$ is defined in the standard way. In the \emph{rooted directed unfolding}, nodes that cannot be reached from $x$ via role names are dropped.

Assume now that $\Omc$ is in normal form and $A$ a concept name. Let $\Amc_{\Omc,A}^{\Sigma}$ be the $\Sigma$-reduct of the canonical model $\Imc_{\Omc,A}$, regarded as an ABox. Denote by  $\Amc_{\Omc,A}^{\Sigma,u},\rho_{A}$ the directed unfolding of $\Amc_{\Omc,A}^{\Sigma},\rho_{A}$, by $\Amc_{\Omc,A}^{\downarrow \Sigma},\rho_{A}$ the sub-ABox of $\Amc_{\Omc,A}^{\Sigma}$ rooted in $\rho_{A}$, and by $\Amc_{\Omc,A}^{\downarrow\Sigma,u},\rho_{A}$ its rooted directed unfolding. Theorem~\ref{inttractable} is a direct consequence of the following characterization of interpolants.
\begin{theorem}\label{thm:critele}
	 There exists a polynomial $p$ such that the following conditions are equivalent for all $\mathcal{ELRO}_{u}$-ontologies $\Omc_{1}, \Omc_{2}$ in normal form, concept names $A,B$, and $\Sigma=\text{sig}(\Omc_{1},A) \cap \text{sig}(\Omc_{2},B)$:
	\begin{enumerate}
		\item An $\mathcal{ELO}_{u}$-interpolant for $A \sqsubseteq B$ under $\Omc_{1},\Omc_{2}$ exists;
		\item $\Omc_{1}\cup \Omc_{2},\Amc_{\Omc_{1}\cup \Omc_{2},A}^{\Sigma}\models B(\rho_{A})$;
		\item there exists a finite subset $\Amc$ of $\Amc_{\Omc_{1}\cup \Omc_{2},A}^{\Sigma,u}$ with $|\text{ind}(\Amc)| \leq 2^{p(||\Omc_{1}\cup \Omc_{2}||)}$ such that the $\mathcal{ELO}_{u}$-concept corresponding to $\Amc,\rho_{A}$ 
		is an $\mathcal{ELO}_{u}$-interpolant for $A \sqsubseteq B$ under $\Omc_{1},\Omc_{2}$.
	\end{enumerate}
The same equivalences hold if in Points~1 to 3, $\mathcal{ELO}_{u}$ is 
replaced by $\mathcal{ELO}$, $\Amc_{\Omc_{1}\cup \Omc_{2},A}^{\Sigma}$ by $\Amc_{\Omc_{1}\cup \Omc_{2},A}^{\downarrow\Sigma}$,
and $\Amc_{\Omc_{1}\cup \Omc_{2},A}^{\Sigma,u}$ by $\Amc_{\Omc_{1}\cup \Omc_{2},A}^{\downarrow\Sigma,u}$.

In Point~3, $\Amc$ can be computed in exponential time, if it exists.
\end{theorem}
Note that the polynomial time decidability of interpolant existence follows from Point~2 of Theorem~\ref{thm:critele} (and the tractability of $\mathcal{ELRO}_{u}$~\cite{IJCAI05-short}).
\begin{example}
{\em Our proof of Theorem~\ref{inttractable} can be regarded as an application of 
Theorem~\ref{thm:critele}: by Example~\ref{exm:can}, the interpretations $\Imc$ and $\Imc'$ coincide with the canonical models $\Imc_{\Omc,A}$ and $\Imc_{\Omc,\Amc^{\Sigma}_{\Omc,A}}$ and so $\rho_{A}=a'\not\in A^{\Imc_{\Omc,\Amc^{\Sigma}_{\Omc,A}}}$ is equivalent to $\Omc,\Amc_{\Omc,A}^{\Sigma}\not\models A(\rho_{A})$ (Point~2 in Theorem~\ref{thm:critele}).
}
\end{example}
The following example illustrates the difference between the existence of explicit definitions in $\mathcal{ELO}$ and $\mathcal{ELO}_{u}$ and thus the need for moving to the ABoxes $\Amc_{\Omc,A}^{\downarrow\Sigma}$,
and $\Amc_{\Omc,A}^{\downarrow\Sigma,u}$ if one does not admit the universal 
role in explicit definitions.
\begin{example}
{\em Let $\Omc=\{A\sqsubseteq \{b\},A\sqsubseteq \exists r.B,B \sqsubseteq \exists s.A\}$ and let $\Sigma=\{b,B\}$. Then $A$ is explicitly $\mathcal{ELO}_{u}(\Sigma)$-definable under $\Omc$ since $\Omc\models A \equiv \{b\} \sqcap \exists u.B$ but $A$ is not explicitly $\mathcal{ELO}(\Sigma)$-definable.
Note that in this case $\Amc^{\Sigma}_{\Omc,A}=\{\{b\}(\rho_{A}),B(y)\}$ but
$\Amc^{\downarrow\Sigma}_{\Omc,A}=\{\{b\}(\rho_{A})\}$.
}
\end{example}
We next sketch the proof idea for Theorem~\ref{thm:critele} for the case with universal role in interpolants. 
We show ``1. $\Rightarrow$ 2.'', observe that ``3. $\Rightarrow$ 1.'' is trivial, and then sketch the proof of 
``2. $\Rightarrow$ 3.'' and the exponential time algorithm computing interpolants, details are provided
in the appendix of the full version. For ``1. $\Rightarrow$ 2.'' assume that $C$ is an $\mathcal{ELO}_{u}(\Sigma)$-concept with (i) $\Omc_{1}\cup \Omc_{2}\models A\sqsubseteq C$ and (ii) $\Omc_{1}\cup \Omc_{2}\models C\sqsubseteq B$. By ($\dagger$) and (i), $\Amc^{\Sigma}_{\Omc_{1} \cup \Omc_{2},A}\models C(\rho_{A})$. But then by (ii)  $\Omc_{1}\cup \Omc_{2},\Amc^{\Sigma}_{\Omc_{1} \cup \Omc_{2},A}\models B(\rho_{A})$, as required. 

If one does not impose a bound on the size of $\Amc$ in Point~3, then
one can prove ``2. $\Rightarrow$ 3.'' using compactness and a generalization of \emph{unraveling tolerance} according to
which $\Omc_{1}\cup \Omc_{2},\Amc_{\Omc_{1}\cup \Omc_{2},A}^{\Sigma}$ and $\Omc_{1}\cup \Omc_{2},\Amc_{\Omc_{1}\cup \Omc_{2},A}^{\Sigma,u}$ entail the same $C(\rho_{A})$~\cite{DBLP:journals/lmcs/LutzW17,DBLP:journals/tocl/HernichLPW20}. As we are interested in an exponential bound on the size of 
$\Amc$ (and a deterministic exponential time algorithm computing it) we require a more syntactic approach. 
Our proof of ``2. $\Rightarrow$ 3.'' is based on \emph{derivation trees} which represent a derivation of a fact $C(a)$ from an ontology $\Omc$ and ABox $\Amc$ using a labeled tree. 
Our derivation trees generalize those introduced in~\cite{DBLP:conf/ijcai/BienvenuLW13,DBLP:journals/jair/BaaderBL16} to languages with nominals and role inclusions. Reflecting the use of 
individual names and concept names in the construction of the domain of the 
canonical model~\cite{IJCAI05-short}, we assume 
$a\in \Delta:= \text{ind}(\Amc) \cup ((\NC\cup \NI) \cap \sig(\Omc))$ 
and $C\in \Theta :=\{\top\} \cup 
(\NC \cap \text{sig}(\Omc))\cup \{\{a\}\mid a\in \NI\cap \sig(\Omc)\}$. 
Then a derivation tree $(T,V)$ for 
$(a,C)\in \Delta\times \Theta$ is a tree $T$ with a labeling function 
$V:T\rightarrow \Delta\times \Theta$ such that $V(\epsilon)=(a,C)$ and $(V,T)$ satisfies
rules stating under which conditions the label of $n$ is derived in one 
step from the labels of the successors of $n$. To illustrate, 
the existence of successors $n_{1},n_{2}$ of $n$ with 
$V(n_{1})=(a,C_{1})$ and $V(n_{2})=(a,C_{2})$ justifies $V(n)=(a,C)$ if
$\Omc\models C_{1}\sqcap C_{2}\sqsubseteq C$. The rules are given in 
the appendix of the full version, we only discuss the rule 
used to capture derivations using RIs: $V(n)=(a_{1},C)$ is
justified if there are role names $r_{2},\ldots,r_{2k-2},r$ such
that $(a_{2k},C')$ is a label of a successor of $n$, 
$\Omc\models \exists r.C'\sqsubseteq C$, $\Omc\models r_{2}\circ \cdots \circ r_{2k-2}\sqsubseteq r$,
and the situation depicted in Figure~\ref{figure:dtree} holds,
where the ``dotted lines'' stand for `either $a_{i}=a_{i+1}$ or some $(a_{i},\{c\}),(a_{i+1},\{c\})$ with $c\in \NI$ are 
labels of successors of $n$', and $\hat{r}_i$ 
stands for `either $r(a_{i},a_{i+1})\in \Amc$ or some $(a_{i},C_{i})$ is a label of a successor of $n$ 
and $\Omc \models C_{i}\sqsubseteq \exists r_{i}.\{a_{i+1}\}$ if 
$a_{i+1}\in \NI$ and $\Omc \models C_{i}\sqsubseteq \exists r_{i}.a_{i+1}$ if $a_{i+1}\in \NC$'. Moreover, for all
$a_{i}\not=a_{1}$, $1\leq i \leq 2k$, there exists a successor of $n$ with label $(a_{i},D)$ for some $D$.
The soundness of this rule should be clear, completeness can be shown similarly to the analysis of canonical models.

	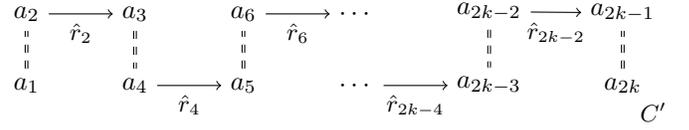
\begin{figure}[t]
	\hspace*{-1ex}\begin{tikzpicture}[every label/.style={font=\small,inner sep=0pt},node distance = 0.5cm and 0.85cm]
		\node (a1) {$a_1$};
		\node[above=of a1,label=below:{}] (a2) {$a_2$};
		
		\node[right=of a1,label=below:{}] (a4) {$a_4$};
		\node[above=of a4,label=below:{}] (a3) {$a_3$};
		
		\node[right=of a4,label=below:{}] (a5) {$a_5$};
		\node[above=of a5,label=below:{}] (a6) {$a_6$};
		
		\node[right=of a6,label=below:{}] (a7) {$\cdots$};
		\node[right=of a5,label=below:{}] (a2k-4) {$\cdots$};
		
		\node[right=of a2k-4,label=below:{}] (a2k-3) {$a_{2k-3}$};
		\node[above=of a2k-3,label=below:{}] (a2k-2) {$a_{2k-2}$};
		
		\node[right=of a2k-3,label=below right:{$C'$}] (a2k) {$a_{2k}$};
		\node[above=of a2k,label=below:{}] (a2k-1) {$a_{2k-1}$};

		\foreach \i/\j in {a1/a2,a3/a4,a5/a6,a2k-2/a2k-3,a2k/a2k-1} {
			\draw[=,style=double,style=dashed] (\i) -- (\j) ;
		}
		
		\foreach \i/\j in {2/3,4/5,6/7,2k-4/2k-3,2k-2/2k-1} {
			\draw[->] (a\i) -- (a\j) node[midway, label=below:{$\hat{r}_{\i}$}]{} ;
		}
		
	\end{tikzpicture}
	\caption{Rule for Role Inclusions.
		\label{figure:dtree}}
\end{figure}

The length of the sequence $a_{1},\ldots,a_{2k}$ can be exponential (for instance, in Example~\ref{exm:succ} for the
fact $(\rho_{A},A)$ in $\Omc_{p},\Amc_{\Omc_{p},A}^{\Sigma_{1}}$). One can show, however, that its length 
can be bounded without affecting completeness by $2^{q(||\Omc||+||\Amc||)}$ with $q$ a polynomial. 
The following lemma summarizes the main properties of derivation trees.
\begin{lemma}\label{lem:deriv0}
Let $\Omc$ be an $\mathcal{ELRO}_{u}$-ontology in normal form and 
$\Amc$ a finite $\text{sig}(\Omc)$-ABox. Then 
\begin{enumerate}
	\item $\Omc,\Amc \models A(x)$ if and only if there is a derivation tree for $A(x)$ in $\Omc,\Amc$. Moreover, if a derivation tree exists, then there exists one of depth and outdegree bounded by $(||\Amc||+||\Omc||)\times ||\Omc||$ 
which can be constructed in exponential time in $||\Omc||+||\Amc||$.
	\item If $(T,V)$ is a derivation tree for $A(x)$ in $\Omc,\Amc$ of at most exponential size,
	then one can construct in exponential time (in $||\Amc||+||\Omc||$) a derivation tree $(T',V')$ 
	for $A(x)$ in $\Omc, \Amc^{u}$ with $\Amc^{u}$ the directed unfolding of $\Amc$ modulo 
        $\Sigma=\text{sig}(\Amc)\cap \NI$ and
        $T'$ of the same depth as $T$ and such that the outdegree of $T'$ does not exceed $\max{\{3,3n\}}$ with $n$ the
	length of the longest chain $a_{1}\cdots a_{n}$ used in the rule for RIs in the derivation tree $(T,V)$.
\end{enumerate}
\end{lemma}
\begin{proof}
We sketch the idea. For Point~1, the bound on the depth of derivation trees can be proved by 
observing that one can assume (using a standard pumping argument) 
that the labels of distinct nodes on a single path are distinct and the 
bound on the outdegree can be proved by observing that one can trivially 
assume that all successor nodes of a node have distinct labels.
For the construction of derivation trees,
let $F_{n}$ denote the set of facts in $\Delta\times\Theta$
for which there is a derivation tree of depth at most $n$. 
Then one can construct in exponential time derivation trees for all 
facts in any $F_{n}$, $n\leq (||\Amc||+||\Omc||)\times ||\Omc||$ by 
starting with derivation trees of depth $0$ for members of $F_{0}$, 
and then constructing derivation trees of depth~$i+1$ for members of $F_{i+1}$ using the trees for members of 
$F_{0},\ldots,F_{i}$. For Point~2, the transformation of $(T,V)$ into $(T',V')$ is by induction over rule application, 
the only interesting step being the rule for RIs. Using the ontology $\Omc_{p}$ of Example~\ref{exm:succ} one can see that the 
exponential blow-up of the outdegree is unavoidable.
\end{proof}
We are now in the position to complete the sketch of the proof of ``2. $\Rightarrow$ 3.'' Assume that
Point~2 holds. Then $\Omc_{1}\cup \Omc_{2},\Amc_{\Omc_{1}\cup \Omc_{2},A}^{\Sigma}\models B(\rho_{A})$.
By Point~1 of Lemma~\ref{lem:deriv0} we can construct a derivation tree $(T,V)$ for $(\rho_{A},B)$ 
in $\Omc_{1}\cup \Omc_{2},\Amc_{\Omc_{1}\cup \Omc_{2},A}^{\Sigma}$ of polynomial depth and outdegree in exponential time.
By Point~2 of Lemma~\ref{lem:deriv0} we can transform $(T,V)$ into a derivation tree $(T',V')$ for
$(\rho_{A},B)$ in $\Omc_{1}\cup \Omc_{2},\Amc_{\Omc_{1}\cup \Omc_{2},A}^{\Sigma,u}$ in exponential time. Now let
$\Amc$ be the restriction of $\Amc_{\Omc_{1}\cup \Omc_{2},A}^{\Sigma,u}$ to all $x\in \text{ind}(\Amc_{\Omc_{1}\cup \Omc_{2},A}^{\Sigma,u})$ which occur in a label of $V'$. Then $(T',V')$ is also a derivation tree for $(\rho_{A},B)$ 
in $\Omc_{1}\cup \Omc_{2},\Amc$ and so $\Omc_{1}\cup \Omc_{2},\Amc\models B(\rho_{A})$. It follows that the
$\mathcal{ELO}_{u}(\Sigma)$-concept corresponding to $\Amc$ is an interpolant for $A \sqsubseteq B$ under 
$\Omc_{1}\cup \Omc_{2}$. Its size is at most exponential in $||\Omc_{1}\cup \Omc_{2}||$ since $(T',V')$ is at most
exponential in $||\Omc_{1}\cup \Omc_{2}||+||\Amc_{\Omc_{1}\cup \Omc_{2},A}^{\Sigma}||$, and so also in
$||\Omc_{1}\cup \Omc_{2}||$.

\section{Interpolant and Explicit Definition Existence in 
$\ELI$ and Extensions}
\label{sec:eli}

We analyze interpolants and explicit definitions for 
$\mathcal{ELI}$ and its extensions with nominals and universal roles,
and show the following result from the introduction.

\medskip
\noindent
{\bf Theorem~\ref{thm:elli}.}
{\em
For $\mathcal{ELI}$ and any extension with any combination of nominals,
the universal role, or $\bot$, the existence of interpolants and 
explicit definitions is \ExpTime-complete. If an interpolant/explicit definition exists, then there exists one of at most 
double exponential size that can be computed in double exponential time. This bound is optimal.
}

\medskip

The double exponential lower bound on the size of explicit definitions and interpolants is shown in the appendix of the full version. The proof is inspired by similar lower bounds for the size of FO-rewritings and uniform interpolants~\cite{DBLP:journals/jsc/LutzW10,DBLP:journals/ai/NikitinaR14}. To prove the remaining claims of
Theorem~\ref{thm:elli}, we lift Theorem~\ref{thm:critele} to $\mathcal{ELI}$.
The main differences are that (1) we now associate undirected graphs with ABoxes and also unfold along inverse roles;
(2) that canonical models become potentially infinite but tree-shaped; (3) that therefore deciding the new variant of Point~2 of
Theorem~\ref{thm:critele} is not an instance of standard entailment checking in $\mathcal{ELI}$, instead we give a reduction to emptiness checking for tree automata; and (4) that to bound the size of $\Amc$ in Point~3, we employ
transfer sequences (and not derivation trees) to represent how facts are derived.
 
In more detail, associate with every ABox $\Amc$ the undirected graph 
$
G^{u}_{\Amc}=(\text{ind}(\Amc),\bigcup_{r\in \NR} \{\{x,y\} \mid r(x,y)\in \Amc\}).
$
We say that $\Amc$ is \emph{tree-shaped} if $G_{\Amc}^{u}$ is acyclic, 
$r(x,y)\in \Amc$ and $s(x,y)\in \Amc$ imply $r=s$,
and $r(x,y)\in \Amc$ implies $s(y,x)\not\in\Amc$ for any $s$. $\Amc$ is \emph{tree-shaped modulo a set $\Gamma$ of individual names} if after dropping some facts $r(x,y)$ with $\{a\}(x)$ or $\{a\}(y)\in \Amc$ for some 
$a\in \Gamma$ it is tree-shaped.
We observe that $\mathcal{ELIO}_{u}(\Sigma)$-concepts correspond to pointed $\Sigma$-ABoxes $\Amc,x$ such that $\Amc$ is tree-shaped modulo $\NI \cap \Sigma$. $\mathcal{ELIO}(\Sigma)$-concepts correspond to \emph{weakly rooted} pointed $\Sigma$-ABoxes $\Amc,x$ such that $\Amc$ is tree-shaped modulo $\NI \cap \Sigma$, where $\Amc,x$ is called weakly rooted if for every $y\in \text{ind}(\Amc)$ there is a path from $x$ to $y$ in $G^{u}_{\Amc}$.

For every $\mathcal{ELIO}_{u}$-ontology $\Omc$ and concept $A$ there exists 
a (potentially infinite) pointed canonical model $\Imc_{\Omc,A},\rho_{A}$  
such that the ABox $\Amc_{\Omc,A}$ corresponding to $\Imc_{\Omc,A}$ is tree-shaped modulo $\NI\cap\text{sig}(\Omc)$. 
The property ($\dagger$) used in the context of canonical models for tractable extensions of $\mathcal{EL}$ holds here as 
well. We also require the \emph{undirected unfolding} of a pointed $\Sigma$-ABox $\Amc,x$ into a pointed $\Sigma$-ABox 
$\Amc^{\ast},x$ which is tree-shaped modulo $\Sigma\cap\NI$.
In the \emph{rooted undirected unfolding}, nodes that cannot be reached from $x$ via roles are dropped.

Assume now that $\Omc$ is in normal form and $A$ a concept name. 
Let $\Amc_{\Omc,A}^{\Sigma}$ be the $\Sigma$-reduct of the canonical model $\Imc_{\Omc,A}$, regarded as an ABox. 
Denote by  $\Amc_{\Omc,A}^{\Sigma,\ast},\rho_{A}$ the undirected unfolding of $\Amc_{\Omc,A}^{\Sigma},\rho_{A}$, by $\Amc_{\Omc,A}^{\downarrow_{w} \Sigma},\rho_{A}$ the sub-ABox of $\Amc_{\Omc,A}^{\Sigma}$ weakly rooted in $\rho_{A}$, and by $\Amc_{\Omc,A}^{\downarrow_{w}\Sigma,\ast},\rho_{A}$ its rooted undirected unfolding. Then we lift Theorem~\ref{thm:critele} as follows.
\begin{theorem}\label{thm:critele2}
	 There exists a polynomial $p$ such that the following conditions are equivalent for all $\mathcal{ELIO}_{u}$-ontologies $\Omc_{1}, \Omc_{2}$ in normal form, concept names $A,B$, and $\Sigma=\text{sig}(\Omc_{1},A) \cap \text{sig}(\Omc_{2},B)$:
	\begin{enumerate}
		\item An $\mathcal{ELIO}_{u}$-interpolant for $A \sqsubseteq B$ under $\Omc_{1},\Omc_{2}$ exists;
		\item $\Omc_{1}\cup \Omc_{2},\Amc_{\Omc_{1}\cup \Omc_{2},A}^{\Sigma}\models B(\rho_{A})$;
		\item there exists a finite subset $\Amc$ of $\Amc_{\Omc_{1}\cup \Omc_{2},A}^{\Sigma,\ast}$ with $|\text{ind}(\Amc)| \leq 2^{2^{p(||\Omc_{1}\cup \Omc_{2}||)}}$ such that the $\mathcal{ELIO}_{u}$-concept corresponding to $\Amc,\rho_{A}$ 
		is an $\mathcal{ELIO}_{u}$-interpolant for $A \sqsubseteq B$ 
under $\Omc_{1},\Omc_{2}$.
	\end{enumerate}
The same equivalences hold if in Points~1 to 3, $\mathcal{ELIO}_{u}$ is 
replaced by $\mathcal{ELIO}$, $\Amc_{\Omc_{1}\cup \Omc_{2},A}^{\Sigma}$ by $\Amc_{\Omc_{1}\cup \Omc_{2},A}^{\downarrow_{w}\Sigma}$,
and $\Amc_{\Omc_{1}\cup \Omc_{2},A}^{\Sigma,\ast}$ by $\Amc_{\Omc_{1}\cup \Omc_{2},A}^{\downarrow_{w}\Sigma,\ast}$.

In Point~3, $\Amc$ can be computed in double exponential time, if it exists.
\end{theorem}
We first sketch how tree automata are used to show that Point~2 entails an exponential time upper 
bound for deciding the existence of an interpolant. To this end we represent finite prefix-closed subsets $\Amc$ of $\Amc_{\Omc_{1}\cup \Omc_{2},A}^{\Sigma}$ as trees and design
\begin{itemize}
\item a non-determistic tree automaton over finite trees (NTA), $\Amf_{1}$, that accepts exactly those trees that 
represent prefix-closed finite subsets of $\Amc_{\Omc_{1}\cup \Omc_{2},A}^{\Sigma}$;
\item a two-way alternating tree automaton over finite trees (2ATA), $\Amf_{2}$, that accepts exactly those trees that 
represent a pointed ABox $\Amc,\rho$ with $\Omc_{1}\cup \Omc_{2},\Amc\models B(\rho)$.
\end{itemize}
Similar tree automata techniques have been used e.g.\ in \cite{JungLMS20}.
$\Amf_{1}$ is constructed using the definition of canonical models; its states
are essentially types occuring in the canonical model and it can be constructed
in exponential time.
The 2ATA $\Amf_{2}$ tries to construct a derivation tree for $B(\rho)$ in
$\Omc_1 \cup \Omc_2,\Amc$, given as input a tree representing $\Amc,\rho$.
It has polynomially many states, and can thus be turned into an equivalent
NTA with exponentially many states \cite{DBLP:conf/icalp/Vardi98}.
By taking the intersection with $\Amf_1$, one can then check in
exponential time whether $L(\Amf_1) \cap L(\Amf_2) \neq \emptyset$,
that is, whether
$\Omc_{1}\cup \Omc_{2},\Amc_{\Omc_{1}\cup \Omc_{2},A}^{\Sigma}\models B(\rho_{A})$.

\medskip

We return to the proof of Theorem~\ref{thm:critele2}. 
The interesting implication is ``2. $\Rightarrow$ 3.''
and the double exponential computation of interpolants.
In this case we use transfer sequences to obtain a bound on the size of the subset $\Amc$ of $\Amc_{\Omc_{1}\cup \Omc_{2},A}^{\Sigma,\ast}$ needed to derive $B(\rho_{A})$ (we note that for $\mathcal{ELI}$ without nominals one can also use the automata encoding above). Transfer sequences describe how facts are derived in a tree-shaped ABox and allow to determine when individuals $a$ and $b$ behave sufficiently similar so that the subtree rooted at $a$ can be replaced by the subtree rooted at $b$~\cite{DBLP:conf/ijcai/BienvenuLW13} without affecting a derivation. This technique can be used to show that one can always choose a prefix closed subset $\Amc$ of $\Amc_{\Omc_{1}\cup \Omc_{2},A}^{\Sigma,\ast}$ of at most exponential depth. This also implies that $\Amc$ can be obtained in double exponential time by constructing the canonical model up to depth $2^{q(||\Omc_{1}\cup \Omc_{2}||)}$ with $q$ a polynomial.
\section{Expressive Horn Description Logics}
\label{sec:final}
We address two questions regarding expressive Horn-DLs. (1) Can our results for $\mathcal{ELI}$ and extensions be lifted to more expressive Horn-DLs? 
(2) In the examples provided in the proof of Theorem~\ref{posneg} we sometimes (for example, for $\mathcal{EL}_{u}$ and $\mathcal{ELI}$) construct explicit Horn-DL definitions to show implicit definability of concept names. Are Horn-DL concepts always sufficient to obtain an explicit definition if an implicit definition exists? We provide a positive answer to (1) if one only admits $\mathcal{ELIO}_{u}$-concepts (or fragments) as interpolants/explicit definitions and a negative answer to (2) in the sense that $\mathcal{ELI}$ and various other Horn-DLs do not enjoy the CIP/PBDP even if one admits Horn-DL concepts as interpolants/explicit definitions.

We introduce expressive Horn DLs~\cite{DBLP:conf/ijcai/HustadtMS05}, 
presented here in the form proposed in~\cite{DBLP:conf/kr/LutzW12}. \emph{Horn-$\mathcal{ALCIO}_{u}$-concepts} $R$ and \emph{Horn-$\mathcal{ALCIO}_{u}$-CIs} $L \sqsubseteq R$ are defined by the syntax rules
\begin{align*} 
	R,R' &::= \top \mid \bot \mid A \mid \neg A \mid \{a\} \mid \neg \{a\} \mid R \sqcap R' \mid 
	L \rightarrow R \mid\\
	  & \hspace*{1cm} \exists r . R \mid
	\forall r . R \\
	L,L' &::= \top \mid \bot \mid A \mid L \sqcap L' \mid L \sqcup L' \mid 
	\exists r . L 
\end{align*}
with $A$ ranging over concept names, $a$ over individual names, and $r$ over roles (including the universal role). 
As usual, the fragment of Horn-$\mathcal{ALCIO}_{u}$ without nominals and the universal role is denoted by Horn-$\mathcal{ALCI}$ and Horn-$\mathcal{ALC}$ denotes the fragment of Horn-$\mathcal{ALCI}$ without inverse roles.
\begin{theorem}\label{thm:7}
	Let $(\Lmc,\Lmc')$ be the pair $($Horn-$\mathcal{ALCI},\mathcal{ELI})$
        or the pair $($Horn-$\mathcal{ALCIO}_{u}$, $\mathcal{ELIO}_{u})$.
	Then 
	\begin{itemize}
		\item deciding the existence of an $\Lmc'$-interpolant for an $\Lmc'$-CI $C \sqsubseteq D$ under 
                      $\Lmc$-ontologies $\Omc_{1},\Omc_{2}$ is \ExpTime-complete; 
		\item deciding the existence of an explicit $\Lmc'(\Sigma)$-definition of a concept name $A$ under an $\Lmc$-ontology $\Omc$ is \ExpTime-complete.
	\end{itemize}
Moreover, if an $\Lmc$-interpolant/explicit definition exists, then there exists one of at most double exponential size that
can be computed in double exponential time.
\end{theorem}
Theorem~\ref{thm:7} follows from Theorem~\ref{thm:elli} and the fact that for any $\Lmc$-ontology one can construct in polynomial
time an $\Lmc'$-ontology in normal form that is a conservative extension of $\Lmc$ (see~\cite{DBLP:conf/ijcai/Bienvenu0LW16} for a similar result).
We next show that despite the fact that Horn-$\mathcal{ALCI}$-concepts sometimes provide explicit definitions if none exist
in $\mathcal{ELI}$ (proof of Theorem~\ref{posneg}), they are not sufficient to prove the CIP/PBDP.
\begin{theorem}\label{thm:horn}
There exists an ontology $\Omc$ in Horn-$\mathcal{ALC}$ (and in $\mathcal{ELI}$),
a signature $\Sigma$, and a concept name $A$ such that $A$ is implicitly definable using $\Sigma$
under $\Omc$ but does not have an explicit Horn-$\mathcal{ALCI}_{u}(\Sigma)$-definition.
\end{theorem}
\begin{proof}
	We modify the ontology used in the proof of Point~1 of Theorem~\ref{posneg}. 
	Let $\Sigma = \{B,D_{1},E,r,r_{1}\}$ and let $\Omc$ contain $B \sqcap \exists r. (C \sqcap E) \sqsubseteq A$ and the following CIs:
$$	
A\sqsubseteq B, \quad B \sqsubseteq \forall r.F,	\quad B \sqsubseteq \exists r. C,
\quad C  \sqsubseteq F \sqcap \forall r_{1}.D_{1}, 
$$
\vspace*{-0.7cm}%
	\begin{align*}
		F & \sqsubseteq \exists r_{1}.D_{1} \sqcap \exists r_{1}.M, \\
		A & \sqsubseteq \forall r.((F \sqcap \exists r_{1}.(D_{1}\sqcap M)) \rightarrow E) \,. 	
	\end{align*}
Intuitively, the final two CIs should be read as
	\begin{align*}
		F & \sqsubseteq \exists r_{1}.D_{1}  \\
		A & \sqsubseteq \forall r.((F \sqcap \forall r_{1}.D_{1}) \rightarrow E) 
	\end{align*}	
	and the concept name $M$ is introduced to achieve this in a projective way as the latter CI is not in Horn-$\mathcal{ALCI}$.
	
	$A$ is implicitly definable using $\Sigma$ under $\Omc$ since 
	$$
	\Omc \models A \equiv B \sqcap \forall r. (\forall r_{1}.D_{1} \rightarrow E).
	$$
	To show that $A$ is not explicitly Horn-$\mathcal{ALCI}_{u}(\Sigma)$-definable under $\Omc$ consider 
        the interpretations $\Imc$ and $\Imc'$ in Figure~\ref{fig:hornsim}. The claim follows from the facts that 
        $\Imc$ and $\Imc'$ are models of 
        $\Omc$, $a\in A^{\Imc}$, $a'\not\in A^{\Imc'}$, but $a\in F^{\Imc}$ implies $a'\in F^{\Imc'}$ holds for every Horn-$\mathcal{ALCI}_{u}(\Sigma)$-concept $F$. The latter can be proved by observing that there exists a Horn-$\mathcal{ALCI}_{u}(\Sigma)$-simulation between $\Imc$ and $\Imc'$ \cite{DBLP:conf/lics/JungPWZ19} containing $(\{a\},a)$, we refer the reader to the appendix of the full version. To obtain an example in $\mathcal{ELI}$, it suffices to take a conservative extension of $\mathcal{O}$ in $\mathcal{ELI}$.
\end{proof}        
	\begin{figure}
		\begin{tikzpicture}[auto,->,
                  every label/.style={font=\small,inner sep=0pt},
                  el/.style={},
			pl/.style={font=\small,inner sep=2pt},
			r1/.style={},r2/.style={},yscale=0.85]

			\node[el,label=above:{$A,B$}] (1) {$a$} ;
			\node[el,label={[yshift=0.1cm]below right:{$C,E,F$}}] at ($(1)-(1,1.5)$) (2) {$b$} ;
			\path (1) edge[above left] node[pl] {$r$} (2) ;
			\node[el,label={[yshift=0.1cm]below right:{$F$}}] at ($(1)+(1,-1.5)$) (3) {$c$} ;
			\path (1) edge node[pos=0.25,pl] {$r$} (3) ;
			
			\node[el,label=below:{$D_1,M$}] at ($(2)-(0,1.5)$) (4) {$d$} ;
			\path[r1] (2) edge[left] node[pl] {$r_1$} (4) ;
			
			\node[el,label=below:{$D_1$}] at ($(3)+(-0.5,-1.5)$) (5) {$e$} ;
			\path[r1] (3) edge[above left] node[pl,pos=0.7] {$r_1$} (5) ;
			\node[el,label=below:{$M$}] at ($(3)+(0.5,-1.5)$) (6) {$f$} ;
			\path[r1] (3) edge node[pl,pos=0.7] {$r_1$} (6) ;

			\node[el,label=above:{$B$}] at ($(1)+(4.5,0)$) (1') {$a'$} ;
			\node[el,label={[yshift=0.1cm]below right:{$E,F$}}] at ($(1')-(1.5,1.5)$) (2') {$b'$} ;
			\path (1') edge[above left] node[pl] {$r$} (2') ;
			\node[el,label={[yshift=0.1cm]below right:{$C,F$}}] at ($(1')-(0,1.5)$) (25') {$b''$} ;
			\path (1') edge[above left] node[pl] {$r$} (25') ;
			\node[el,label={[yshift=0.1cm]below right:{$F$}}] at ($(1')+(1.5,-1.5)$) (3') {$c'$} ;
			\path (1') edge node[pos=0.25,pl] {$r$} (3') ;
			
			\node[el,label=below:{$D_1,M$}] at ($(2')-(0,1.5)$) (4') {$d'$} ;
			\path[r1] (2') edge[left] node[pl] {$r_1$} (4') ;
			
			\node[el,label=below:{$D_1,M$}] at ($(25')-(0,1.5)$) (45') {$d''$} ;
			\path[r1] (25') edge[left] node[pl] {$r_1$} (45') ;
			
			\node[el,label=below:{$D_1$}] at ($(3')+(-0.5,-1.5)$) (5') {$e'$} ;
			\path[r1] (3') edge[above left] node[pl,pos=0.7] {$r_1$} (5') ;
			\node[el,label=below:{$M$}] at ($(3')+(0.5,-1.5)$) (6') {$f'$} ;
			\path[r1] (3') edge node[pl,pos=0.75] {$r_1$} (6') ;
			
                      \end{tikzpicture}
\caption{Interpretations $\Imc$ (left) and $\Imc'$ (right).}
\label{fig:hornsim}		
\end{figure}       

\section{Discussion}
\label{sec:discussion}
For a few important extensions of $\mathcal{EL}/\mathcal{ELI}$ the complexity of interpolant and explicit definition existence remains to be investigated. Examples include extensions of $\mathcal{ELI}$ with role inclusions, and extensions of $\mathcal{EL}$ or $\mathcal{ELI}$ with functional roles or more general number restrictions.
It would also be of interest to investigate interpolant existence if Horn-concepts are admitted as interpolants (using, for example, the games introduced in  \cite{DBLP:conf/lics/JungPWZ19}). Finally, the question arises whether there exists at all a decidable Horn language extending, say, Horn-$\mathcal{ALCI}$, with the CIP/PBDP.  
We note that Horn-FO enjoys the CIP (Exercise 6.2.6 in~\cite{modeltheory}) but is undecidable and that we show in the appendix of the full version that the Horn fragment
of the guarded fragment does not enjoy the CIP/PBDP. 
\section*{Acknowledgments}
This research was supported by the EPSRC UK grant EP/S032207/1.

\bibliographystyle{kr}
\bibliography{lics21}

\newpage

\newpage

\appendix

\section{Further Prelimaries}
We call an ontology $\Omc'$ a \emph{conservative extension} of an ontology $\Omc$ if $\Omc'\models \alpha$ for all $\alpha\in \Omc$ and every model $\Imc$ of $\Omc$ can be expanded to a model $\Jmc$ of $\Omc'$ by modifying the interpretation of symbols in $\text{sig}(\Omc')\setminus\text{sig}(\Omc)$. In other words, the $\text{sig}(\Omc)$-reducts of $\Imc$ and $\Jmc$ coincide.
The following result is folklore~\cite{IJCAI05-short}.
\begin{lemma}\label{lem:normalform}
	Let $\Lmc$ be any DL from $\EL,\ELI,\ELO,\mathcal{ELRO},\mathcal{ELIO}$ or an extension with the universal role,
	and let $\Omc$ be an $\Lmc$-ontology. Then one can construct in polynomial time an $\Lmc$-ontology $\Omc'$ in normal form such that $\Omc'$ is a  conservative extension of $\Omc$.
\end{lemma}
We next give a more detailed introduction to ABoxes and how they relate to concepts. Recall that an \emph{ABox $\Amc$} is a (possibly infinite) set of assertions of the form
$A(x)$, $r(x,y)$, $\{a\}(x)$, and $\top(x)$ with $A\in \NC$, $r\in \NR$, $a\in \NI$, and $x,y$ individual variables.
An ABox is \emph{factorized} if $\{a\}(x),\{a\}(y)\in \Amc$ imply $x=y$.

ABox assertions are interpreted in an interpretation $\Imc$ using a \emph{variable assignment} $v$ that maps individual variables to elements of $\Delta^{\Imc}$. 
Then $\Imc,v$ satisfies an assertion $A(x)$ if $v(x)\in A^{\Imc}$,
$r(x,y)$ if $(v(x),v(y))\in r^{\Imc}$, $\{a\}(x)$ if $a^{\Imc}=v(x)$, and $\top(x)$ is always satisfied.
$\Imc,v$ satisfies an ABox if it satisfies all assertions in it. 
We write $\Imc\models \Amc[x\mapsto d]$ if there exists an assignment
$v$ with $v(x)=d$ such that $\Imc,v$ satisfies $\Amc$. We say that an assertion \emph{$A_{0}(x_0)$ is entailed by an ontology $\Omc$ and ABox $\Amc$}, in symbols $\Omc,\Amc\models A_{0}(x_{0})$, if
$v(x)\in A_{0}^{\Imc}$ for all models $\Imc$ of $\Omc$ and assignments $v$ such that $\Imc,v$ satisfy $\Amc$.
This is the standard notion of entailment from a knowledge base consisting of an ontology and an ABox.
Deciding entailment is in \PTime for the DLs between $\mathcal{EL}$ and $\mathcal{EL}^{++}_{u}$~\cite{IJCAI05-short}
and \ExpTime-complete for the DLs between $\mathcal{ELI}$ and $\mathcal{ELIO}_{u}$~\cite{DBLP:books/daglib/0041477}.    

Every interpretation $\Imc$ defines a factorized ABox $\Amc_{\Imc}$ by identifying every $d\in \Delta^{\Imc}$ 
with a variable $x_{d}$ and taking $A(x_{d})$ if $d\in A^{\Imc}$, $r(x_{c},x_{d})$ if $(c,d)\in r^{\Imc}$,
$\{a\}(x_{d})$ if $a^{\Imc}=d$. Conversely, factorized ABoxes define interpretations in the obvious way.

The following lemma provides a formal description of the relationship between
ABoxes that are ditree-shaped modulo some set of individual names and $\mathcal{ELO}$-concepts.
  
\begin{lemma}\label{lem:dirunfold}
	For any $\ELO_{u}(\Sigma)$-concept $C$ one can construct in polynomial time
	a pointed $\Sigma$-ABox $\Amc,x$ such that $\Amc$ is ditree-shaped modulo $\NI \cap \Sigma$ and $d\in C^{\Imc}$ iff $\Imc\models \Amc[x\mapsto d]$,
	for all interpretations $\Imc$ and $d\in \Delta^{\Imc}$.
	
	Conversely, for any pointed $\Sigma$-ABox $\Amc,x$ such that $\Amc$ is a ditree-shaped ABox modulo $\Gamma$, one can construct in polynomial time an $\ELO_{u}(\Sigma)$-concept $C$ such that $\Gamma=\NI \cap \Sigma$ and $d\in C^{\Imc}$ iff $\Imc\models \Amc_{C}[x\mapsto d]$,
	for all interpretations $\Imc$ and $d\in \Delta^{\Imc}$.
	
	The above also holds if one replaces $\ELO_{u}(\Sigma)$-concepts by $\mathcal{ELO}(\Sigma)$-concepts and requires the pointed ABoxes to be rooted.
\end{lemma}
We define a \emph{canonical model} $\Imc_{\Omc,A_{0}}$ for an $\mathcal{ELRO}_{u}$-ontology $\Omc$ in normal form and a concept name $A_{0}$. 
This has been done in~\cite{IJCAI05-short}, but as we do not use canonical models for subsumption or instance checking we give a succinct model-theoretic construction. 

Assume $\Omc$ and $A_{0}$ are given and $\Omc$ is in normal form.
Define an equivalence relation $\sim$ on the set of individual names $a$ in $\text{sig}(\Omc)$ by setting $a\sim b$ if $\Omc \models \exists u.A_{0} \sqcap \{a\} \sqsubseteq \{b\}$. Let $[a] = \{ b \in \text{sig}(\Omc) \mid a\sim b\}$ and set $\Delta_{I}=\{ [a]\mid a\in \text{sig}(\Omc)\}$. Say that a concept name $A$ is \emph{absorbed by} an individual name $a$ if $\Omc \models \exists u.A_{0} \sqcap A \sqsubseteq \{a\}$ and let
$\Delta_{C}$ denote the set of concept names $A$ in $\Omc$ such that $\Omc\models A_{0} \sqsubseteq \exists u.A$ and $A$ is not absorbed by any individual name.

Now let $\Delta^{\Imc_{\Omc,A_{0}}}= \Delta_{I} \cup \Delta_{C}$ and let
\begin{eqnarray*}
	A^{\Imc_{\Omc,A_{0}}} & = & \{ [a] \in \Delta^{\Imc_{\Omc,A_{0}}} \mid \Omc \models \exists u.A_{0} \sqcap \{a\} \sqsubseteq A \} \cup \\
	& & \{ B \in \Delta^{\Imc_{\Omc,A_{0}}} \mid \Omc\models\exists u.A_{0} \sqcap B \sqsubseteq A \}\\
	a^{\Imc_{\Omc,A_{0}}} & = & [a]\\
	r^{\Imc_{\Omc,A_{0}}} & = & \{([a],[b]) \in \Delta^{\Imc_{\Omc,A_{0}}}\times \Delta^{\Imc_{\Omc,A_{0}}}\mid \\
	& & \hspace*{1cm} \Omc \models \exists u.A_{0} \sqcap \{a\} \sqsubseteq \exists r.\{b\} \}\cup\\
	&  &  \{([a],B) \in \Delta^{\Imc_{\Omc,A_{0}}}\times \Delta^{\Imc_{\Omc,A_{0}}}\mid \\
	&  & \hspace*{1cm}\Omc \models \exists u.A_{0} \sqcap \{a\} \sqsubseteq \exists r.B \}\cup  \\
	&  &  \{(B,[a])\in \Delta^{\Imc_{\Omc,A_{0}}}\times \Delta^{\Imc_{\Omc,A_{0}}} \mid \\
	&  & \hspace*{1cm} \Omc \models \exists u.A_{0} \sqcap B \sqsubseteq \exists r.\{a\} \} \cup \\
	&  &  \{(A,B)\in \Delta^{\Imc_{\Omc,A_{0}}}\times \Delta^{\Imc_{\Omc,A_{0}}} \mid \\
	& & \hspace*{1cm} \Omc \models \exists u.A_{0} \sqcap A \sqsubseteq \exists r.B \}
\end{eqnarray*}
for every concept name $A\in\NC$, $a\in \text{sig}(\Omc)\cap \NI$, and $r\in \NR$. We often denote the nodes $[a]$ and $A$ by $\rho_{[a]}$ or, for simplicity, $\rho_{a}$ and, respectively, $\rho_{A}$. If $A_{0}$ is absorbed by an individual $a$ we still often denote
$\rho_{[a]}$ by $\rho_{A_{0}}$.

\begin{lemma}\label{lem:canonicalel}
	The canonical model $\Imc_{\Omc,A_{0}}$ is a model of $\Omc$ and for every model $\Jmc$ of $\Omc$ and any $d\in \Delta^{\Jmc}$ with $d\in A_{0}^{\Jmc}$, $(\Imc_{\Omc,A_{0}},\rho_{A_{0}}) \preceq_{\mathcal{ELO}_{u},\Sigma} (\Jmc,d)$, where $\Sigma$ is any signature.
\end{lemma}
\begin{proof}
	We first show that $\Imc_{\Omc,A_{0}}$ is a model of $\Omc$. It is straightforward to show that $\Imc_{\Omc,A_{0}}$ satisfies the CIs of the 
	form $\top \sqsubseteq A, A_{1} \sqcap A_{2} \sqsubseteq A$, $A\sqsubseteq \{a\}$, $\{a\} \sqsubseteq A$. 
	
	Assume now that $A \sqsubseteq \exists r.B\in\Omc$ and
	$\rho_{C}\in A^{\Imc_{\Omc,A_{0}}}$ with $C$ of the form $a$ or $A$.
	We have $\Omc\models \exists u.A_{0} \sqsubseteq \exists u.C$, $\Omc\models \exists u.A_{0} \sqcap C
	\sqsubseteq A$. Thus $\Omc\models \exists u.A_{0} \sqcap C
	\sqsubseteq \exists r. B$. But then $(\rho_{C},\rho_{B})\in r^{\Imc_{\Omc,A_{0}}}$ and $\rho_{B}\in B^{\Imc_{\Omc,A_{0}}}$.
	Thus $\rho_{C}\in (\exists r,B)^{\Imc_{\Omc,A_{0}}}$, as required. 
	
	Assume now that $\exists r.A \sqsubseteq B\in \Omc$ and $\rho_{C}\in (\exists r.A)^{\Imc_{\Omc,A_{0}}}$. Then there exists $\rho_{D}$ such that
	$(\rho_{C},\rho_{D})\in r^{\Imc_{\Omc,A_{0}}}$ and $\rho_{D}\in A^{\Imc_{\Omc,A_{0}}}$. Hence $\Omc\models \exists u.A_{0} \sqcap C \sqsubseteq \exists r.D$
	and $\Omc\models \exists u.A_{0} \sqcap D
	\sqsubseteq A$. Thus, $\Omc\models \exists u.A_{0} \sqcap C \sqsubseteq \exists r.A$.
	Hence since $\exists r.A\sqsubseteq B\in \Omc$, $\Omc\models \exists u.A_{0} \sqcap C \sqsubseteq B$. But then $\rho_{C}\in B^{\Imc_{\Omc,A_{0}}}$, as required.
	
	Finally, assume that $r_{1}\circ \cdots \circ r_{n} \sqsubseteq r\in \Omc$
	and $(\rho_{C},\rho_{D})\in r_{1}^{\Imc_{\Omc,A_{0}}}\circ \cdots \circ r_{n}^{\Imc_{\Omc,A_{0}}}$. Then there are $\rho_{C_{0}},\ldots,\rho_{C_{n}}$
	with $(\rho_{C_{i}},\rho_{C_{i+1}})\in r_{i+1}^{\Imc_{\Omc,A_{0}}}$
	for all $i<n$, where $C_{0}=C$ and $C_{n}=D$. We obtain
	$\Omc\models \exists u.A_{0} \sqcap C_{i}
	\sqsubseteq \exists r_{i+1}.C_{i+1}$ for all $i<n$. Thus
	$\Omc\models \exists u.A_{0} \sqcap C
	\sqsubseteq \exists r_{1} \cdots \exists r_{n}.D$.
	Hence $\Omc\models \exists u.A_{0} \sqcap C
	\sqsubseteq \exists r.D$. Hence $(\rho_{C},\rho_{D})\in r^{\Imc_{\Omc,A_{0}}}$,
	as required.
	
	Let $\Jmc$ be a model of $\Omc$ with $A_{0}^{\Jmc}\not=\emptyset$. Define a relation between $\Delta^{\Imc_{\Omc,A_{0}}}$ and $\Delta^{\Jmc}$ as follows: for any $\rho_{C}\in \Delta^{\Imc_{\Omc,A_{0}}}$ and $d\in \Delta^{\Jmc}$, let $(\rho_{C},d)\in S$ if $d\in C^{\Jmc}$. One can now show that this is well-defined and that for any $\rho_{C}$ there exists a $d\in \Delta^{\Jmc}$ with $(\rho_{C},d)\in S$. It is straightforward to show that $S$ is a $\ELOu(\Sigma)$-simulation, as required.	
\end{proof}
The following observation is a consequence of Lemma~\ref{lem:simulationel} and Lemma~\ref{lem:canonicalel}.	
\begin{lemma}\label{lem:can2}
	Let $\Omc$ be an $\mathcal{ELRO}_{u}$-ontology in normal form, $A_{0}$ a concept name, and $C$ an $\ELOu$-concept. Then the following conditions are equivalent:
	\begin{enumerate}
		\item $\rho_{A_{0}}\in C^{\Imc_{\Omc,A_{0}}}$;
		\item $\Omc\models A_{0} \sqsubseteq C$.
	\end{enumerate}
\end{lemma}

Next assume that $\Omc$ and an ABox $\Amc$ are given. Assume $\Omc$ is in normal form. Then one can construct in polynomial time a canonical model $\Imc_{\Omc,\Amc}$ of $\Omc$
that satisfies $\Amc$ via an assignment $v_{\Omc,\Amc}$. The details are straightforward, and we only give the main properties of $\Imc_{\Omc,\Amc}$.

\begin{lemma}\label{lem:canonicalelabox}
	Given an $\mathcal{ELRO}_{u}$-ontology $\Omc$ in normal form and an ABox $\Amc$ one can construct in polynomial time a model $\Imc_{\Omc,\Amc}$ of $\Omc$ and an assignment $v_{\Omc,\Amc}$ such that for all $x\in \text{ind}(\Amc)$ and all $\ELOu$-concepts $C$ the following conditions
	are equivalent:
	\begin{enumerate}
		\item $v_{\Omc,\Amc}(x)\in C^{\Imc_{\Omc,\Amc}}$;
		\item $\Omc,\Amc\models C(x)$.
	\end{enumerate}
\end{lemma}

The following lemma provides a formal description of the relationship between
ABoxes that are tree-shaped modulo some set of individual names and $\mathcal{ELIO}$-concepts.

\begin{lemma}
	For any $\mathcal{ELIO}_{u}(\Sigma)$-concept $C$ one can construct in polynomial time
	a pointed $\Sigma$-ABox $\Amc,x$ such that $\Amc$ is tree-shaped modulo $\NI \cap \Sigma$ and $d\in C^{\Imc}$ iff $\Imc\models \Amc[x\mapsto d]$,
	for all interpretations $\Imc$ and $d\in \Delta^{\Imc}$.
	
	Conversely, for any pointed $\Sigma$-ABox $\Amc,x$ such that $\Amc$ is a tree-shaped ABox modulo $\Gamma$, one can construct in polynomial time an $\mathcal{ELIO}_{u}(\Sigma)$-concept $C$ such that $\Gamma=\NI \cap \Sigma$ and $d\in C^{\Imc}$ iff $\Imc\models \Amc_{C}[x\mapsto d]$,
	for all interpretations $\Imc$ and $d\in \Delta^{\Imc}$.
	
	The above also holds if one replaces $\mathcal{ELIO}_{u}(\Sigma)$-concepts by $\mathcal{ELIO}(\Sigma)$-concepts and requires the pointed ABoxes to be weakly rooted.
\end{lemma}

\section{Proof for Section~\ref{CIPandPBDP}}

We start by proving Remark~\ref{rem:bot1}.

\medskip
\noindent
{\bf Proof of Remark~\ref{rem:bot1}.}
	We have to show that the CIP and PBDP are invariant under adding $\bot$ (interpreted as the empty set) to the languages introduced in this paper. 
    Assume that $\Lmc$ is any such language and let $\Lmc_{\bot}$ denote
	its extension with $\bot$. We claim that $\Lmc$ enjoys the CIP/PBDP iff $\Lmc_{\bot}$ does. We show this for the CIP, the proof for the PBDP is similar. Assume first that $C\sqsubseteq D$ and $\Omc_{1},\Omc_{2}$ are a counterexample to the CIP of $\Lmc$. Then they are also a counterexample to the CIP of
	$\Lmc_{\bot}$. Conversely, assume that $C\sqsubseteq D$ and $\Omc_{1},\Omc_{2}$ are a counterexample to the CIP of $\Lmc_{\bot}$.  We may assume that no CI in $\Omc_{1}\cup \Omc_{2}$ uses $\bot$ in the concept on its left hand side (if it does, the CI is redundant). 
	Let $B$ be a fresh concept name and replace $\bot$ by $B$
	in $\Omc_{1}$ and $\Omc_{2}$. Also add to $\Omc_{i}$ the CIs
	$$
	\exists r.B \sqsubseteq B, \quad B \sqsubseteq A \sqcap \exists r.B
	$$
	for all role names $r$ in $\text{sig}(\Omc_{i}$) and $A\in \text{sig}(\Omc_{i})$.
	We also let $r$ range over inverse roles in $\text{sig}(\Omc_{i}$) if $\Lmc$ admits inverse roles, the universal role if $\Lmc$ admits the universal role, and $A$ over nominals in $\text{sig}(\Omc_{i}$) if $\Lmc$ admits nominals. Let $\Omc_{i}'$ denote the resulting ontology. Then it is easy to see that $C\sqsubseteq D$ and $\Omc_{1}',\Omc_{2}'$
	are a counterexample to the CIP of $\Lmc$.

\medskip

We continue with a few comments and missing proofs for Theorem~\ref{posneg}.

\medskip
\noindent
{\bf Theorem~\ref{posneg}.}
{\em The following DLs do not enjoy the CIP nor PBDP:
	\begin{enumerate}
		\item $\mathcal{EL}$ with the universal role, 
		\item $\mathcal{EL}$ with nominals, 
		\item $\mathcal{EL}$ with a single role inclusion $r\circ s\sqsubseteq s$, 
		\item $\mathcal{EL}$ with role hierarchies and a transitive role, 
		\item $\mathcal{EL}$ with inverse roles.
	\end{enumerate}	
In Points~2 to 5, the CIP/PBDP also fails if the universal role
can occur in interpolants/explicit definitions.}
\begin{proof}
We first supply a proof for Point~4.
Let $\Omc_{rs}$ contain
$$
A \sqsubseteq \exists s.E, \quad E \sqsubseteq \exists s_1.B, \quad \exists s_2.B \sqsubseteq A, 
$$
$$
s_1 \sqsubseteq s, \quad s\sqsubseteq s_2, \quad s\circ s\sqsubseteq s,
$$
and let $\Sigma= \{s_1, s_2,E\}$. Then $A$ is implicitly definable using $\Sigma$ under $\Omc_{rs}$ since
	$$
	\Omc_{rs} \models \forall x (A(x) \leftrightarrow \exists y (E(y) \wedge \forall z (s_1(y,z) \rightarrow s_2(x,z))).
	$$
	In the same way as above, the interpretations $\I$ and $\I'$ given in Figure~\ref{figure:fig3} show that $A$ has no $\mathcal{EL}_u(\Sigma)$-definition under $\Omc_{rs}$.
	\begin{figure}[th]
		\begin{center}
			\begin{tikzpicture}[every label/.style={font=\small,inner sep=0pt},node distance = 1cm and 1.5cm]
				\node[label=above:{${\color{blue}A}$}] (a) {$a$};
				\node[below right=of a,label=above right:{$E,{\color{blue}A}$}] (c) {$c$};
				\node[below=of a, left=of c, label=below:{${\color{blue}B}$}] (b) {$b$};
				
				\node[below=2.5cm of a] (a') {$a'$};
				\node[below right=of a',label=above right:{$E,{\color{blue}A}$}] (c') {$c'$};
				\node[below=of a', left=of c'] (b') {$b'$};
				\node[below right=of c',label=above right:{$E,{\color{blue}A}$}] (f) {$c''$};
				\node[below=of c', left=of f, label=below:{${\color{blue}B}$}] (g) {$b''$};
				
				\path[->] 
				(a) edge node[left]  {${\color{blue}s},s_2$} (b)
				(a) edge node[above right] {${\color{blue}s},s_2$} (c)
				(c) edge node[below] {$s_1,{\color{blue}s},s_2$} (b)
				(c) edge[in=-30,out=-60, loop] node[below right] {${\color{blue}s},s_2$} (c)
				(a') edge node[left] {$s_2$} (b')
				(a') edge node[above right] {$s_2$} (c')
				(c') edge node[below] {$s_1,{\color{blue}s},s_2$} (b')
				(c') edge node[left] {${\color{blue}s},s_2$} (g)
				(c') edge node[above right] {${\color{blue}s},s_2$} (f)
				(c') edge[in=10,out=-10, loop] node[right] {$s_2$} (c')
				(f) edge node[below] {$s_1, {\color{blue}s}, s_2$} (g)
				(f) edge[in=-30,out=-60, loop] node[below right] {${\color{blue}s},s_2$} (f)
				;
				
				\node[left=0.5cm of a] {$\I$:} ;
				\node[left=0.5cm of a'] {$\I'$:} ;
			\end{tikzpicture}
		\end{center}
		\caption{Interpretations $\Imc$  and $\Imc'$  used for $\Omc_{rs}$.}
		\label{figure:fig3}
	\end{figure}

We next observe that Point~5 can easily be strengthened. The concept name $A$ does not only have no explicit 
$\mathcal{ELI}_u(\Sigma)$-definition, but no such definition exists in the positive fragment of $\mathcal{ALCI}_u$. 
To see this, consider the interpretations given in Figure~\ref{figure:fig5}.
	\begin{figure}[th]
		\begin{center}
			\begin{tikzpicture}[every label/.style={font=\small,inner sep=0pt},node distance = 0.5cm and 0.5cm]
				\node[label=above:{${\color{blue}A},B$}] (a) {$a$};
				\node[below left=of a,label=below:{${\color{blue}C},D,E$}] (b) {$b$};
				\node[below right=of a] (c) {$c$};
				
				\node[label=above:{$B$},right=4cm of a] (a') {$a'$};
				\node[below left=of a',label=below:{$D,E$}] (b') {$b'$};
				\node[below=of a'] (c') {$c'$};
				\node[below right=of a',label=below:{${\color{blue}C},D$}] (d') {$c''$};
				
				\foreach \i/\j in {a/b,a/c,a'/b',a'/c',a'/d'} {
					\draw[->] (\i) -- (\j) ;
				}
				
			\end{tikzpicture}
		\end{center}
        \caption{Interpretations $\Imc$ (left) and $\Imc'$ (right) for $\Omc_{i}$.
        \label{figure:fig5}}
	\end{figure}
Observe that the interpretations $\Imc,\Imc'$ show that $A$ is not definable under
$\Omc_{i}$ using any concept constructed from $\Sigma$ using $\sqcap,\sqcup,
\exists, \forall$ since for any such concept $F$ we have for 
$(x,x') \in \{(a,a'),(b,b'),(c,c'),(c,c'')\}$ that $x \in F^\I$ implies $x' \in F^\I$.
Of course, the interpretations $\Imc$ and $\Imc'$ given in Figure~\ref{figure:fig5}.
also demonstrate that concepts with implicit definitions in $\EL_u$ may not have
explicit definitions in positive $\ALC_u$. The interpretations depicted 
in Figure~\ref{figure:fig5} differ from the interpretations constructed previously
in that they are not the canonical models. The nodes $c$ and $c'$ are not enforced 
by the ontology but are needed to ensure $\forall r.E$ does not distinguish $a$ and $a'$.
\end{proof}

We defer the proof of Theorem~\ref{th:safety} to the end of
Section~\ref{sec:proofs_existence1} as we need the canonical model and ABox
unfolding machinery developed in that section.

\section{Proofs for Section~\ref{sec:exist}}

We give a proof for Remark~\ref{rem:bot2}.

\medskip
\noindent
{\bf Proof of Remark~\ref{rem:bot2}.}
Assume that $\Lmc$ is any DL introduced in this paper and let $\Lmc_{\bot}$ denote
its extension with $\bot$. The polynomial time reductions of 
$\Lmc$-interpolant existence and $\Lmc$-explicit definition existence to 
$\Lmc_{\bot}$-interpolant existence and $\Lmc_{\bot}$-explicit definition existence, 
respectively, are trivial. For the converse direction, we consider the CIP, the reduction for the PBDP is similar. 
The idea is the same as in Remark~\ref{rem:bot1}. 
Assume that $C\sqsubseteq D$ and $\Omc_{1},\Omc_{2}$ are in $\Lmc_{\bot}$.
If $\Omc_{1}\cup \Omc_{2} \models C \sqsubseteq \bot$, then an interpolant exists and we are done. Assume $\Omc_{1}\cup \Omc_{2} \not\models C \sqsubseteq \bot$. We may assume that no CI in $\Omc_{1}\cup \Omc_{2}$ uses $\bot$ in the concept on its left hand side (if it does, the CI is redundant). Now let $B$ be a fresh concept name and replace $\bot$ by $B$
in $C$, $D$, $\Omc_{1}$, and $\Omc_{2}$. Also add to $\Omc_{i}$ the CIs
$$
\exists r.B \sqsubseteq B, \quad B \sqsubseteq A \sqcap \exists r.B
$$
for all role names $r$ in $\text{sig}(\Omc_{i}$) and $A\in \text{sig}(\Omc_{i})$.
We also let $r$ range over inverse roles in $\text{sig}(\Omc_{i}$) if $\Lmc$ admits inverse roles, the universal role if $\Lmc$ admits the universal role, and $A$ over nominals in $\text{sig}(\Omc_{i}$) if $\Lmc$ admits nominals. Let $\Omc_{i}'$ denote the resulting ontology. Then there exists an $\Lmc_{\bot}$-interpolant for $C \sqsubseteq D$ under $\Omc_{1},\Omc_{2}$ iff there exists an $\Lmc$-interpolant for $C \sqsubseteq D$ under $\Omc_{1}',\Omc_{2}'$. 

\section{Proofs for Section~\ref{sec:existence1}}
\label{sec:proofs_existence1}
We first give a proof of the polynomial time decidability of interpolant existence that has not been
discussed in the main paper. Then we provide the missing proofs from the main paper.

The following complexity upper bound proof does not
provide an upper bound on the size of interpolants/explicit definitions, but is
more elementary than the one we sketched in the main paper.

We start by proving a characterization for the existence of interpolants using canonical models and simulations. 
\begin{lemma}\label{lem:crit}
	Let $\Omc_{1},\Omc_{2}$ be $\mathcal{ELRO}_{u}$-ontologies in normal form, $A,B$ concept names, 
        and $\Lmc\in \{\mathcal{ELO}, \mathcal{ELO}_{u}\}$.
	Let $\Sigma= \text{sig}(\Omc_{1},A) \cap \text{sig}(\Omc_{2},B)$. Then
	there does not exist an $\Lmc$-interpolant for $A\sqsubseteq B$ under $\Omc_{1},\Omc_{2}$ iff
	there exists a model $\Jmc$ of $\Omc_{1} \cup \Omc_{2}$
	and $d\in \Delta^{\Jmc}$ such that 
	\begin{enumerate}
		\item $d\not\in B^{\Jmc}$;
		\item $(\Imc_{\Omc_{1}\cup \Omc_{2},A},\rho_{A}) \preceq_{\Lmc,\Sigma} (\Jmc,d)$.
	\end{enumerate}
\end{lemma}
\begin{proof}
	Assume an $\Lmc$-interpolant $F$ exists, but there exists a model $\Jmc$ of $\Omc_{1} \cup \Omc_{2}$
	and $d\in \Delta^{\Jmc}$ satisfying the conditions of the lemma.
	As $\Omc_{1}\cup \Omc_{2} \models A \sqsubseteq F$, by Lemma~\ref{lem:canonicalel}, we obtain $\rho_{A}\in F^{\Imc_{\Omc_{1}\cup \Omc_{2},A}}$. By Lemma~\ref{lem:simulationel}, $d\in F^{\Jmc}$.
	We have derived a contradiction to the condition that $d\not\in B^{\Jmc}$,
	$\Jmc$ is a model of $\Omc_{1}\cup\Omc_{2}$, and $\Omc_{1}\cup \Omc_{2} \models F \sqsubseteq B$.
	
	Assume no $\Lmc$-interpolant exists. 
	Let
	$$
	\Gamma = \{ C \in\Lmc(\Sigma) \mid \rho_{A} \in C^{\Imc_{\Omc_{1}\cup \Omc_{2},A}}\}
	$$
	By Lemma~\ref{lem:can2} and compactness, there exists a model $\Jmc$ of $\Omc_{1}\cup \Omc_{2}$ and $d\in \Delta^{\Jmc}$ such that $d\in C^{\Jmc}$ for all $C\in \Gamma$ but
	$d\not\in B^{\Jmc}$. We may assume that $\Jmc$ is $\omega$-saturated.\footnote{See~\cite{modeltheory} for an introduction to $\omega$-saturated interpretations and their properties.} Thus,
	by a straightforward gneralization of Lemma~\ref{lem:simulationel} from finite to $\omega$-saturated interpretations, $(\Imc_{\Omc_{1}\cup \Omc_{2},A},\rho_{A}) \preceq_{\Lmc,\Sigma} (\Jmc,d)$, and $\Jmc$ satisfies the conditions of the lemma.
\end{proof}
The characterization provided in Lemma~\ref{lem:crit} can be checked in polynomial time. Consider a fresh concept name $X_{d}$ for each $d \in \Delta^{\Imc}$ for $\Imc= \Imc_{\Omc_{1}\cup \Omc_{2},A}$.
We define the \emph{$\ELO_{u}(\Sigma)$ diagram $\Dmc(\Imc)$} of $\Imc$ as the ontology consisting of the following CIs:
\begin{itemize}
	\item $X_{d} \sqsubseteq A$, for every $A\in \Sigma$ and $d \in A^{\Imc}$;
	\item $X_{b^{\Imc}}  \sqsubseteq \{b\}$, for every $b\in \Sigma$;
	\item $X_{d} \sqsubseteq \exists r.X_{d'}$, for every $r\in \Sigma$ and $(d,d')\in r^{\Imc}$;
	\item $X_{d} \sqsubseteq \exists u.X_{d'}$, for every $d, d' \in \Delta^{\Imc}$.
\end{itemize}

Denote by $\Imc_{|\Sigma}$ the $\Sigma$-reduct of the interpretation $\Imc$.
Now it is straightforward to show that there exists a model $\Jmc$ of $\Omc_{1}\cup \Omc_{2}$ and $d\in \Delta^{\Jmc}$ such that the conditions of Lemma~\ref{lem:crit} hold for $\Lmc=\ELOu$ iff 
$\Omc_{1}\cup \Omc_{2} \cup \Dmc((\Imc_{\Omc_{1}\cup\Omc_{2},A})_{|\Sigma})\not\models X_{\rho_{A}}\sqsubseteq B$. The latter condition can be checked in polynomial time. If we aim at interpolants without the universal role we simply remove the CIs of the final item from the definition of $\Dmc(\Imc)$, denote the resulting set of inclusions by $\Dmc'(\Imc)$ and have that there exists a model $\Jmc$ of $\Omc_{1}\cup \Omc_{2}$ and $d\in \Delta^{\Jmc}$ such that the conditions of Lemma~\ref{lem:crit} hold for $\Lmc=\mathcal{ELO}$ iff 
$\Omc_{1}\cup \Omc_{2} \cup \Dmc'((\Imc_{\Omc_{1}\cup\Omc_{2},A})_{|\Sigma})\not\models X_{\rho_{A}}\sqsubseteq B$.

\bigskip

\paragraph{Directed Unfolding of ABox.}
We give a precise definition of the directed unfolding of an ABox. 
Let $\Amc$ be a factorized $\Sigma$-ABox and $\Gamma=\NI \cap \Sigma$. The \emph{directed unfolding of $\Amc$} into a ditree-shaped ABox $\Amc^{u}$ modulo $\Gamma$ is defined as follows. The individuals of $\Amc^{u}$ are 
the words $w=x_{0}r_{1}\cdots r_{n}x_{n}$ with $r_{1},\ldots,r_{n}$ role names and 
$x_{0},\ldots x_{n}\in \text{ind}(\Amc)$ such that $\{a\}(x_{i})\not\in\Amc$ for any $i\not=0$ and $a\in \Gamma$ and 
$r_{i+1}(x_{i},x_{i+1})\in \Amc$ for all $i<n$. We set $\text{tail}(w)=x_{n}$ and define 
\begin{itemize}
	\item $A(w)\in \Amc^{u}$ if $A(\text{tail}(w))\in \Amc$, for $A\in \NC$;
	\item $r(w,wrx)\in \Amc^{u}$ if $r(\text{tail}(w),x)\in \Amc$ and 
	$r(w,x)\in \Amc^{u}$ if $\{a\}(x)\in \Amc$ for some $a\in \Gamma$ and $r(\text{tail}(w),x)\in \Amc$, for $r\in \NR$;
	\item $\{a\}(x)\in \Amc^{u}$ if $\{a\}(x)\in \Amc$, for $a\in \Gamma$ and $x\in \text{ind}(\Amc)$.
\end{itemize}

\paragraph{Derivation Trees.}
Fix an $\mathcal{ELRO}_{u}$-ontology $\Omc$ in normal form, a $\text{sig}(\Omc)$-ABox $\Amc$, and recall the definition of $\Delta$ and $\Theta$. Let $(a,C)\in \Delta \times \Theta$.
A \emph{derivation tree} for the assertion $(a,C)$ in $\Omc,\Amc$ is a finite $\Delta \times \Theta$-labeled
tree $(T,V)$, where $T$ is a set of nodes and $V : T \to \Delta\times \Theta$ the labeling function, 
such that
\begin{itemize}	
	\item $V(\varepsilon)=(a,C)$;	
	\item if $V(n)=(a,C)$, then (i) $a\in \text{ind}(\Amc)$ and $C=\top$ or (ii) $C(a)\in \Amc$ or (iii) $a\in \NI$ and 
              $C = \{a\}$ or 
	\begin{enumerate}  
		\item $a=C=A$ for a concept name $A$ and $n$ has a successor $n'$ with $V(n') = (b,A)$; or
		\item $a=C=A$ for a concept name $A$ and $n$ has a successor $n'$ such that $V(n')= (b,C')$ and 
		$\Omc\models C' \sqsubseteq \exists u.A$; or
		\item $n$ has successors $n_1,n_{2}$ with $V(n_i) = (a,C_i)$ for $i=1,2$ and 
		and $\Omc\models C_{1} \sqcap C_{2} \sqsubseteq C$; or
		\item $n$ has successors $n_1,n_2,n_3$ with $V(n_1) = (b,C)$,
		$V(n_2) = (a,\{c\})$, and $V(n_3) = (b,\{c\})$; or
		\item the conditions of the rule for RIs discussed in the main paper hold:
                there are role names $r_{2},\ldots,r_{2k-2},r$ and members $a=a_{1},\ldots, a_{2k}$
                of $\Delta$ such that $(a_{2k},C')$ is a label of a successor of $n$, 
                $\Omc\models \exists r.C'\sqsubseteq C$, $\Omc\models r_{2}\circ \cdots \circ r_{2k-2}\sqsubseteq r$,
                and the situation depicted in Figure~\ref{figure:dtree} holds,
                where the ``dotted lines'' stand for `either $a_{i}=a_{i+1}$ or some 
                $(a_{i},\{c\}),(a_{i+1},\{c\})$ with $c\in \NI$ are labels of successors of $n$', and $\hat{r}_i$ 
                stands for `either $r(a_{i},a_{i+1})\in \Amc$ or some $(a_{i},C_{i})$ is a label of a successor 
                of $n$ and $\Omc \models C_{i}\sqsubseteq \exists r_{i}.\{a_{i+1}\}$ if $a_{i+1}\in \NI$ and 
                $\Omc \models C_{i}\sqsubseteq \exists r_{i}.a_{i+1}$ if $a_{i+1}\in \NC$'. Moreover, for all
                $a_{i}\not=a$, $1< i \leq 2k$, there exists a successor of $n$ with label $(a_{i},D)$ for some $D$; or
                \item $n$ has a successor $n'$ with $V(n')=(b,C')$ and $\Omc\models \exists u.C'\sqsubseteq C$.
	\end{enumerate}
\end{itemize}
The purpose of Conditions~1 and 2 is to establish that it follows from $\Omc$ and $\Amc$ that $A$ is not empty. 
In this case $(A,A)$ is derived. The purpose of the remaining rules should be clear.
\begin{example}\label{exm:deriv}
	We use the ontology from Example~\ref{exm:succ}. Recall that 
	\begin{eqnarray*}
		\Omc_{p}	& = & \{r_{i}\circ r_{i} \sqsubseteq r_{i+1} \mid 0\leq i <n\}\cup\\
		&   & \{A \sqsubseteq \exists r_{0}.B, B \sqsubseteq \exists r_{0}.B, \exists r_{n}.B \sqsubseteq A\}
	\end{eqnarray*}	
	Then $\Imc_{\Omc_{p},A}$ is defined by setting 
	\begin{eqnarray*}
		\Delta^{\Imc_{\Omc_{p},A}} & = & \{\rho_{A},y\}\\
		A^{\Imc_{\Omc_{p},A}}   & = & \{\rho_{A}\} \\
		B^{\Imc_{\Omc_p,A}}   & = & \{y\} \\
		r_{i}^{\Imc_{\Omc_{p},A}} & = & \{(\rho_{A},y),(y,y)\}, \text{ for $0\leq i \leq n$.}
	\end{eqnarray*}
	Recall that $\Sigma=\{r_{0},B\}$ and that $\exists r_{0}^{2^{n}}.B$ is an explicit definition of $A$ using $\Sigma$ 
        under $\Omc_{p}$. 
        Consider the ABox $\Amc_{|\Sigma}$ corresponding
	to the $\Sigma$-reduct of $\Imc_{\Omc_{p},A}$. Then a derivation tree $(T,V)$ for $(\rho_{A},A)$
	in $\Omc_{p},\Amc_{|\Sigma}$ is defined by setting $V(\epsilon)=(\rho_{A},A)$ and taking a single successor 
        $n$ of $\epsilon$ with $V(n)=(y,B)$. In the notation of Rule 5, we have $a_{1}=a_{2}=\rho_{A}$ 
        and $a_{3}=\cdots=a_{2^{n}}=y$. We use that $\Omc_{p} \models r_{0}^{2^{n}}\sqsubseteq r_{n}$ and 
        $\Omc_{p}\models \exists r_{n}.B\sqsubseteq A$. 
\end{example}

We next show Part~1 of Lemma~\ref{lem:deriv0}.

\medskip
\noindent
{\bf Proof of Part 1 of Lemma~\ref{lem:deriv0}.}
Let $\Omc$ be an $\mathcal{ELRO}_{u}$-ontology in normal form and 
$\Amc$ a finite $\text{sig}(\Omc)$-ABox. Assume $(x,A)$ with $x\in \text{ind}(\Amc)$ and $A\in \Theta$ is given.
It is straightforward to show by induction that if there is a derivation tree for $(x,A)$ in $\Omc,\Amc$, then $\Omc,\Amc\models A(x)$. We construct a sequence of ABoxes $\Amc_0,\Amc_1,\ldots$ as follows.
Define $\A_0$ as the union of $\Amc$ and all assertions $\{a\}(a)$ with $a$
an individual name in $\Omc$ and $\top(x)$ with $x\in \text{ind}(\Amc)$. Let $\Amc_{i+1}$ be obtained from $\Amc_i$ by
	applying one of the following rules: 
	\begin{enumerate}
		\item if $A(b)\in \Amc_{i}$, then add $A(A)$ to $\Amc_{i}$;
		\item if $C'(b)\in \Amc_{i}$ and $\Omc\models C' \sqsubseteq \exists u.A$, then add $A(A)$ to $\Amc_{i}$;
		\item if $C_{1}(a),C_{2}(a) \in \Amc_{i}$ and $\Omc\models C_{1} \sqcap C_{2} \sqsubseteq C$, then add $C(a)$ to $\Amc_{i}$;
		\item if $C(b), \{c\}(a), \{c\}(b)\in \Amc_{i}$, then add $C(a)$ to $\Amc_{i}$;
		\item if there is a sequence  $a_{1},\ldots,a_{2k}$ of elements of $\Delta$ and a sequence 
                      $r_{2},r_{4},\ldots,r_{2k-2}$ of role names such that $a=a_{1}$ and for every $a_{2j+1}$ either 
		$a_{2j+1}=a_{2j+2}$ or there is $c$ with $\{c\}(a_{2j+1}), \{c\}(a_{2j+2})\in \Amc_{i}$ 
		such that for every $a_{2j}$:
		\begin{itemize}
			\item $r_{2j}(a_{2j},a_{2j+1})\in \Amc$; 
			or
			\item $a_{2j+1}\in \NI \cap \text{sig}(\Omc)$ and there exists 
                              $C_{2j} \in (\NC\cup \NI)\cap \text{sig}(\Omc)$
			such that $C_{2j}(a_{2j})\in \Amc_{i}$ and 
			$\Omc\models C_{2j} \sqsubseteq \exists r_{2j}.\{a_{2j+1}\}$; or
			\item $a_{2j+1}\in \NC \cap \text{sig}(\Omc)$ and there exists $C_{2j} \in (\NC\cup \NI)\cap \text{sig}(\Omc)$ 
			such that $C_{2j}(a_{2j})\in \Amc_{i}$ and 
			$\Omc\models C_{2j} \sqsubseteq \exists r_{2j}.a_{2j+1}$
		\end{itemize}
		and there exist $C'\in (\NC\cup \NI)\cap \text{sig}(\Omc)$ and a role name $r$ 
		such that $C'(a_{2k})\in \Amc_{i}$, $\Omc\models \exists r.C' \sqsubseteq A$, and 
		$\Omc\models r_{2}\circ r_{4} \circ \ldots \circ r_{2k-2}\sqsubseteq r$, then add $A(a)$ to $\Amc_{i}$.
               \item if $C'(b)\in \Amc_{i}$ and $\Omc\models \exists u.C'\sqsubseteq C$, then add $C(a)$ to $\Amc_{i}$.	
\end{enumerate}
	Note that the sequence is finite, and denote by $\Amc^*$ the final ABox.
	
	\smallskip\noindent\textit{Claim.} There is a model $\Imc,v$ of
	$\Amc^*$ and \Omc such that for all $x\in \text{ind}(\Amc)$ and $A\in\NC$,
	$v(x)^\Imc \in A^\Imc$ implies $A(x)\in \Amc^*$.
	
	\smallskip\noindent\textit{Proof of the Claim.}
	For all $a,b \in \mn{ind}(\Amc^*)$,
	we write $a \sim b$ if $a=b$ or $\{c\}(a), \{c\}(b) \in \Amc^*$ for some $c$.
	Notice that due to Rule 4, $a \sim b$ implies $C(a) \in \Amc^*$ if and only
	if $C(b) \in \Amc^*$. It follows that $\sim$ is an equivalence relation.
        We let $[a]$ denote the equivalence class of~$a$.
	Start with an interpretation $\Imc_0$ defined by:
	\begin{align*}
		\Delta^{\Imc_0} & = \mn{ind}(\Amc^*)/{\sim} \\
		A^{\Imc_0} & = \{[a]\mid A(a)\in \Amc^*\} \\
		a^{\Imc_0} & = \{[a]\} \\
		r^{\Imc_0} & = \{([a],[b])\mid \exists a' \in [a], b' \in [b].\ 
		r(a',b')\in \Amc^*\} \, .
	\end{align*}
	
	By definition, $\Imc_{0}$ satisfies all CIs in $\Omc$ that do not involve role names or the universal role. 
        We next extend $\Imc_{0}$ by adding pairs of the form $([a],[b])$ with $b\in \NC \cup \NI$ to the 
        interpretation of role names.
	In detail, if $[a]\in \Delta^{\Imc_{0}}$ and there exist $C\in \NC\cup \NI$ with $a\in C^{\Imc_{0}}$ and 
        $c\in \NI$ with $c\in [b]$ such that $\Omc\models C \sqsubseteq \exists r.\{c\}$, 
        then add $([a],[b])$ to $r^{\Imc_{0}}$. 
        Also, if $[a]\in \Delta^{\Imc_{0}}$ and there exist $C\in \NC\cup \NI$ with $a\in C^{\Imc_{0}}$ and 
        $A\in \NC$ with $A\in [b]$ such that $\Omc\models C \sqsubseteq \exists r.A$, then add $([a],[b])$ to $r^{\Imc_{0}}$. 
        Finally, add any pair $([a],[b])$ to $r^{\Imc_{0}}$ if there exists an RI
	$r_{1} \circ \cdots \circ r_{n} \sqsubseteq r$ that follows from $\Omc$ such that 
	$([a],[b])$ is in relation $r_{1} \circ \cdots \circ r_{n}$ under the updated interpretations of $r_{1},\ldots,r_{n}$. 
        This defines an interpretation $\Imc$. By Rule~2 all CIs of the form $A \sqsubseteq \exists r.B$ are satisfied in 
        $\Imc$. By definition, all RIs in $\Omc$ are satisfied in $\Imc$. By Rules~5 and 6, all CIs of the form 
        $\exists r.B \sqsubseteq A$ are satisfied as well. This finishes the proof of the claim.
	
	\medskip Now suppose $\Omc,\Amc\models A_0(x_0)$.
	By the Claim, we have $A_0(x_0)\in \Amc^*$.
	Since the six rules to construct $\Amc_0,\Amc_1,\ldots$ are in
	one-to-one correspondence with Conditions~(1)--(6) from the
	definition of derivation trees, we can inductively construct a
	derivation tree for $A_0(x_0)$ in \Amc w.r.t.\ \Omc.
      
        The remaining claims made in Part~1 of Lemma~\ref{lem:deriv0} have been shown in the main paper already.
\qed

We next come to Part~2 of Lemma~\ref{lem:deriv0}. The following example illustrates
how one can construct from a derivation tree of $A(x)$ in $\Omc,\Amc$ a derivation tree in 
$\Omc,\Amc^{u}$ with $\Amc^{u}$ the directed unfolding of $\Amc$. The derivation tree has the same depth 
but the outdegree might be exponential. 

\begin{example}
	Recall the ontology $\Omc_{p}$ and concept name $A$ from Example~\ref{exm:deriv}. We consider
        the $\Sigma$-reduct $\Amc_{|\Sigma}$ of the ABox $\Amc$ corresponding to the canonical model $\Imc_{\Omc_{p},A}$.
        It is defined by $\Amc_{|\Sigma}=\{r_{0}(\rho_{A},y),r_{0}(y,y),B(y)\}$. The directed unfolding 
	$\Amc_{|\Sigma}^{u}$ has individuals 
$$
\rho_{A}, \quad \rho_{A}r_{0}y, \quad \rho_{A}r_{0}yr_{0}y,\quad \ldots
$$
and the assertions
$$
B(\rho_{A}r_{0}y), \quad B(\rho_{A}xr_{0}yr_{0}y),\quad \ldots
$$
$$
r_{0}(\rho_{A},\rho_{A}r_{0}y),\quad r_{0}(\rho_{A}r_{0}y,\rho_{A}r_{0}yr_{0}y),\quad \ldots
$$
In a derivation tree $(T',V')$ for $A(\rho_{0})$ in $\Omc_{p},\Amc_{|\Sigma}$ we require 
that $\epsilon$ has $2^{n}$ successors labeled with:
	$$
	(\rho_{A}r_{0}y,B),\quad (\rho_{A}r_{0}yr_{0}y,B),  
	\ldots, \quad (\rho_{A}(r_{0}y)^{2n},B).
	$$
\end{example}
We now give the general construction of the derivation tree in the directed unfolding from a derivation tree
in the original ABox.

\medskip
\noindent
{\bf Proof of Part 2 of Lemma~\ref{lem:deriv0}.}
Assume that $(T,V)$ is a derivation tree for $A(x)$ in $\Omc,\Amc$ of at most exponential size.
We obtain a very similar derivation tree $(T',V')$ for $A(x)$ in $\Omc,\Amc^{u}$ with $\Amc^{u}$ the directed unfolding of $\Amc$ modulo $\Sigma=\text{sig}(\Amc) \cap \NI$. In fact, with the exception of 
Condition~5, the construction is identical. For Condition~5, one potentially has to introduce "copies"
of the nodes in $T$ which correspond to the fresh individuals introduced in the unfolded ABox.

In the following construction of $(T',V')$ the following holds: if the label of $n$ in $(T,V)$ is $(a,C)$, 
then the label of copies $n'$ of $n$ in $(T',V')$ takes the form $(w,C)$ with $\text{tail}(w)=a$. Moreover,
if $\{b\}(a)\in \Amc$ for some $b\in \Sigma\cap \NI$ or $a\in \NI \cup \NC$, then the label of $n'$ 
is identical to the label of $n$. Note $V'$ is a mapping form $T'$ to $\Delta'\times \Theta$ with
$$
\Delta'=\text{ind}(\Amc^{u}) \cup ((\NC \cup \NI) \cap \text{sig}(\Omc))
$$
In detail, we define $(T',V')$ as follows from $(T,V)$, starting with the 
root by setting $V'(\epsilon):=V(\epsilon)=(x,A)$. 
	
Assume inductively that $m$ is a copy of $n$, $V(n)=(a,C)$, and $V'(m)=(w,C)$. 
To define the successors of $m$ and their labelings we consider the possible derivation steps for 
$(a,C)$ in $\Omc,\Amc$: (i) if $a\in \text{ind}(\Amc)$ and $C=\top$, then $w\in \text{ind}(\Amc^{u})$ and $C=\top$;
(ii) if $C(a)\in \Amc$, then $C(w)\in \Amc^{u}$;
(iii) if $a\in \NI$ and $C=\{a\}$, then $V'(m)=(a,\{a\})$. We next consider the cases 1 to 6:    
	\begin{enumerate}  
		\item $a=C=A$ for a concept name $A$ and $n$ has a successor $n'$ with $V(n') = (b,A)$: 
                take a copy $m'$ of $n'$ as the only successor of $m$ and set $V'(m')=(b,A)$.
		\item $a=C=A$ for a concept name $A$ and $n$ has a successor $n'$ such that $V(n')= (b,C')$ and 
		$\Omc\models C' \sqsubseteq \exists u.A$: take a copy $m'$ of $n'$ as the only successor of 
                $m$ and set $V'(m')=(b,C)$. 
		\item $n$ has successors $n_1,n_{2}$ with $V(n_i) = (a,C_i)$ and 
		and $\Omc\models C_{1} \sqcap C_{2} \sqsubseteq C$: take copies $m_{1},m_{2}$ of $n_{1},n_{2}$ 
		as the successors of $m$ and set $V'(m_i) = (w,C_i)$.
		\item $n$ has successors $n_1,n_2,n_3$ with $V(n_1) = (b,C)$,
		$V(n_2) = (a,\{c\})$, and $V(n_3) = (b,\{c\})$: 
		take copies $m_{1},m_{2},m_{3}$ of $n_{1},n_{2},n_{3}$ as successors of $m$ and set $V(m_1) = (b,C)$,
		$V(m_2) = (w,\{c\})$, and $V(m_3) = (b,\{c\})$.
		\item Suppose that $n$ has successors such that the conditions of Point~5 for derivation trees 
                hold for $r_{2},\ldots,r_{2k-2},r$ and members $a=a_{1},\ldots,a_{2k}$ of $\Delta$.
                We define the new members $b_{1},\ldots,b_{2k}$ of $\Delta'$ and relevant successors of $m$ with labeling 
                by induction. We set $b_{1}=w$.
                Assume that $b_{2i+1}$ has been defined and $(b_{2i+1},D)$ is the label of a copy of a successor of 
                $n$ with label $(a_{2i+1},D)$.

                Case 1. $a_{2i+1}=a_{2i+2}$. Then we set $b_{2i+2}:=b_{2i+1}$.

                Case 2. There exists $c\in \NI$ and successors $n_{1},n_{2}$ of $n$ with $V(n_{1})=(a_{2i+1},\{c\})$ 
                        and $V(n_{2})=(a_{2i+2},\{c\})$. Then we let $b_{2i+2}:=a_{2i+2}$ and we introduce copies 
                        $m_{1},m_{2}$ of $n_{1},n_{2}$ with $V'(m_{1})=(b_{2i+1},\{c\})$ and $V'(m_{2})=(a_{2i+2},\{c\})$.  

                Now assume that $b_{2i}$ has been defined and $(b_{2i},D)$ is the label of a copy of a successor 
                of $n$ with label $(a_{2i},D)$.

                Case 1. $r_{2i}(a_{2i},a_{2i+1})\in \Amc$ and $V(n')=(a_{2i+1},D')$ for some successor $n'$ of $n$. 
                If $\{b\}(a_{2i+1})\in \Amc$ for some $b\in \NI$, 
                then we set $b_{2i+1}=a_{2i+1}$ and introduce a copy $m'$ of $n$
                and set $V'(m')=(a_{2i+1},D')$. Observe that $r_{2i}(b_{2i},a_{2i+1})\in \Amc^{u}$.
                Otherwise (if no $b$ with $\{b\}(a_{2i+1})\in \Amc$ exists), we set
                $b_{2i+1}=b_{2i}r_{2i}a_{2i+1}$ and introduce a copy $m'$ of $n'$
                and set $V'(m')=(b_{2i+1},D')$.

                Case 2. $a_{2i+1}\in (\NI \cup \NC) \cap \text{sig}(\Omc)$, $V(n_{1})=(a_{2i+1},D')$ for some successor 
                $n_{1}$ of $n$, and $V(n_{2})=(a_{2i},F)$ for a successor $n_{2}$ of $n$ and $\Omc\models F \sqsubseteq \exists r_{2i}.a_{2i+1}$ (if $a_{2i+1}\in \NC$) or $\Omc\models F \sqsubseteq \exists r_{2i}.\{a_{2i+1}\}$ (if $a_{2i+1}\in \NI$), respectively.
                Then we introduce copies $m_{1},m_{2}$ of $n_{1},n_{2}$ and set 
                $b_{2i+1}=a_{2i+1}$, $V'(m_{1})=(b_{2i+1},D')$, and $V'(m_{2})=(b_{2i},F)$. 
           \item $n$ has a successor $n'$ with $V(n')=(b,C')$ and $\Omc\models \exists u.C'\sqsubseteq C$: then
                introduce a copy $m'$ of $n'$ and set $V'(m')=(b,C)$.  
	\end{enumerate}
	Then $(T',V')$ is a derivation tree for $A(x)$ in $\Omc,\Amc^{u}$ satisfying the 
        conditions of the lemma.
\qed

The proof of ``2. $\Rightarrow$ 3.'' of Theorem~\ref{thm:critele} is now as sketched in the main paper. Note also that we can construct $\Amc$ in exponential time since we can construct the derivation tree for $B$ in $\Amc_{\Omc_{1}\cup \Omc_{2},A}^{\Sigma}$ in exponential time, then lift it to a derivation tree in its unfolding in exponential time, and from that derivation tree obtain the individuals in the ABox $\Amc$ in exponential time.

A proof of the statement of Theorem~\ref{thm:critele} for interpolants without the universal role is obtained from the proof above in a straightforward way.

\bigskip

We conclude this section with a deferred proof of Theorem~\ref{th:safety}.
        \begin{figure*}[t]
		\begin{eqnarray*}
			\Delta^{\Jmc_{i+1}} & = & \Delta^{\Jmc_{i}},\\
			A^{\Jmc_{i+1}} & = & A^{\Jmc_{i}}, \textrm{for all $A \in \NC$},\\
			r^{\Jmc_{i+1}} & = & r^{\Jmc_{i}}\cup 
			\left\{(d_1,d_{n+1})\left|
			\begin{array}{l}
				r_1\circ\dots\circ r_n\sqsubseteq r\in\Omc_1\cup\Omc_2\\
				\{d_1,\dots,d_{n+1}\}\subseteq \Delta^{\Jmc_i},
				(d_1,d_{n+1})\notin r^{\Jmc_i}\\ 
				(d_k,d_{k+1})\in r_k^{\Jmc_i} \textrm{ for all $1\leq k\leq n$}
			\end{array}
			\right. \right\}
		\end{eqnarray*}
        \caption{Definition of $\Jmc_{i+1}$.\label{fig:jmc}}
        \end{figure*}

\medskip
\noindent
{\bf Theorem~\ref{th:safety}.}
{\em 
 	Let $\Omc_{1}, \Omc_{2}$ be $\mathcal{EL}$-ontologies with RIs,
 	$C_{1},C_{2}$ be $\mathcal{EL}$-concepts, and set $\Sigma=\text{sig}(\Omc_{1},C_1) \cap \text{sig}(\Omc_{2},C_2)$.
 	Assume that the set of RIs in $\Omc_1\cup\Omc_2$ is safe for $\Sigma$ and 
 	$\Omc_1\cup\Omc_2\models C_1\sqsubseteq C_2$.
 	Then an $\mathcal{EL}$-interpolant for $C_1\sqsubseteq C_2$ under
 	$\Omc_1$, $\Omc_2$ exists.
}
\begin{proof}
	For convenience of notation, we assume w.l.o.g., by Lemma~\ref{lem:normal}, that 
    $\Omc_1$ and $\Omc_2$ are in normal form,
    $A\in\sig(\Omc_1)$,
	$B\in\sig(\Omc_2)$ and $\{A,B\}\cap \Sigma=\emptyset$.
	Suppose for a proof by contradiction that $\Omc_1\cup\Omc_2\models
	A\sqsubseteq B$ but there exists no $\mathcal{EL}$-interpolant for
	$A\sqsubseteq B$. Then $\Omc_{1}\cup \Omc_{2},\Amc_{\Omc_{1}\cup
		\Omc_{2},A}^{\downarrow\Sigma}\not\models B(\rho_{A})$.
	Moreover, since the language under consideration contains neither
	nominals nor the universal role, this strengthens to
	$\Omc_{1}\cup \Omc_{2},\Amc_{\Omc_{1}\cup \Omc_{2},A}^{\Sigma}
	\not\models B(\rho_{A})$.

	Let ${\Jmc_0}$ be the canonical model of $\Omc_1\cup\Omc_2$ and
	$\Amc_{\Omc_{1}\cup \Omc_{2},A}^{\Sigma}$.  
	In what follows, we identify the domain of $\Imc_{\Omc_{1}\cup
		\Omc_{2},A}^{\Sigma}$ and individuals of $\Amc_{\Omc_{1}\cup
		\Omc_{2},A}^{\Sigma}$, and consider both to be subsets of the domain of
	${\Jmc_0}$.
	By the properties of the
	canonical model, we then have $\rho_A\notin B^{\Jmc_0}$.  
	Furthermore,
	as $\Imc_{\Omc_{1}\cup \Omc_{2},A}$ is a model for both $\Omc_1\cup\Omc_2$ and
	$\Amc^\Sigma_{\Omc_{1}\cup \Omc_{2},A}$,
	there exists a $\sig(\Omc_1\cup\Omc_2)$-simulation $S$ between $\Jmc_{0}$ and
	$\Imc_{\Omc_{1}\cup \Omc_{2},A}$ such that $(x,x)\in S$ 
	for all $x\in\Delta^{\Imc_{\Omc_{1}\cup \Omc_{2},A}}$.

	Consider an interpretation $\Jmc_1$ defined as follows:
	\begin{eqnarray*}
		\Delta^{\Jmc_1} & = & \Delta^{\Jmc_0},\\
		P^{\Jmc_1} & = & P^{\Jmc_0}\cup P^{\Imc_{\Omc_{1}\cup \Omc_{2},A}}, 
		\textrm{for all $P \in (\sig(O_1)\setminus{\Sigma})$},\\
		P^{\Jmc_1} & = & P^{\Jmc_0}, \textrm{for all $P \notin
			(\sig(O_1)\setminus{\Sigma})$},
	\end{eqnarray*}
	where $P$ is a concept or role name. 
	If $\Jmc_1\models \Omc_1\cup\Omc_2$ we immediately derive a contradiction as we
	then have $\rho_A\in\A^{\Jmc_1}$ and $\rho_A\notin B^{\Jmc_1}$, contradicting
	$\Omc_1\cup\Omc_2\models A\sqsubseteq B$.

	\begin{itemize}
        \item If $\Omc_1\cup\Omc_2$ does not contain RIs, as ${\Jmc_0}$ and
        $\Jmc_1$ are identical on all elements except
        $\Delta^{\Imc_{\Omc_{1}\cup \Omc_{2},A}}$, for all
		$x\in\Delta^{\Imc_{\Omc_{1}\cup \Omc_{2},A}}$
		the relation $S$ is a $\sig(\Omc_1\cup\Omc_2)$-simulation 
		between 
		$\Jmc_1$ and $\Imc_{\Omc_{1}\cup \Omc_{2},A}$. Conversely, the embedding of
		$\Imc_{\Omc_{1}\cup \Omc_{2},A}$ into $\Jmc_1$ generates a simulation, that is 
		$
		(\Imc_{\Omc_{1}\cup \Omc_{2},A}, x)
		\preceq_{\mathcal{EL},\sig(\Omc_1)}
		(\Jmc_1, x)
		$ for all $x\in\Delta^{\Imc_{\Omc_{1}\cup \Omc_{2},A}}$.
		By Lemma~\ref{lem:simulationel}, for any $\sig(O_1)$-$\mathcal{EL}$-concept $C$
		and for all $x\in\Delta^{\Imc_{\Omc_{1}\cup \Omc_{2},A}}$ we have $x\in
		C^{\Jmc_1}$ if, and only if $x\in C^{\Imc_{\Omc_{1}\cup \Omc_{2},A}}$.  Thus,
		$\Jmc_1$ is a model of CIs in $\Omc_1$.
		By construction $\Jmc_1\models \Omc_2$.
		
		\item 
        Suppose that $\Omc_1\cup\Omc_2$ contains RIs.
        Since the interpretation $\Jmc_1$ may not satisfy some RIs, we consider
        a sequence of interpretations $\Jmc_i$ obtained by extending the
        interpretations of roles in
		$\Jmc_1$ to satisfy RIs.  
        We give the construction of 
        $\Jmc_{i+1}$, for $i\geq 1$, in Figure~\ref{fig:jmc}.

		A simple inductive argument shows that 
		by the safety condition and the fact that $(d_1,d_{n+1})\notin
		r^{\Jmc_i}$ we have that $\{r_1,\dots,r_n, r\}\subseteq\sig(\Omc_1)$. 
		
		Furthermore, we prove by induction that
		the relation $S$  is
		a $\sig(\Omc_1\cup\Omc_2)$-simulation between
		$\Jmc_{i+1}$ and $\Imc_{\Omc_{1}\cup \Omc_{2},A}$.
		For $i=1$ this has been established above. 
		For the induction stop it suffices to consider $r$-successors of $d_1$ in
		$\Jmc_{i+1}$, where $r$ is from the definition of $\Jmc_{i+1}$ above.
		By the induction hypothesis, $S$ is a a $\sig(\Omc_1\cup\Omc_2)$-simulation
		between $\Jmc_{i}$ and $\Imc_{\Omc_{1}\cup \Omc_{2},A}$.
		Then there exist
		$\{v_2,\dots,v_{n+1}\}\subseteq\Delta^{\Imc_{\Omc_{1}\cup \Omc_{2},A}}$ with
		$(d_{j+1},v_{j+1})\in S$ and $(v_1,v_{n+1})\in {r_i}^{\Imc_{\Omc_{1}\cup \Omc_{2},A}}$ 
		for $j\in\{1,\dots,n\}$.
		As ${\Imc_{\Omc_{1}\cup \Omc_{2},A}}$ is a model of $\Omc_1$, we have 
		$(v_1,v_{n+1})\in r^{\Imc_{\Omc_{1}\cup \Omc_{2},A}}$ and 
		$(d_{n+1},v_{n+1})\in S$ as required.

		As $\mathcal{EL}$ canonical models defined in this paper are finite, there exists
		$N>0$ such that for all $i>N$, $\Jmc_i = \Jmc_N$. 
		It can be seen that $\Jmc_N$ satisfies all RIs in $\Omc_1\cup\Omc_2$ and
		the satisfaction of CIs is proved similarly to the case above.  Then $\Jmc_N$
		is a model of $\Omc_1\cup\Omc_2$ with $\rho_A\in\A^{\Jmc_{N}}$ and
		$\rho_A\notin B^{\Jmc_{N}}$, contradicting $\Omc_1\cup\Omc_2\models
		A\sqsubseteq B$.
		
	\end{itemize}
\end{proof}

\section{Proofs for Section~\ref{sec:eli}}
\label{sec:proofs_eli}
The section is organized as follows. We first introduce canonical models
and derivation trees for $\mathcal{ELIO}_{u}$. We then give the automata based 
proof of the \ExpTime upper bound for interpolant existence. We then show the double exponential lower bound on the size of explicit definitions, the implication 
``2. $\Rightarrow$ 3.'' of Theorem~\ref{thm:critele2}, and that interpolants can be computed in double exponential time.

\paragraph{Canonical Models.}
Assume $\Omc$ is an $\mathcal{ELIO}_{u}$ ontology in normal form
and $A$ a concept name with $A \in \text{sig}(\Omc)$.
We introduce the canonical model $\Imc_{\Omc,A}$.
Let $\text{sub}(\Omc)$ denote the set of subconcepts of concepts in $\Omc$,
and denote by  $\text{sub}^{\exists}(\Omc)$ the set of $\exists r.\{a\}$ with 
$r$ or $r^{-}$ a role name in $\text{sig}(\Omc)$ and $a\in \text{sig}(\Omc)$.
We may assume that $\exists u.A \in \text{sub}(\Omc)$.
An \emph{$\Omc$-type} is a subset $\tau$ of $\text{sub}(\Omc)\cup \text{sub}^{\exists}(\Omc)$ 
such that $\Omc \models \bigsqcap_{C \in \tau} C \sqsubseteq C'$ implies
$C' \in \tau$ for all concepts $C'\in \text{sub}(\Omc)\cup \text{sub}^{\exists}(\Omc)$.
We sometimes identify $\tau$ and $\bigsqcap_{C \in \tau} C$.
For a role $r$, we write $\tau_1 \rightsquigarrow_r \tau_2$ if $\tau_2$
is a maximal (w.r.t.\ inclusion) $\Omc$-type such that
$\Omc \models \tau_1 \sqsubseteq \exists r. \tau_2$.
Note that the set of all $\Omc$-types and relation $\rightsquigarrow_r$ 
can be computed in exponential time.

For any concept name $B$, $\tau_{B}$ denotes the minimal $\Omc$-type containing 
$B$ and $\exists u.A$. Similarly, for any individual $a$, $\tau_{a}$ denotes
the minimal $\Omc$-type containing $\{a\}$ and $\exists u.A$.
Let $S = S_{C} \cup S_{N}$ with $
S_{C}= \{\tau_{B} \mid \Omc \models A \sqsubseteq \exists u.B\}$
and $S_{N} = \{ \tau_{a} \mid a \in \NI \cap \sig(\Omc)\}$.
The \emph{canonical model} $\Imc_{\Omc,A}$ of $\Omc$ and $A$ is defined
as follows:
\begin{eqnarray*}
\Delta^{\Imc_{\Omc,A}} 
& =  & \{ \tau_0 r_1 \tau_1 \cdots r_n \tau_n \mid 
	\tau_0 \in S, \tau_{1},\ldots,\tau_{n}\not\in S_{N}, \\
&  & r_{1},\ldots,r_{n}\in \NR\cup \NR^{-}, \tau_i \rightsquigarrow_{r_{i+1}} \tau_{i+1}\} \\
	a^{\Imc_{\Omc,A}} & = & \tau_{a} \\
	B^{\Imc_{\Omc,A}}
	& = & \{ w \mid w \in \Delta^{\Imc_{\Omc,A}}, B \in \text{tail}(w) \} \\
	r^{\Imc_{\Omc,A}}
	& = & \{ (w, w r \tau) \mid w,w r \tau \in \Delta^{\Imc_{\Omc,A_1}} \}
	\cup \\
        &   & \{ (w {r^-} \tau, w) \mid w,w r^- \tau \in \Delta^{\Imc_{\Omc,A_1}} \} 
        \cup \\
       &  & \{ r(w,\tau_{a}) \mid \exists r.\{a\}\in \text{tail}(w)\} \cup \\ 
       &  & \{ r(\tau_{a},w) \mid \exists r^{-}.\{a\}\in \text{tail}(w)\}         
\end{eqnarray*}
We also use $\rho_A$ to denote $\tau_A$.
The following properties of canonical models can be proved in a standard~way.

\begin{lemma}
	For all $\mathcal{ELIO}_u$-ontologies $\Omc$ in normal form and concept
	names $A \in \text{sig}(\Omc)$:
	\begin{enumerate}
		\item $\Imc_{\Omc,A}$ is a model of $\Omc$;
		\item for every model $\Jmc$ of $\Omc$ and any
                  $d \in \Delta^\Jmc$ with $d \in A^\Jmc$,
                   $(\Imc_{\Omc,A},\rho_{A}) \preceq_{\mathcal{ELIO}_u,\Sigma} (\Jmc,d)$;
		\item for every $\mathcal{ELIO}_u(\sig(\Omc))$-concept $C$, 
                $\Omc \models A \sqsubseteq C$
		if and only if $\rho_{A} \in C^{\Imc_{\Omc,A}}$.
	\end{enumerate}
\end{lemma}
We use $\Amc_{\Omc,A}$ to denote the ABox associated with the canonical model 
$\Imc_{\Omc,A}$, and $\Amc_{\Omc,A}^\Sigma$ its $\Sigma$-reduct.
We denote the individuals $x_{\tau_{a}}$ and $x_{\tau_{B}}$ by $x_{a}$ and 
$x_{B}$, respectively and observe that $x_{a}=x_{b}$ iff $\Omc \models \{a\} \sqcap \exists u.A
\sqsubseteq \{b\}$ and $x_{B}=x_{a}$ iff $\Omc \models B \sqcap \exists u.A
\sqsubseteq \{a\}$.

\paragraph{Undirected Unfolding of an ABox.}
We give a precise definition of the undirected unfolding of an ABox.
Let $\Amc$ be a $\Sigma$-ABox and $\Gamma = \NI \cap \Sigma$.
The \emph{undirected unfolding of $\Amc$} into a tree-shaped ABox 
$\Amc^{\ast}$ modulo $\Gamma$ is defined as follows.
The individuals of $\Amc^{\ast}$ are 
the set of words $w=x_{0}r_{1}\cdots r_{n}x_{n}$ with $r_{1},\ldots,r_{n}$ roles and $x_{0},\ldots x_{n}\in \text{ind}(\Amc)$ such that $\{a\}(x_{i})\not\in \Amc$ for any $i\not=0$ and $a \in \Gamma$, and $r_{i+1}(x_{i},x_{i+1})\in \Amc$ if $r_{i+1}$ is a role name and
$r_{i+1}^{-}(x_{i+1},x_{i})\in \Amc$ if $r_{i+1}$ is an inverse role, for all $i<n$.
We set $\text{tail}(w)=x_{n}$ and let 
\begin{itemize}
	\item $A(w)\in \Amc^{\ast}$ if $A(\text{tail}(w))\in \Amc$, for $A\in \NC$;
	\item $r(w,wrx)\in \Amc^{\ast}$ if $r(\text{tail}(w),x)\in \Amc$ and
	$r(w,x)\in \Amc^{\ast}$ if $\{a\}(x)\in \Amc$ for some $a\in \Gamma$ and $r(\text{tail}(w),x)\in \Amc$, for $r\in \NR$;
	\item $r(wr^-x,w)\in \Amc^{\ast}$ if $r(x,\text{tail}(w))\in \Amc$ and
	$r(x,w)\in \Amc^{\ast}$ if $\{a\}(x)\in \Amc$ for some $a\in \Gamma$ and $r(x,\text{tail}(w))\in \Amc$, for $r\in \NR$;
	\item $\{a\}(x)\in \Amc^{\ast}$ if $\{a\}(x)\in \Amc$, for $a\in \Gamma$ and $x \in \text{ind}(\Amc)$.
\end{itemize}

\paragraph{Derivation Trees.}
Fix an $\mathcal{ELIO}_{u}$-ontology $\Omc$ in normal form and an ABox~\Amc,
$x_0 \in \mn{ind}(\Amc)$ and $A_0 \in \NC$.
Let $\Theta_1 = \mn{ind}(\Amc) \cup (\NI \cap \sig(\Omc))$, and
$\Theta_2 = \NC \cap \sig(\Omc) \cup \{\{a\} \mid a \in \NI \cap \sig(\Omc)\} \cup \{\exists u. A \mid A \in \NC \cap \sig(\Omc)\}$.
A \emph{derivation tree} for the assertion $A_0(x_0)$ in $\Omc,\Amc$
is a finite $\Theta_1 \times \Theta_2$-labeled tree $(T,V)$, where $T$ is a set
of nodes and $V : T \to \Theta_1 \times \Theta_2$ the labeling function, such
that:
\begin{itemize}
	
	\item $V(\varepsilon)=(x_0,A_0)$;
	
	\item If $V(n) = (x,C)$ with $x \in \mn{ind}(\Amc)$, then
          $C(x) \in \Amc$ or 
	$\Omc\models \top \sqsubseteq C$ or
	\begin{enumerate}
		\item\label{dt1} $n$ has successors $n_1,\ldots,n_k$, $k\geq 1$ with
		$V(n_i) = (a_i,C_i)$,
		such that $a_i = x$ or $a_i \in \NI \cap \sig(\Omc)$ for all $i$, and
		defining $C'_i = C_i$ if $a_i = x$, and
		$C'_i = \exists u. (\{a_i\} \sqcap C_i)$ otherwise,
		we have $\Omc\models C'_1\sqcap \ldots\sqcap C'_k\sqsubseteq C$; or
		
		\item\label{dt2} $C = \exists u.A$ and $n$ has a single successor $n'$ with
		$V(n')=(y,\exists u.A)$; or
		
		\item\label{dt3} $n$ has a single successor $n'$ with $V(n')=(y,A)$ such
		that $r(x,y)\in \Amc$ and $\Omc \models \exists r.A\sqsubseteq C$
		(where $r$ is a role name or an inverse role).
	\end{enumerate}
	
	\item If $V(n) = (a,C)$ with $a \in \NI \cap \sig(\Omc)$, then
	$C = \{a\}$ or:
	\begin{enumerate}[resume]    
		\item\label{dt4} There exists $x \in \mn{ind}(\Amc)$ such that $n$ has successors
		$n_1,\ldots,n_k$, $k\geq 1$ with $V(n_i)=(a_i,C_i)$ and $a_i = x$ or
		$a_i \in \NI \cap \sig(\Omc)$ for all $i$,
		and, defining $C'_i = C_i$ if $a_i = x$, and
		$C'_i = \exists u. (\{a_i\} \sqcap C_i)$ otherwise, we have
		$\Omc\models C'_1\sqcap \ldots\sqcap C'_k\sqsubseteq \exists u.(\{a\} \sqcap C)$.
	\end{enumerate}
\end{itemize}
Note that a special case of rule 1 is when $n$ has two successors labeled
$(x,\{a\})$ and $(a,C)$, and a special case of rule 4 is when $n$ has two
successors labeled $(x,\{a\})$ and $(x,C)$.

\medskip

We now prove the analogue of Lemma~\ref{lem:deriv0} for
$\mathcal{ELIO}_u$, except not considering the size of derivation trees.

\begin{lemma}\label{lem:derivELI}
Let $\Omc$ be an $\mathcal{ELIO}_{u}$-ontology in normal form and 
$\Amc$ a finite $\text{sig}(\Omc)$-ABox. Then 
\begin{enumerate}
\item $\Omc,\Amc \models A_0(x_0)$ if and only if there is a derivation tree for $A_0(x_0)$ in $\Omc,\Amc$.
\item If $(T,V)$ is a derivation tree for $A_0(x_0)$ in $\Omc,\Amc$, then one
  can construct a derivation tree $(T',V')$ for $A_0(x_0)$ in $\Omc, \Amc^{*}$,
  with $\Amc^*$ the undirected unfolding of $\Amc$, and such that $T = T'$.
\end{enumerate}
\end{lemma}

\begin{proof}
  We start with the proof of Part 1.
	$(\Leftarrow)$ is straightforward.	
	For $(\Rightarrow)$,
	we construct a sequence of ABoxes $\Amc_0,\Amc_1,\ldots$ generalized
	with assertions of the form $(\exists u.A)(x)$.
	Take $\Amc_0 = \A \cup \{ \{a\}(x_a) \mid a \in \NI \cap \sig(\Omc)\}$
	where the $x_a$'s are fresh individual variables.
	Let $\Amc_{i+1}$ be obtained from $\Amc_i$ by applying one of the following
	rule, where $C$ is a concept of the form $C \in \NC$ or $C = \{a\}$ or
	$C = \exists u.A$, and $x,y \in \mn{ind}{(\Amc_i)}$:
	\begin{enumerate}
		
		\item\label{dt-i1} if $C_1(x_1),\ldots,C_k(x_k)\in \Amc_i$, with
		$x_i=x$ or $x_i = x_{a_i}$ for some $a_i \in \NI \cap \sig(\Omc)$,
		and $\Omc\models C'_1\sqcap \ldots\sqcap C'_k\sqsubseteq C$,
		where $C'_i = C_i$ if $x_i = x$ and $C'_i = \exists u.(\{a_i\} \sqcap C_i)$
		if $x = x_{a_i}$, then add $C(x)$;
		
		\item\label{dt-i2'} if $(\exists u.A)(y) \in \Amc_i$ then add
		$(\exists u.A)(x)$;
		
		\item\label{dt-i2} if $r(x,y),A(y) \in \Amc_i$ and $\Omc \models \exists r.A \sqsubseteq C $,
		then add~$C(x)$;
		
		\item\label{dt-i3} if $C_1(x_1),\ldots,C_k(x_k)\in \Amc_i$, with
		$x_i=x$ or $x_i = x_{a_i}$ for some $a_i \in \NI \cap \sig(\Omc)$,
		and $\Omc\models C'_1\sqcap \ldots\sqcap C'_k\sqsubseteq
		{\exists u. (\{a\} \sqcap C)}$,
		where $C'_i = C_i$ if $x_i = x$ and $C'_i = \exists u.(\{a_i\} \sqcap C_i)$
		if $x = x_{a_i}$, then add $C(x_a)$.
	\end{enumerate}
	Note that the sequence is finite, and denote by $\Amc^*$ the final ABox.
	
	\smallskip\noindent\textit{Claim.} There is a model $\Imc,v$ of
	$\Amc^*$ and \Omc such that for all $x\in\mn{ind}(\Amc)$ and $A\in\NC$,
	$v(x)^\Imc \in A^\Imc$ implies $A(x)\in \Amc^*$.
	
	\smallskip\noindent\textit{Proof of the Claim.}
	For all $x,y \in \mn{ind}(\Amc^*)$,
	we write $x \sim y$ if $\{a\}(x), \{a\}(y) \in \Amc^*$ for some
	$a \in \NI \cap \sig(\Omc)$.
	Notice that if $\{a\}(x), \{a\}(y), C(x) \in \Amc^*$, then
	$C(x_a) \in \Amc^*$ by rule \ref{dt-i3}, and $C(y) \in \Amc^*$
	by rule \ref{dt-i1}.
	Therefore, $x \sim y$ implies $C(x) \in \Amc^*$ if and only
	if $C(y) \in \Amc^*$.
	In particular, $\sim$ is an equivalence relation.
	We let $[x]$ denote the equivalence class of~$x$.
	Start with an interpretation $\Imc_0$ defined by:
	\begin{align*}
		\Delta^{\Imc_0} & = \mn{ind}(\Amc^*)/{\sim} \\
		A^{\Imc_0} & = \{[x]\mid A(x)\in \Amc^*\} \\
		a^{\Imc_0} & = [x_a] \\
		r^{\Imc_0} & = \{([x],[y])\mid r(x,y)\in \Amc^*\} \, .
	\end{align*}
	Let $C_x$ denote the conjunction of all concepts of the form
	$C \in \NC$, $C = \{a\}$, $C = \exists u. A$, or
	$C = \exists u. (\{a\} \sqcap A)$ such that $\Amc^* \models C(x)$.
	Let $\Imc_x$ denote the canonical model for $\Omc$ and~$C_x$ rooted
	at $[x]$.
	Due to rule~\ref{dt-i1} and the universality of $\Imc_x$,
	for every concept name or nominal~$C$, we have $[x] \in C^{\Imc_0}$
	if and only if $[x] \in C^{\Imc_x}$.
	Similarly, because of rule \ref{dt-i3}, for every $a \in \NI \cap \sig(\Omc)$,
	$a^{\Imc_x} \in C^{\Imc_x}$ if and only if $a^{\Imc_0} \in C^{\Imc_0}$.
	
	We can now define $\Imc$ as follows: $\Delta^{\Imc}$ is the disjoint union of
	$\Delta^{\Imc_0}$ and all elements in domains $\Delta^{\Imc_x} \setminus
	(\{[x]\} \cup \{a^{\Imc_x} \mid a \in \NI \cap \sig(\Omc)\})$.
	Interpretations of concept names and nominals are inherited from the
	$\Imc_0$ or $\Imc_x$ each element comes from.
	Finally, $r^\Imc$ is obtained by taking the union of $r^{\Imc_0}$ and all
	$r^{\Imc_x}$ after replacing edges to/from $a^{\Imc_x}$ with edges to/from
	$a^{\Imc_0}$.
	It is clear that for the variable assignment $v(x) = [x]$,
	$\Imc_0,v$ satisfies $\Amc^*$, and thus so does $\Imc,v$.
	
	By rule \ref{dt-i1}, all concept inclusions of $\Omc$ of the form
	$\top \sqsubseteq A$, $A_1 \sqcap A_2 \sqsubseteq B$, $A \sqsubseteq \{a\}$
	and $\{a\} \sqsubseteq A$ are satisfied by $\Imc_0$. They are also satisfied
	by every $\Imc_x$ (since $\Imc_x$ is a model of $\Omc$), and thus by $\Imc$.
	Now consider a concept inclusion $A \sqsubseteq \exists r.B \in \Omc$,
	where $r$ is a role name or an inverse role.
	Recall that for every $a$ and $x$, $a^\Imc \in B^\Imc$ if and only if
	$a^{\Imc_x} \in B^{\Imc_x}$. Therefore, for all $d \in \Delta^{\Imc_x}$,
	$d \in (\exists r.B)^{\Imc_x}$ implies $d \in (\exists r.B)^{\Imc}$.
	The case $A \sqsubseteq \exists u. B$ is similar.
	Since every $\Imc_x$ satisfies $A \sqsubseteq \exists r.B$, so does $\Imc$.
	Similarly, every concept inclusion $\exists r.B \sqsubseteq A \in \Omc$
	is satisfied in $\Imc$: if the witness pair for $\exists r.B$ is part of
	$\Imc_0$, this follows from rule~\ref{dt-i2}, and if not, then it is part
	of some $\Imc_x$, which is by definition a model of $\Omc$.
	For concept inclusions of the form $\exists u. B \sqsubseteq A \in \Omc$,
	we can observe that if there exists some $d' \in \Delta^\Imc$ such that
	$d' \in B^\Imc$, then $(\exists u.B)$ is in  $C_x$ for some $x$, i.e., by
	rule~\ref{dt-i2'}, for all $x$.
	
	Finally, for all $x\in\mn{ind}(\Amc)$ and $A\in\NC$,
	$[x]^\Imc \in A^\Imc$ implies $[x]^{\Imc_0} \in A^{\Imc_0}$, i.e., $A(x)\in \Amc^*$.
	This concludes the proof of the claim.
	
	\medskip Now suppose $\Omc,\Amc\models A_0(x_0)$.
	By the Claim, we have $A_0(x_0)\in \Amc^*$.
	Since the four rules to construct $\Amc_0,\Amc_1,\ldots$ are in
	one-to-one correspondence with Conditions~(1)--(4) from the
	definition of derivation trees, we can inductively construct a
	derivation tree for $A_0(x_0)$ in \Amc w.r.t.\ \Omc.
        This concludes the proof of Part 1.

        \bigskip
        
	The proof of Part 2 is similar to that of Lemma~\ref{lem:deriv0}.
	We define $(T,V')$ as follows from $(T,V)$, starting with the root by
	setting $V'(\epsilon)=V(\epsilon)=(x_{0},A_{0})$.
	At each step, if $V(n)=(a,C)$ then $V'(n)=(w,C)$ for some $w$
        such that $\text{tail}(w)=a$.
	To define the labelings of the successors of $n$, we consider the possible
	derivation steps for $(a,C)$ in $\Amc$.
	\begin{enumerate}
		\item $a = x \in \mn{ind}(\Amc)$, and $n$ has successors
		$n_1,\ldots,n_k$, $k\geq 1$ with $V(n_i) = (a_i,C_i)$,
		such that $a_i = x$ or $a_i \in \NI \cap \sig(\Omc)$ for all $i$, and
		defining $C'_i = C_i$ if $a_i = x$, and
		$C'_i = \exists u. (\{a_i\} \sqcap C_i)$ otherwise,
		we have $\Omc\models C'_1\sqcap \ldots\sqcap C'_k\sqsubseteq C$.
		Take $V'(n_i) = (w,C_i)$ if $x_i = x$, and $V'(n_i) = (a_i,C_i)$
		if $a_i \in \NI \cap \sig(\Omc)$.
		
		\item $C = \exists u.A$ and $n$ has a single successor $n'$ with
		$V(n')=(y,\exists u.A)$.
		Take $V'(n') = (y,\exists u.A)$.
		
		\item  $n$ has a single successor $n'$ with $V(n')=(y,A)$ such
		that $r(a,y)\in \Amc$ and $\Omc \models \exists r.A\sqsubseteq C$
		(where $r$ is a role name or an inverse role).
		Take $V'(n') = (w r y, A)$.
		
		\item $a \in \NI \cap \sig(\Omc)$ and there exists
		$x \in \mn{ind}(\Amc)$ such that $n$ has successors
		$n_1,\ldots,n_k$, $k\geq 1$ with $V(n_i)=(a_i,C_i)$ and $a_i = x$ or
		$a_i \in \NI \cap \sig(\Omc)$ for all $i$,
		and, defining $C'_i = C_i$ if $a_i = x$, and
		$C'_i = \exists u. (\{a_i\} \sqcap C_i)$ otherwise, we have
		$\Omc\models C'_1\sqcap \ldots\sqcap C'_k\sqsubseteq \exists u.(\{a\} \sqcap C)$.
		Take $V'(n_i) = (x,C_i)$ if $x_i = x$, and $V'(n_i) = (a_i,C_i)$
		if $a_i \in \NI \cap \sig(\Omc)$.
	\end{enumerate}
	Then $(T,V')$ is a derivation tree for $A_{0}(x_{0})$ in $\Amc^{*}$
	w.r.t.~$\Omc$.
\end{proof}

\paragraph{Tree Automata.}
A \emph{tree} is a non-empty set $T \subseteq (\mathbb{N} \setminus \{0\})^\ast$
closed under prefixes and such that $n \cdot (i+1) \in T$ implies
$n \cdot i \in T$.
It is \emph{$k$-ary} if $T \subseteq \{1,\ldots,k\}^\ast$.
The node $\varepsilon$ is the \emph{root} of $T$.
As a convention, we take $n \cdot 0 = n$ and $(n \cdot i) \cdot -1 = n$.
Note that $\varepsilon \cdot -1$ is undefined.
Given an alphabet $\Theta$, a \emph{$\Theta$-labeled tree} is a pair $(T,L)$
consisting of a tree $T$ and a node-labeling function $L : T \to \Theta$.

A \emph{non-deterministic tree automaton (NTA)} over finite $k$-ary trees
is a tuple $\Amf=(Q,\Theta,I,\Delta)$, where $Q$ is a set of states,
$\Theta$ is the input alphabet, $I \subseteq Q$ is the set of initial states,
and $\Delta\subseteq Q\times \Theta\times \bigcup_{0 \le \ell \le k} Q^\ell$ is the transition relation.
A \emph{run} of an NTA $\Amf=(Q,\Theta,I,\Delta)$ over a
$k$-ary input $(T,L)$ is a $Q$-labeled tree $(T,r)$ such that
for all $x \in T$ with children $y_1, \ldots, y_\ell$, 
$(r(w),L(w),r(y_1),\ldots,r(y_\ell))\in \Delta$.
It is \emph{accepting} if $r(\varepsilon) \in I$.
The language accepted by $\Amf$, denoted $L(\Amf)$, is the set of all finite
$k$-ary $\Theta$-labeled trees over which $\Amf$ has an accepting run.

A \emph{two-way alternating tree automaton over finite $k$-ary trees (2ATA)}
is a tuple $\Amf = (Q,\Theta,q_0,\delta)$ where $Q$ is a finite set
of {\em states}, $\Theta$ is the {\em input alphabet}, $q_0\in Q$ is
the {\em initial state}, and $\delta$ is a {\em transition function}.
The transition function $\delta$ maps every state $q$ and input letter
$\theta \in \Theta$ to a positive Boolean formula $\delta(q,\theta)$ over
the truth constants $\mn{true}$ and $\mn{false}$ and \emph{transition atoms}
of the form $(i,q)\in [k]\times Q$, where $[k]=\{-1,0,1,\ldots,k\}$.
The semantics is given in terms of \emph{runs}.
More precisely, let $(T,L)$ be a finite $k$-ary $\Theta$-labeled tree and
$\Amf=(Q,\Theta,q_0,\delta)$ a 2ATA.
An {\em accepting run of \Amf over $(T,L)$} is a $(T\times Q)$-labeled tree
$(T_r,r)$ such that: 
\begin{enumerate}
	
	\item $r(\varepsilon)=(\varepsilon,q_0)$, and 
	
	\item for all $y\in T_r$ with $r(y)=(x,q)$, there is a
	subset $S\subseteq [k]\times Q$ such that $S\models
	\delta(q,L(x))$ and for every $(i,q')\in S$, there is some
	successor $y'$ of $y$ in $T_r$ with $r(y)=(x\cdot i,q')$.
	
\end{enumerate}
The language accepted by $\Amf$, denoted $L(\Amf)$, is the set of all finite
$k$-ary $\Theta$-labeled trees $(T,L)$ for which there is an accepting run. 

From a 2ATA $\Amf$, one can compute in exponential time an NTA $\Amf'$
whose number of states is exponential in the number of states of $\Amf$
and such that $L(\Amf) = L(\Amf')$ \cite{DBLP:conf/icalp/Vardi98}.

\paragraph{Interpolant Existence.}
We now give the proof that Point~2 in Theorem~\ref{thm:critele2} entails
an exponential time upper bound for deciding the existence of an interpolant.
We focus on the case of $\mathcal{ELIO}_u$.
Let $\Omc_1, \Omc_2$ be $\mathcal{ELIO}_u$-ontologies in normal form,
$A,B \in \NC$, and $\Sigma = \sig(\Omc_1,A) \cap \sig(\Omc_2,B)$.
We can assume that $A \in \sig(\Omc_1)$ and $B \in \sig(\Omc_2)$.

\newcommand\ap{{\Lambda}}

As our proof relies on tree automata, let us first explain how we represent
ABoxes that are tree-shaped modulo 
$\NI \cap \Sigma$ as trees over the alphabet $2^{\ap}$, where
\begin{align*}
  \ap = {}
  & \NC \cap \Sigma \cup {} \\
  & \{ \{a\} \mid a \in \NI \cap \Sigma \} \cup {} \\
  & \{r,r^- \mid r \in \NR \cap \Sigma \} \cup {} \\
  & \{\exists r.\{a\} \mid r \in \NR \cap \Sigma \land a \in \NI \cap \Sigma\}
    \, \cup {} \\
  & \{\exists r^-.\{a\} \mid r \in \NR \cap \Sigma 
    \land a \in \NI \cap \Sigma\} \, .
\end{align*}

Intuitively, the nodes of the tree correspond to the individual variables of the ABox;
labels $C \in \NC, \{a\}, \exists r. \{a\}, \exists r^-. \{a\}$ indicate concepts that
hold at the current node, while labels $r$ or $r^-$ are used to indicate which
roles (if any) connect a node to its parent.
Note that there need not be such a label $r$ or $r^-$, so connected nodes
in the tree representation are not necessarily connected in the ABox.

More precisely, we associate with every $2^\ap$-labeled tree $(T,L)$ the
following ABox, where $x_a$ are fresh individual variables:
\begin{align*}
	\Amc_{(T,L)} = {}
	& \{\top(x) \mid x \in T\} \cup {} \\
	& \{ \{a\}(x_a) \mid \exists x \in T: \{a\} \in L(x) \}
          \cup {} \\
  	& \{\{a\}(x) \mid x \in T \land \{a\} \in L(x)\} \cup {} \\
	& \{B(x) \mid x \in T \land B \in L(x)\} \cup {} \\
	& \{r(x,x \cdot i) \mid x \cdot i \in T \land r \in \NR \land r \in L(x \cdot i)\}
	\cup {} \\
	& \{r(x \cdot i,x) \mid x \cdot i \in T \land r \in \NR \land r^- \in
	L(x \cdot i)\} \cup {} \\
	& \{r(x,x_a) \mid x \in T \land \exists r. \{a\} \in L(x)\} \cup {} \\
	&  \{r(x_a,x) \mid x \in T \land \exists r^-. \{a\} \in L(x)\} \, .
\end{align*}
Notice that $\Amc_{(T,L)}$ is tree-shaped modulo $\NI \cap \Sigma$.
Conversely, for every ABox $\Amc$ that is tree-shaped modulo
$\NI \cap \Sigma$, there exists
a (not necessarily unique) tree $(T,L)$ such that $\Amc = \Amc_{(T,L)}$.
In addition, if the degree of $G_{\Amc}^{u}$ is less than $k$, then
there exists a $k$-ary tree $(T,L)$ such that $\Amc = \Amc_{(T,L)}$.
For instance, $\Amc_{\Omc_{1}\cup \Omc_{2},A}^{\Sigma}$ can be represented
by a $k$-ary tree for any $k$ larger than the number of concept inclusions
in $\Omc_1 \cup \Omc_2$.

We also denote by $\Amc_{(T,L)}^\Sigma$ the $\Sigma$-reduct of $\Amc_{(T,L)}$.

\smallskip

We describe below an NTA $\Amf_1$ with exponentially many states accepting
trees that represent prefix-closed finite subsets of
$\Amc_{\Omc_1 \cup \Omc_2,A}^\Sigma$,
and a 2-ATA $\Amf_2$ with polynomially many states accepting trees $(T,L)$
such that $\Amc_{(T,L)} \models B(\epsilon)$.
The existence of an interpolant then reduces to the non-emptiness of
$L(\Amf_1) \cap L(\Amf_2)$.

\paragraph{Definition of $\Amf_1$.}
We represent the canonical model for $\Omc_1 \cup \Omc_2$ and $A$ by
a tree with $\tau_{A}$ at the root of the tree,
other $\tau \in S$ inserted at arbitrary positions in the tree,
and $\tau_0 r_1 \tau_1 \cdots r_n \tau_n$ below
$\tau_0 r_1 \tau_1 \cdots r_{n-1} \tau_{n-1}$ if $n>0$.
We want $\Amf_1$ to accept finite subsets of $\Amc_{\Omc_1 \cup \Omc_2,A}^\Sigma$
obtained by keeping a prefix-closed finite subset of nodes,
and possibly removing some concepts and relations from the labels
(including all concepts and relations not in $\Sigma$).
To do so, the automaton will simply guess in its state the type of
each node, and check that all guesses are locally consistent
by allowing only transitions that match the definition of canonical models.
Concretely, the states of the automaton consist of a pair
of $\Omc$-types, where state $(\tau,\tau')$ should be interpreted as
the parent node having type $\tau$ and the current node type $\tau'$.

To keep the definition simple, the automaton also accepts trees where,
compared to the canonical model, some nodes are duplicated
(that is, we do not require that the node corresponding to some
$\tau \in \Delta^{\Imc_{\Omc,A}} \cap S$ is unique). This does not change
the set of concepts entailed at the root.

We take $\Amf_1=(Q_1,2^\ap,I_1,\Delta_1)$, where
\begin{itemize}
\item $Q_1 = (S \cup \{\bot\}) \times S$,
  where $S$ is the set of $\Omc$-types introduced in the
  definition of $\Imc_{\Omc,A}$;
\item $I_1 = \{(\bot,\tau_A)\}$;
\item For states $q = (\tau,\tau'), q_1 = (\tau_1,\tau_1'),\ldots,
  q_\ell = (\tau_\ell,\tau_\ell') \in Q_1$ and input letter
  $\alpha \subseteq \ap$,
  $(q,\alpha,q_1,\ldots,q_\ell) \in \Delta_1$ if the following conditions
  are satisfied, for all $1 \le i \le \ell$:
  \begin{itemize}
  \item the current state and label are consistent with the definition
    of the canonical model:
    for all $r \in \alpha$, $\tau \rightsquigarrow_r \tau'$;
  \item the set of concepts associated with $\alpha$ is a subset of the
    $\Omc$-type $\tau'$:
    $\alpha \cap (\sub(\Omc) \cup \sub^\exists(\Omc)) \subseteq \tau'$;
  \item the current type $\tau'$ is stored in the state of all child nodes:
    for all $1 \le i \le \ell$, $\tau_i = \tau'$.
  \end{itemize}
      \end{itemize}

      Note that $\Amf_1$ can be computed in exponential time.

      \begin{lemma}\label{lem:Amf1}
        $\Omc_1 \cup \Omc_2, \Amc_{\Omc_1 \cup \Omc_2,A}^\Sigma \models B(\rho_A)$
        if and only if there exists $(T,L) \in L(\Amf_1)$ such that
        $\Omc_1 \cup \Omc_2, \Amc_{(T,L)} \models B(\epsilon)$, where
        $\epsilon$ is the root of $(T,L)$.
      \end{lemma}

      \begin{proof}
        The run of $\Amf_1$ on some $(T,L) \in L(\Amf_1)$ can be used to
        define a homomorphism from $\Amc_{(T,L)},\epsilon$ to
        $\Amc_{\Omc_1 \cup \Omc_2,A}^\Sigma,\rho_A$.
        Therefore, if $\Omc_1 \cup \Omc_2, \Amc_{(T,L)} \models B(\epsilon)$
        then
        $\Omc_1 \cup \Omc_2, \Amc_{\Omc_1 \cup \Omc_2,A}^\Sigma \models B(\rho_A)$.
        Conversely, if $\Omc_1 \cup \Omc_2, \Amc_{\Omc_1 \cup \Omc_2,A}^\Sigma
        \models B(\rho_A)$ then there exists a finite subset $\Amc$ of
        $\Amc_{\Omc_1 \cup \Omc_2,A}^\Sigma$ such that
        $\Omc_1 \cup \Omc_2, \Amc \models B(\rho_A)$.
        Take as $(T,L)$ any finite prefix of an encoding of
        $\Amc_{\Omc_1 \cup \Omc_2,A}^\Sigma$ that contains all nodes corresponding
        to individuals in $\Amc$.
        Then the labeling of $(T,L)$ with the full types from the canonical
        model defines an accepting run of $\Amf_1$ on $(T,L)$,
        and $\Omc_1 \cup \Omc_2, \Amc_{(T,L)} \models B(\epsilon)$.
      \end{proof}

\paragraph{Definition of $\Amf_2$.}
The construction of $\Amf_2 = (Q_2,2^\ap,q_{B},\delta_2)$ relies
on derivation trees.
Intuitively, runs of $\Amf_2$ on some $(T,L)$ correspond to derivation trees
for $B(\varepsilon)$ in $\Omc_1 \cup \Omc_2, \Amc_{(T,L)}$.
The states of $\Amf_2$ are
\begin{align*}
	Q_2 = {} 
	& \{q_{A'} \mid A' \in \NC \cap \sig(\Omc_1,\Omc_2)\} \cup {} \\
	& \{q_{\{a\}} \mid a \in \NI \cap \sig(\Omc_1,\Omc_2)\} \cup {} \\
	& \{q_{\exists r.A'}, q_{\exists r^-. A'} \mid r \in \NR \cap \sig(\Omc_1,\Omc_2), \\
	& \hphantom{\{q_{\exists r.A'}, q_{\exists r^-. A'} \mid {}}
          A' \in \NC \cap \sig(\Omc_1,\Omc_2)\} \cup {} \\
	& \{q_{\exists u.A'} \mid A' \in \NC \cap \sig(\Omc_1,\Omc_2)\} \cup {} \\
	& \{q_{\exists u.(\{a\} \sqcap A')} \mid a \in \NI \cap \sig(\Omc_1,\Omc_2),\\
        & \hphantom{\{q_{\exists u.(\{a\} \sqcap A')} \mid {}}
          A' \in \NC \cap \sig(\Omc_1,\Omc_2)\} \cup {} \\
	& \{q_{\exists u.(\{a\} \sqcap \{b\})} \mid a, b \in \NI \cap \sig(\Omc_1,\Omc_2)\} \cup {} \\
	& \{q_r,q_{r^-} \mid r \in \NR \cap \Sigma\} \cup {} \\
	& \{q_{\exists r.\{a\}},q_{\exists r^-.\{a\}} \mid a \in \NI \cap \sig(\Omc_1,\Omc_2), \\
        & \hphantom{\{q_{\exists r.\{a\}},q_{\exists r^-.\{a\}} \mid {}}
	r \in \NR \cap \sig(\Omc_1,\Omc_2)\} \, .
\end{align*}
Intuitively, state $q_C$ is used to check that $C$ is entailed at the current
node. States $q_r$ and $q_{\exists r.\{a\}}$ are used to check the label of the current node.
The initial state is $q_{B}$, as we are trying to construct a derivation tree
for $B$ at the root.

Let us now define the transition relation.
From a state $q_r$ or $q_{\exists r.\{a\}}$, where $r \in \NR \cup \NR^-$
and $a \in \NI$, the automaton simply checks the current label:
\begin{align*}
	\delta_2(q_r,\alpha)
	& = \begin{cases}
		\mn{true} & \text{if } r \in \alpha \\
		\mn{false} & \text{if } r \notin \alpha
	\end{cases} \\
	\delta_2(q_{\exists r.\{a\}},\alpha)
	& = \begin{cases}
		\mn{true} & \text{if } \exists r.\{a\} \in \alpha \\
		\mn{false} & \text{if } \exists r.\{a\} \notin \alpha \, .
              \end{cases}
\end{align*}
From a state $q_{\exists r.A'}$, with $r \in \NR \cup \NR^-$ and $A' \in \NC$,
the automaton checks that the current node has an $r$-successor from which
there exists a run starting in $q_{A'}$. This $r$-successor can be
(i) the parent of the current node, i.e.~there is a run from $q_{r^-}$ from
the current node and a run from $q_{A'}$ from the parent node,
(ii) some $i$-th child of the current node, i.e.~there is a run from $q_{r}$
and one from $q_{A'}$ from the $i$-th child, or
(iii) an individual $a$, i.e.~there is a run from $q_{\exists r.\{a\}}$ and from
$q_{\exists u.(\{a\} \sqcap A')}$ from the current node:
\begin{align*}
  \delta_2(q_{\exists r.A'},\alpha) = {}
  & (0,q_{r^-}) \land (-1,q_{A'}) \lor {} \\
  & \bigvee_{1 \le i \le k} (i,q_{A'}) \land (i,q_r)  \lor {} \\
  & \bigvee_{a \in \NI \cap \Gamma} (0,q_{\exists r.\{a\}}) \land (0,q_{\exists u.(\{a\} \sqcap {A'})}) \, .
\end{align*}
From a state $q_{\exists u.A'}$, the automaton checks if
(i) condition~\ref{dt1} from derivation trees can be applied,
that is, there exist concepts
$C_1,\ldots,C_n$ of the form $B'$, $\{a\}$, $\exists u. (\{a\} \sqcap B')$
or $\exists u. (\{a\} \sqcap \{b\})$
such that
$\Omc_1 \cup \Omc_2 \models C_1 \sqcap \cdots \sqcap C_n \models \exists u.A'$
and there exists a run from each $q_{C_i}$ from the current node,
or
(ii) condition~\ref{dt2} from derivation trees can be applied, which can
be checked by propagating the search for a run from $q_{\exists u.A'}$
to all neighbouring nodes,
or
(iii) condition~\ref{dt3} from derivation trees can be applied, that is,
there exists $r,B'$ such that
$\Omc_1 \cup \Omc_2 \models \exists r.B' \sqsubseteq \exists u. A'$
and the automaton has a run from $q_{\exists r.B'}$ starting
from the current node:
\begin{align*}
	\delta_2(q_{\exists u.A'},\alpha) = {}
	& \bigvee_{\Omc_1 \cup \Omc_2 \models C_1 \sqcap \cdots \sqcap C_n \models \exists u.A'} \bigwedge_{1 \le i \le n} (0,q_{C_i}) \lor {} \\
	& \bigvee_{i \in \{-1,1,\ldots,k\}} (i,q_{\exists u.A'}) \lor {} \\
        & \bigvee_{\Omc_1 \cup \Omc_2 \models \exists r.B' \sqsubseteq \exists u.A'} (0,q_{\exists r.B'}) \, .
\end{align*}
From a state $q_{\exists u.C}$ where $C = \{a\} \sqcap A'$ or
$C = \{a\} \sqcap \{b\}$ with $b \neq a$, the automaton checks if
condition~\ref{dt4} from derivation trees can be applied either
(i) taking the current node as $x$, that is, there exist concepts
$C_1,\ldots,C_n$ of the form $B'$, $\{b\}$, $\exists u. (\{b\} \sqcap B')$
or $\exists u. (\{b\} \sqcap \{c\})$
such that
$\Omc_1 \cup \Omc_2 \models C_1 \sqcap \cdots \sqcap C_n \models \exists u.C$
and there exists a run from each $q_{C_i}$ from the current node,
or
(ii) taking some other node as $x$, which can be checked by propagating the
search for a run from $q_{\exists u.C}$ to all neighbouring nodes:
\begin{align*}
	\delta_2(q_{\exists u.C},\alpha)
	& = 
	\bigvee_{\Omc_1 \cup \Omc_2 \models C_1 \sqcap \cdots \sqcap C_n \models \exists u.C} \bigwedge_{1 \le i \le n} (0,q_{C_i}) \lor {} \\
	& \hspace{10em}
          \bigvee_{i \in \{-1,1,\ldots,k\}} (i,q_{\exists u.C}) \, .
\end{align*}
We also set
\[
  \delta_2(q_{\exists u. (\{a\} \sqcap \{a\})},\alpha) = \mn{true} \, .
\]
For $C = \{a\}$ or $C \in \NC$, $\delta(q_C,\alpha) = \mn{true}$
if $C \in \alpha$ or $\Omc_1 \cup \Omc_2 \models \top \sqsubseteq C$,
and otherwise, the automaton checks if conditions~\ref{dt1} or~\ref{dt3}
from derivation trees can be applied:
\begin{multline*}
	\delta_2(q_C,\alpha) =
	\bigvee_{\Omc_1 \cup \Omc_2 \models C_1 \sqcap \cdots \sqcap C_n \models C}
	\bigwedge_{1 \le i \le n} (0,q_{C_i}) \lor {} \\
	\bigvee_{\Omc_1 \cup \Omc_2 \models \exists r.B' \sqsubseteq C} (0,q_{\exists r.B'}) \, .
\end{multline*}

\begin{lemma}\label{lem:Amf2}
  For all finite $k$-ary $2^\ap$-labeled trees $(T,L)$,
  we have $(T,L) \in L(\Amf_2)$ if and only if
  $\Omc_1 \cup \Omc_2, \Amc_{(T,L)} \models B(\epsilon)$.
\end{lemma}

\begin{proof}
  We observe that for all $(T,L)$,
  \begin{itemize}
  \item For all $\sig(\Omc_1,\Omc_2)$-concept $C$ of the form
    $C= A'$, $C = \{a\}$ or $C = \exists u.A'$ with $A' \in \NC$ and
    $a \in \NI$,
    $\Amf$ has a run starting from state $q_C$ on $(T,L)$ if and only if
    there exists a derivation tree for $(\epsilon,C)$ in
    $\Omc_1 \cup \Omc_2, \Amc_{(T,L)}$.
  \item For all $a \in \NI \cap \sig(\Omc_1,\Omc_2)$, for all
    $\sig(\Omc_1,\Omc_2)$-concept $C = \{b\}$ or $C = A' \in \NC$,
    $\Amf$ has a run starting from state $q_{\exists u.(\{a\} \sqcap C)}$ if
    and only if there exists a derivation tree for $(a,C)$ in
    $\Omc_1 \cup \Omc_2, \Amc_{(T,L)}$. \qedhere
  \end{itemize}
\end{proof}

From $\Amf_2$, one can construct an equivalent NTA $\Amf'_2$ with exponentially
many states \cite{DBLP:conf/icalp/Vardi98}.
By Lemmas~\ref{lem:Amf1} and~\ref{lem:Amf2}, we have 
$\Omc_1 \cup \Omc_2, \Amc_{\Omc_1 \cup \Omc_2}^\Sigma \models B(\rho_A)$
if and only if $L(\Amf_1) \cap L(\Amf_2') = \emptyset$, which can be
checked in exponential time.

\paragraph{Lower Bound for Explicit Definitions.}
We construct an $\mathcal{ELI}$-ontology $\Omc$, signature $\Sigma$, and concept name $A$ such that the smallest explicit $\mathcal{ELI}(\Sigma)$-definition of $A$ under $\Omc$ is of double exponential size in $||\Omc||$. $\Omc$ is a variant of ontologies constructed in~\cite{DBLP:journals/jsc/LutzW10,DBLP:journals/ai/NikitinaR14} and defined as follows. It contains $\top \sqsubseteq \exists r.\top \sqcap \exists s.\top$, 
\begin{align*}
	A  \sqsubseteq M \sqcap \overline{X_{0}} \sqcap  \ldots \sqcap \overline{X_{n}} &\\
	\exists \sigma^{-}. (\overline{X_{i}} \sqcap X_{0} \sqcap \ldots \sqcap X_{i-1}) \sqsubseteq X_{i} &\qquad \sigma \in \{r,s\}, i \leq n \\
	\exists \sigma^{-}. (X_{i} \sqcap X_{0} \sqcap \ldots \sqcap X_{i-1}) \sqsubseteq \overline{X_{i}} &\qquad \sigma \in \{r,s\}, i \leq n\\ 
	\exists \sigma^{-}. (\overline{X_{i}} \sqcap \overline{X_{j}}) \sqsubseteq \overline{X_{i}} &\qquad \sigma \in \{r,s\}, j < i \leq n\\
	\exists \sigma^{-}. (X_{i} \sqcap \overline{X_{j}}) \sqsubseteq X_{i} &\qquad \sigma\in \{r,s\}, j < i \leq n\\
	X_{0} \sqcap \ldots \sqcap X_{n} \sqsubseteq L & 
\end{align*}
and 
$$
L \sqsubseteq B, \quad \exists r.B \sqcap \exists s.B \sqsubseteq B, \quad B \sqcap M \sqsubseteq A.
$$
Let $\Sigma=\{M,r,s,L\}$. Note that $A$ triggers a marker $M$ and a binary tree of depth $2^{n}$ using counter concept names 
$X_{0},\ldots,X_{n}$ and $\overline{X_{0}},\ldots,\overline{X_{n}}$. 
A concept name $L$ is made true at the leafs. Conversely, if $L$ is true at the leafs of a binary tree of depth $2^{n}$ then $B$ is true at all nodes of the tree and $A$ is entailed by $M$ and $B$ at its root.
Define inductively 
$$
C_{0}= L, \quad C_{k+1}= \exists r.C_{k} \sqcap \exists s.C_{k}, \quad C=C_{2^{n}}\sqcap M.
$$
Then $C$ is the smallest explicit $\mathcal{ELI}(\Sigma)$-definition of $A$ under $\Omc$.

\paragraph{Transfer Sequences.}
For the proof of ``2. $\Rightarrow$ 3.'' of Theorem~\ref{thm:critele2} and the proof
that interpolants can be computed in double exponential time we require an extension of the notion of transfer sequences first introduced in \cite{DBLP:conf/ijcai/BienvenuLW13}
to logics with nominals.

\medskip

Assume that Condition~2 of Theorem~\ref{thm:critele2} holds. So we have $\mathcal{ELIO}_{u}$-ontologies $\Omc_{1},\Omc_{2}$ in normal form, concept names 
$A,B$, and $\Sigma=\text{sig}(\Omc_{1},A) \cap \text{sig}(\Omc_{2},B)$ such that
$\Omc_{1}\cup\Omc_{2},\Amc^{\Sigma}_{\Omc_{1} \cup \Omc_{2},A}\models B(\rho_{A})$.
Set $\Omc=\Omc_{1}\cup \Omc_{2}$. We use $\Amc_{\Omc,A}$ to denote the ABox associated with the canonical model $\Imc_{\Omc,A}$. We require some notation for the individuals
that occur in $\Amc_{\Omc,A}$. We set $a\sim b$ if $\Omc\models \{a\} \sqcap \exists u.A\sqsubseteq \{b\}$
and set $[a]=\{b\in \text{sig}(\Omc) \mid a\sim b\}$. We say that concept name $E$
is \emph{absorbed} by $a$ if $\Omc\models E \sqcap \exists u.A \sqsubseteq \{a\}$.
We denote the individual $x_{\tau_{a}}$ of $\Amc_{\Omc,A}$ by $x_{a}$ and the 
individuals $x_{\tau_{E}}$ of $\Amc_{\Omc,A}$ by $x_{E}$. 
Note that $x_{a}=x_{b}$ if $a\sim b$ and $x_{a}=x_{A}$ if $A$ is absorbed by $a$. 

Given $w\in \text{ind}(\Amc_{\Omc,A})$, we call the individuals 
of the form $ww'\in \text{ind}(\Amc_{\Omc,A})$ \emph{the subtree of $\Amc_{\Omc,A_1}$ rooted at $w$}.

By compactness we have a finite subset $\Amc$ of $\Amc_{\Omc,A}$ containing $x_{A}$ such that
$\Omc,\Amc_{|\Sigma}\models B(x_{A})$. 
We may assume that $\Amc$ is prefix closed and that $\Amc_{|\Sigma}$ contains
\begin{itemize}
  \item $\{a\}(x_{a})$ and $A(x_{A})$ for all $a,A\in \Sigma$;
  \item $\top (x_{A})$ and $\top (x_{a})$ for all $a,A\in \text{sig}(\Omc)\setminus \Sigma$;
\end{itemize}
We obtain the ABox $\Amc_{\Sigma}$ from $\Amc_{|\Sigma}$ by adding the assertions 
\begin{itemize}
	\item $\{a\}(x_{a,\text{new}})$ and $\top (x_{A,\text{new}})$, for all 
$a,A\in \text{sig}(\Omc) \setminus \Sigma$, 
where $x_{a,\text{new}}$ and $x_{A,\text{new}}$ are fresh individuals. 
\end{itemize}
Let $I$ denote the set of individuals $x_{a},x_{A}$ with $a,A\in \text{sig}(\Omc)$ and let
$I_{\text{new}}$ denote the set of individuals $x_{a,\text{new}}, x_{A,\text{new}}$ with $a,A\in \text{sig}(\Omc)\setminus \Sigma$. Observe that $\Omc,\Amc_{\Sigma}$ and $\Omc,\Amc_{|\Sigma}$ entail the same assertions $C(a)$ for $a\in \text{ind}(\Amc_{|\Sigma}$), so the additional individuals do not influence what is entailed. In fact, we introduce the individuals $I_{\text{new}}$ only to enable explicit bookkeeping about when in a transfer sequence (defined below) an assertion of the form $C(a)$ or $\exists u.A$ is derived. 

We aim to define a small subset $\Amc'$ of $\Amc_{\Sigma}$ such that 
$\Omc\models A \sqsubseteq C$ for the concept $C$ corresponding to $\Amc'$
and such that still $\Omc,\Amc'\models B(x_{A})$. If $\Amc'$ has at most exponential depth in the size of $\Omc$ then we are done, as then $\Amc'$ is of at most double exponential size in the size of $\Omc$. We obtain $\Amc'$ from $\Amc_{\Sigma}$ by determining $w$ and $ww'\in \text{ind}(\Amc)\setminus (I \cup I_{\text{new}})$ which behave `sufficiently similar' such that if we obtain $\Amc'$ from 
$\Amc$ by replacing the subtree rooted at
$w$ in $\Amc_{\Sigma}$ by the subtree rooted at $ww'$, then we still have 
$\Omc,\Amc'\models B(x_{A})$ and $\Omc\models A\sqsubseteq C$ for the concept 
$C$ defined by $\Amc'$. The replacement of subtrees is then performed exhaustively.
 
For $w$ and $ww'$ to be sufficiently similar, we firstly require that $\text{tail}(w)= \text{tail}(ww')$ (with $\text{tail}(w)$ the final type in $w$ for any $w$). 
This ensures that $\Omc\models A \sqsubseteq C$ 
for the concept $C$ corresponding to $\Amc'$. This also has the consequence that $\Amc'$
is (isomorphic) to a prefix closed subABox of $\Amc_{\Sigma}$.
For the second condition for being sufficiently similar, we apply the notion of transfer sequences~\cite{DBLP:conf/ijcai/BienvenuLW13}. 
To define transfer sequences, we consider derivations using $\Omc$ and \emph{intermediate} ABoxes
$\Bmc$ such that 
$$
I\cup I_{\text{new}}\subseteq \text{ind}(\Bmc) \subseteq \text{ind}(\Amc_{\Sigma})
$$ 
We admit $\Bmc$ to contain equations $x_{e}=x_{e'}$ for 
$x_{e},x_{e'}\in I \cup I_{\text{new}}$,
with the obvious semantics.
Consider such an intermediate $\Bmc$ and $w\in \text{ind}(\Bmc)\setminus (I \cup I_{\text{new}})$. 
Then the set $D_{\Bmc}(w)$ is defined as the set of assertions $\alpha$ with 
$\Omc,\Bmc'\models\alpha$ and $\alpha$ of the form 
\begin{itemize}
  \item $A(c)$ or $\{a\}(c)$ with $A,a\in \text{sig}(\Omc)$ and $c \in \{w\} \cup I \cup I_{\text{new}}$; or
\item $r(w,c)$ with $r\in \NR \cup \NR^{-}$ and $c\in I\cup I_{\text{new}}$;
\item $r(c,d)$ with $r\in \NR \cup \NR^{-}$ and $c,d\in I\cup I_{\text{new}}$;
\item $c=d$ with $c,d\in I\cup I_{\text{new}}$.
\end{itemize}
and 
$\Bmc' = \Bmc \cup \{ A(x_{A,\text{new}}) \mid \Omc,\Bmc\models \exists u.A\}$.
For $w\in \text{ind}(\Bmc)\setminus (I \cup I_{\text{new}})$, let 
\begin{itemize}
   \item $\Bmc_{w}^{\downarrow}$ denote the restriction of $\Bmc$ to the individuals in the 
         subtree of $\Bmc$ rooted at $w$ and $I \cup I_{\text{new}}$; and let
    \item $\Bmc_{w}^{\uparrow}$ be the ABox obtained from $\Bmc$ by dropping $\Bmc_{w}^{\downarrow}$ from $\Bmc$ except for $w$ 
          itself and $I \cup I_{\text{new}}$.
\end{itemize}
Define the \emph{transfer sequence $\Xmc_{0},\Xmc_{1},\ldots$ of $(\Amc_{\Sigma},w)$ w.r.t.\ $\Omc$} 
as follows: 
\begin{eqnarray*}
     \Xmc_{0}  & = & D_{(\Amc_{\Sigma})_{w}^{\downarrow}}(w)\\
     \Xmc_{1}  & = & D_{(\Amc_{\Sigma})_{w}^{\uparrow}\cup \Xmc_{0}}(w)\\
     \Xmc_{2}  & = & D_{(\Amc_{\Sigma})_{w}^{\downarrow}\cup \Xmc_{1}}(w)\\
     \Xmc_{3}  & = &  ...
\end{eqnarray*}
Intuitively, we first consider the set $\Xmc_{0}$ of assertions that are entailed by 
$\Omc$ and $\Amc_{\Sigma}$ at $\{w\}\cup I \cup I_{\text{new}}$ 
if we only use assertions in ${\Amc_{\Sigma}}_{w}^{\downarrow}$. We update $\Amc_{\Sigma}$ 
by those assertions. 
Next we consider the set $\Xmc_{1}$ of assertions
that are entailed by $\Omc$ and the updated $\Amc_{\Sigma}$ at $\{w\}\cup I \cup I_{\text{new}}$ 
if we only use assertions in the updated 
${\Amc_{\Sigma}}_{w}^{\uparrow}$. We update $\Amc_{\Sigma}$ again, and so on. 
It is not difficult to see that if $w,ww'\in \text{ind}(\Amc_{\Sigma})\setminus 
(I \cup I_{\text{new}})$ and 
\begin{itemize}
  \item the restrictions of $\Amc_{\Sigma}$ to $\{w\}\cup (I \cup I_{\text{new}})$ and 
$\{ww'\} \cup (I \cup I_{\text{new}})$ coincide 
       (modulo renaming $w$ to $ww'$) and
  \item the transfer sequences of $(\Amc_{\Sigma},w)$ w.r.t.\ $\Omc$ coincides with the transfer sequence of 
       $(\Amc_{\Sigma},ww')$ w.r.t.\ $\Omc$ (modulo renaming $w$ to $ww'$)
\end{itemize} 
then one can replace ${\Amc_{\Sigma}}_{w}^{\downarrow}$ by ${\Amc_{\Sigma}}_{ww'}^{\downarrow}$ in 
$\Amc_{\Sigma}$ and it still holds that $\Omc,\Amc'\models B(x_{A})$
for the resulting ABox $\Amc'$. If in addition we require that 
$\text{tail}(w)=\text{tail}(ww')$, then the 
resulting ABox is (isomorphic to) a prefix closed sub ABox of $\Amc_{\Sigma}$ and so
the concept corresponding to the ABox $\Amc'$ is still entailed by $A$ w.r.t.~$\Omc$.

By performing the above replacement exhaustively, we obtain a prefix closed subset $\Amc$ of $\Amc_{\Sigma}$ that is of depth $\leq 2^{q(||\Omc||)}$ with $q$ a polynomial and therefore has the properties required for Point~3 of Theorem~\ref{thm:critele2}. Such an $\Amc$ can be constructed in at most double exponential time since one can construct the canonical model $\Imc_{\Omc,A}$ up to nodes of depth $\leq 2^{q(||\Omc||)}$ in double exponential time.   

The claims stated in Theorem~\ref{thm:critele2} for interpolants without the universal 
role are shown by modifying the proof above in a straightforward way. 

\section{Proofs for Sections~\ref{sec:final} and~\ref{sec:discussion}}
We first complete the proof of Theorem~\ref{thm:horn} by showing that there is a Horn-$\mathcal{ALCI}$-simulation between the interpretations $\Imc$ and $\Imc'$ defined in Figure~\ref{fig:hornsim}.
The definition of Horn-simulations is as follows.
For any two sets $X$ and $Y$ and a binary relation $R$, we set 
\begin{itemize}
	\item $XR^{\uparrow}Y$ if for all $x\in X$ there exists $y\in Y$ with $(x,y)\in R$;
	\item $XR^{\downarrow}Y$ if for all $y\in Y$ there exists $x\in X$ with  $(x,y)\in R$.
\end{itemize}
A relation $Z \subseteq
\mathcal{P}(\Delta^{\Imc}) \times \Delta^{\Imc'}$ is a
\emph{Horn-$\mathcal{ALCI}(\Sigma)$-simulation} between $\Imc$ and $\Imc'$  if $(X,b) \in Z$ implies  $X\not=\emptyset$
and the following hold:
\begin{itemize}
	\item for any $A\in \Sigma$, if $(X,b) \in Z$ and $X\subseteq A^{\Imc}$, then $b\in A^{\Imc'}$;
	
	\item for any role $r$ in $\Sigma$, if $(X,b) \in Z$ and $X r^{\Imc\uparrow} Y$, then there exist $Y' \subseteq Y$ and $b'\in \Delta^{\Imc'}$ with $(b,b')\in r^{\Imc'}$ and $(Y',b') \in Z$;
	
	\item for any role $r$ in $\Sigma$, if $(X,b) \in Z$ and $(b,b')\in r^{\Imc'}$,
	then there is $Y \subseteq \Delta^{\Imc}$ with $X r^{\Imc\downarrow} Y$ and $(Y,b') \in Z$;
	
	\item if $(X,b) \in Z$, then $\Imc',b\preceq_{\ELI,\Sigma}\Imc,a$ for every $a\in X$ 
	(where $\preceq_{\ELI,\Sigma}$ indicates that we have a simulation that does not only respect role names in 
	$\Sigma$ but also the inverse of role names in $\Sigma$).
\end{itemize}
We write $\Imc,X \preceq_{\textit{horn},\Sigma}\Imc',b$ if there exists a Horn-$\mathcal{ALCI}(\Sigma)$-simulation $Z$ between $\Imc$ and $\Imc'$ such that $(X,b) \in Z$. It is shown in \cite{DBLP:conf/lics/JungPWZ19} that if $\Imc,X \preceq_{\textit{horn},\Sigma}\Imc',b$, then all Horn-$\mathcal{ALCI}(\Sigma)$-concepts true in all nodes in $X$ are also true in $b$. 

Now observe that the relation $Z$ between $2^{\Delta^{\Imc}}$ and $\Delta^{\Imc'}$ containing all pairs $(\{x\},x')$, $(\{b,c\},b'')$, and $(\{d,e\},d'')$ is a
Horn-$\mathcal{ALCI}(\Sigma)$-simulation between the interpretations $\Imc$ and $\Imc'$
defined in Figure~\ref{fig:hornsim}, as required.

We next observe that moving to the Horn fragment Horn-GF of the guarded fragment is not sufficient to obtain a logic in which interpolants/explicit definitions always exist.
To this end we modify the ontology given in the proof of Theorem~\ref{thm:horn}. 
In detail, let $\Omc'$ contain the following CIs:
	\begin{align*}
		A & \sqsubseteq B \\
		B & \sqsubseteq \forall r.F\\
		F & \sqsubseteq \exists r_{1}.D_{1} \sqcap \exists r_{2}.D_{2} \sqcap \exists r_{1}.M \sqcap \exists r_{2}.M \\
		A & \sqsubseteq \forall r.((F \sqcap \exists r_{1}.(D_{1}\sqcap M) \sqcap \exists r_{2}.(D_{2}\sqcap M)) \rightarrow E) \\
		B & \sqsubseteq \exists r. C \\
		C & \sqsubseteq F \sqcap \forall r_{1}.D_{1} \sqcap \forall r_{2}.D_{2} \	
	\end{align*}
	and also $B \sqcap \exists r. (C \sqcap E)  \sqsubseteq A$. Define the signature $\Sigma$ by setting $\Sigma = \{B,D_{1},D_{2},E,r,r_{1},r_{2}\}$. We note that, intuitively, the third and fourth CI should be read as
	\begin{align*}
		F & \sqsubseteq \exists r_{1}.D_{1} \sqcap \exists r_{2}.D_{2} \\
		A & \sqsubseteq \forall r.((F \sqcap \forall r_{1}.D_{1} \sqcap \forall  r_{2}.D_{2}) \rightarrow E) 
	\end{align*}	
	and the concept name $M$ is introduced to achieve this in a projective way as the latter CI is not in Horn-$\mathcal{ALCI}$.
	
	We first observe that $A$ is implicitly definable from $\Sigma$ under $\Omc'$ since 
	$$
	\Omc' \models A \equiv B \sqcap \forall r. (\forall r_{1}.D_{1} \sqcap \forall r_{2}.D_{2} \rightarrow E).
	$$
	We next sketch the proof that $A$ is not explicitly Horn-GF$(\Sigma)$-definable under $\Omc'$. For a definition of Horn-GF and Horn-GF simulations we refer the reader to 
	\cite{DBLP:conf/lics/JungPWZ19}. Now consider the interpretations $\Imc$ and $\Imc'$ defined in Figure~\ref{figure:figfinal}. Both $\Imc$ and $\Imc'$ are models of $\Omc'$, $a\in A^{\Imc}$, $a'\not\in A^{\Imc'}$, but $a\in F^{\Imc}$ implies $a'\in F^{\Imc'}$ holds for every Horn-GF$(\Sigma)$-formula $F$, and the claim follows. The latter can be proved by observing that there exists a Horn-GF$(\Sigma)$-simulation between $\Imc$ and $\Imc'$ \cite{DBLP:conf/lics/JungPWZ19} containing $(\{a\},a)$.  In fact, one can show that the relation $Z$ containing all pairs $(\{x\},x')$, $(\{b,c\},b'')$, and $(\{d,e\},d'')$ is a Horn-GF$(\Sigma)$-simulation.
	
	\begin{figure}
          		\begin{tikzpicture}[auto,->,
                  every label/.style={font=\small,inner sep=0pt},
                  el/.style={},
			pl/.style={font=\small,inner sep=2pt},
			r1/.style={},r2/.style={},yscale=0.85]
			
			\node[el,label=above:{$A,B$}] (1) {$a$} ;
			\node[el,label={[yshift=0.1cm]below right:{$C,E,F$}}] at ($(1)-(1,1.5)$) (2) {$b$} ;
			\path (1) edge[above left] node[pl] {$r$} (2) ;
			\node[el,label={[yshift=0.1cm]below right:{$F$}}] at ($(1)+(1,-1.5)$) (3) {$c$} ;
			\path (1) edge node[pos=0.25,pl] {$r$} (3) ;
			
			\node[el,label=below:{$D_1,M$}] at ($(2)-(0,1.5)$) (4) {$d$} ;
			\path[r1] (2) edge[left] node[pl] {$r_1$} (4) ;
			
			\node[el,label=below:{$D_1$}] at ($(3)+(-0.5,-1.5)$) (5) {$e$} ;
			\path[r1] (3) edge[above left] node[pl,pos=0.7] {$r_1$} (5) ;
			\node[el,label=below:{$M$}] at ($(3)+(0.5,-1.5)$) (6) {$f$} ;
			\path[r1] (3) edge node[pl,pos=0.7] {$r_1$} (6) ;
			
			\node[el,label=above:{$D_2,M$}] at ($(1)+(1.8,0)$) (7) {$g$} ;
			\path[r2] (2) edge[above left] node[pos=0.8,pl] {$r_2$} (7) ;
			\path[r2] (3) edge[below right] node[pl] {$r_2$} (7) ;

			\node[el,label=above:{$B$}] at ($(1)+(4.2,0)$) (1') {$a'$} ;
			\node[el,label={[yshift=0.1cm]below right:{$E,F$}}] at ($(1')-(1.5,1.5)$) (2') {$b'$} ;
			\path (1') edge[above left] node[pl] {$r$} (2') ;
			\node[el,label={[yshift=0.1cm]below right:{$C,F$}}] at ($(1')-(0,1.5)$) (25') {$b''$} ;
			\path (1') edge[above left] node[pl] {$r$} (25') ;
			\node[el,label={[yshift=0.1cm]below right:{$F$}}] at ($(1')+(1.5,-1.5)$) (3') {$c'$} ;
			\path (1') edge node[pos=0.25,pl] {$r$} (3') ;
			
			\node[el,label=below:{$D_1,M$}] at ($(2')-(0,1.5)$) (4') {$d'$} ;
			\path[r1] (2') edge[left] node[pl] {$r_1$} (4') ;
			
			\node[el,label=below:{$D_1,M$}] at ($(25')-(0,1.5)$) (45') {$d''$} ;
			\path[r1] (25') edge[left] node[pl] {$r_1$} (45') ;
			
			\node[el,label=below:{$D_1$}] at ($(3')+(-0.5,-1.5)$) (5') {$e'$} ;
			\path[r1] (3') edge[above left] node[pl,pos=0.7] {$r_1$} (5') ;
			\node[el,label=below:{$M$}] at ($(3')+(0.5,-1.5)$) (6') {$f'$} ;
			\path[r1] (3') edge node[pl,pos=0.75] {$r_1$} (6') ;
			
			\node[el,label=above:{$D_2,M$}] at ($(1')+(2.2,0)$) (7') {$g'$} ;
			\path[r2] (2') edge[above left] node[pos=0.8,pl] {$r_2$} (7') ;
			\path[r2] (25') edge[below right] node[pos=0.55,pl] {$r_2$} (7') ;
			\path[r2] (3') edge[below right] node[pl] {$r_2$} (7') ;
			
		\end{tikzpicture}
	\caption{Interpretations $\Imc$ (left) and $\Imc'$ (right) used for $\Omc'$.}
\label{figure:figfinal}

\end{figure}

We finally make a few observations regarding the Horn fragment of first-order logic. 
Recall that \emph{Horn-FO} is defined as the closure
of formulas of the form 
$R(\vec{t})$,
$$ 
R_{1}(\vec{t}_{1})\wedge \cdots \wedge R_{n}(\vec{t}_{n}) \rightarrow R(\vec{t}), \quad
R_{1}(\vec{t}_{1})\wedge \cdots \wedge R_{n}(\vec{t}_{n}) \rightarrow \bot
$$
under conjunction, universal quantification, and existential quantification, where 
$\vec{t}_{1},\ldots,\vec{t}_{n},\vec{t}$ are sequences of individual variables and individual names \cite{modeltheory}. According to Exercise 6.2.6 in~\cite{modeltheory} Horn-FO has the following property.
\begin{theorem} Let $\varphi,\psi$ be sentences in Horn-FO such that $\varphi \wedge \psi$ is not satisfiable. 
	Then there exists a sentence $\chi$ in Horn-FO such that $\text{sig}(\chi) \subseteq \text{sig}(\varphi)\cap \text{sig}(\psi)$, $\varphi \models \chi$, and $\chi \wedge \psi$ is not satisfiable.
\end{theorem}
We directly obtain the following interpolation result.
\begin{theorem} Let $\Omc_{1},\Omc_{2}$ be Horn-$\mathcal{ALCIO}_{u}$-ontologies and let $C_{1}, C_{2}$ be 
	Horn-$\mathcal{ALCIO}_{u}$-concepts such that $\Omc_{1}\cup \Omc_{2}\models C_{1}\sqsubseteq C_{2}$. Then there exists
	a formula $\chi(x)$ in Horn-FO such that 
	\begin{itemize}
		\item $\text{sig}(\chi) \subseteq \text{sig}(\Omc_{1},C_{1})\cap \text{sig}(\Omc_{2},C_{2})$;
		\item $\Omc_{1} \models \forall x (C_{1}(x) \rightarrow \chi(x))$;
		\item $\Omc_{2} \models \forall x (\chi(x) \rightarrow C_{2}(x))$.
	\end{itemize}	
\end{theorem}	
\begin{proof}
	Take a fresh unary relation symbol $A(x)$ and a fresh individual name $c$. Let $\varphi$ be the conjunction of all sentences in $\Omc_{1}\cup \{C_{1}(c)\}$ and let $\psi$ be the conjunction of all sentences in $\Omc_{2} \cup \{\forall x (C_{2}(x) \leftrightarrow A(x)), \neg A(c)\}$.
	Then $\varphi$ and $\psi$ are both equivalent to sentences in Horn-FO.
	By definition $\varphi \wedge \psi$ is not satisfiable. Thus there exists a Horn-FO sentence $\chi$ using only $c$ and symbols in  $ \text{sig}(\Omc_{1},C_{1})\cap \text{sig}(\Omc_{2},C_{2})$
	such that $\varphi \models \chi$ and $\chi \wedge \psi$ is not satisfiable. Thus:
	\begin{itemize}
		\item $\Omc_{1} \models C_{1}(c) \rightarrow \chi$;
		\item $\Omc_{2} \cup \{\forall x (C_{2}(x) \leftrightarrow A(x))\}\models \chi \rightarrow A(c)$.
	\end{itemize}	
	Replace $c$ by $x$ in $\chi, C_{1}(c)$, and $A(c)$. Then
	\begin{itemize}
		\item $\Omc_{1} \models \forall x (C_{1}(x) \rightarrow \chi(x))$;
		\item $\Omc_{2} \models \forall x (\chi(x) \rightarrow C_{2}(x))$,
	\end{itemize}	
	as required.
\end{proof}
Applied to Horn-$\mathcal{ALCI}$ ontologies and concepts we thus always obtain an interpolant in Horn-FO and an interpolant in $\mathcal{ALCI}$ (since $\mathcal{ALCI}$ enjoys the CIP~\cite{TenEtAl13}). 

It would be interesting to find out whether there exists an interpolant in the intersection of Horn-FO and $\mathcal{ALCI}$ and whether it is possible to give an informative syntactic description of that intersection.
\end{document}